\documentclass[11pt]{article}
\def\shownotes{0}   

\usepackage[margin=1in]{geometry}

\usepackage[multiple]{footmisc}

\usepackage{amsmath}
\usepackage{amssymb}
\usepackage{multirow}
\usepackage{color}
\usepackage{graphicx}
\usepackage{subfigure}
\usepackage{xspace}
\usepackage{algorithm}
\usepackage[noend]{algorithmic}

\newtheorem{theorem}{Theorem}
\newtheorem{lemma}{Lemma}
\newtheorem{corollary}{Corollary}

\newtheorem{claim}{Claim}

\newtheorem{definition}{Definition}

\newcommand{\qedclaim}{\hfill $\diamond$ \medskip}
\newenvironment{proofclaim}{\noindent{\em Proof.}}{\qedclaim}

\newcommand{\qed}{\hfill $\Box$ \medskip}
\newenvironment{proof}{\noindent{\em Proof.}}{\qed}

\usepackage[english]{babel}
\usepackage{blindtext}

\usepackage{etoolbox}
\makeatletter
\patchcmd{\maketitle}{\@copyrightspace}{}{}{}
\makeatother

\usepackage{amsmath,amssymb,bbm}

\newcommand{\midwor}[1]{\;\textnormal{ #1 }\;} 
\newcommand{\vf}{\,,\,} 
\newcommand{\ff}{\,.} 
\newcommand{\lset}{\left\{\left.\;}

\newcommand{\dimset}{\right.\;\left|\;}
\newcommand{\rset}{\;\right.\right\}}

\newenvironment{disarray}
 {\everymath{\displaystyle\everymath{}}\array}
 {\endarray}

\newcommand{\RelInt}{\mathbb{Z}} 
\newcommand{\NatInt}{\mathbb{N}} 

\newcommand{\ind}[1]{\mathbb{I}_{\{#1\}}} 

\newcommand{\setW}{\mathcal{W}}

\ifnum\shownotes=1
\newcommand{\authnote}[2]{{ $\ll${\footnotesize\sffamily Comment from #1: #2}$\gg$}}
\else
\newcommand{\authnote}[2]{}
\fi

\newcommand{\ie}{{\em i.e.,~}}
\newcommand{\eg}{{\em e.g.,~}}

\makeatletter
\def\timenow{\@tempcnta\time
  \@tempcntb\@tempcnta
  \divide\@tempcntb60
  \ifnum10>\@tempcntb0\fi\number\@tempcntb
  \multiply\@tempcntb60
  \advance\@tempcnta-\@tempcntb
  :\ifnum10>\@tempcnta0\fi\number\@tempcnta}
\makeatother

\begin{document}


\title{Can Selfish Groups be Self-Enforcing?}

\author{
Guillaume Ducoffe\footnote{
Columbia University \& 
\'Ecole Normale Sup\'erieure de Cachan,
\texttt{gducoffe@ens-cachan.fr}
}
, 
Dorian Mazauric\footnote{
Columbia University \& Laboratoire d'Informatique Fondamentale Marseille, 
\texttt{dorian.mazauric@lif.univ-mrs.fr}
}
, 
Augustin Chaintreau\footnote{
Columbia University, 
\texttt{augustin@cs.columbia.edu}
}
}

\ifnum\shownotes=1
\date{
\footnotesize (
\emph{Timestamp}:~\today, \timenow;
\emph{Length}:~\pageref{totalpages} pages, excluding title, biblio, appendices
)
} 
\else
\date{}
\fi

\maketitle
\thispagestyle{empty}

\begin{abstract}
Algorithmic graph theory has thoroughly analyzed how, given a network describing constraints between various nodes, groups can be formed among these so that the resulting configuration optimizes a \emph{global} metric. In contrast, for various social and economic networks, groups are formed \emph{de facto} by the choices of selfish players. A fundamental problem in this setting is the existence and convergence to a \emph{self-enforcing} configuration: assignment of players into groups such that no subset of players have an incentive to move into another group than theirs. Motivated by information sharing on social networks -- and the difficult tradeoff between its benefits and the associated privacy risk -- we study the possible emergence of such stable configurations in a general selfish group formation game.

Our paper defines a general class of \emph{cooperative coloring game} that includes previous models of coloring and social network analysis.
While non trivial lower bounds were never established for any of these particular cases, using a novel combinatorial technique we provide the exact convergence time when up to two players collude. We also \emph{complete} the analysis of polynomial time convergence for cooperative coloring games \emph{in general}, closing open problems from previous works. 
We show in particular that when a general potential function argument fails to prove polynomial convergence, then games always exist that either do not converge or take more than a polynomial number of steps. 
Moreover, all our results are shown to be robust to variants of the games, including overlaps between groups and more complex utility functions representing \emph{multi-modal} interactions. We finally prove that a small level of collusion, although it makes convergence more difficult, has a significant and \emph{positive} effect on the \emph{efficiency} of the equilibrium that is attained. 

\end{abstract}


\newpage
\pagestyle{plain}

\section{Introduction}

Understanding how self-interested agents decide to form groups is a fundamental computational problem which regained importance as multiple information sharing platforms are widely adopted. In addition to form direct connections to each others, users of online social networks such as Facebook or Google+ and microblogging platforms such as Twitter and Weibo are encouraged to participate to groups in which information is exchanged. Participating to one of these online groups comes with the benefit of receiving information that could be an advantage, as it was previously well documented~\cite{Coleman:1957vq,Granovetter:1974te,Burt:2004tv,Montgomery:1991uj,CalvoArmengol:2004ug}, but it also comes with the expectation to disclose some personal information to its members. At times, sharing information with certain individuals can be harmful~\cite{Wang:2011gb,Acquisti:2004gp}, and the increasing concern of users for privacy stresses the need to be protected from such risk. In addition, one's time and attention span being by nature limited, there is only so many groups an individual can meaningfully participate to. This paper is about an idealized but general model reproducing this tradeoff once summarized by the famous essayist:
\begin{quote}
``Nul plaisir n'a goust pour moi sans communication, \\ 
mais il vaut mieux encore estre seul qu'en compagnie 
ennuyeuse et inepte.''\footnote{``No pleasure keeps its savor unless I tell someone, but one'd rather be alone than in dull company.''}

\hfill 
--
 Montaigne, in \textit{Essais} III-9
\end{quote}

In the simplest form of this model, a user may share information with others only by choosing to belong to a \emph{single} group. If users are represented as vertices of a graph, the results of their choices is hence a partition, or a vertex coloring, of the underlying graph. The general model defined below, which we call a \emph{coloring game}, extends previous related work and lies at the frontier between the behavioral analysis of social networks~\cite{kleinberg2013information,Kearns:2006wj,Chaudhuri:2008ts}, the theory of coalition formation and hedonic games~\cite{Dreze:1980tv,Bloch:2011ta}, and its more recent analysis through computational complexity~\cite{Panagopoulou:2008dr,Chatzigiannakis:2010vo,Escoffier:2012vu}. 

A coloring game takes as input a fixed weighted graph whose edges denote potential connections, and the weight on an edge indicates how advantageous this connection is to both nodes. It is  \emph{positive} if this connection is desirable to them. An edge with a \emph{negative} weight represents a ''conflict'', \emph{a fortiori} if the weight is infinitely large and we then call the nodes \emph{enemies}. 
When a player picks a color, she receives a utility determined by the weights of her adjacent edges, but only those connecting her to players with the same color\footnote{Note that our game allows a player to choose a color no-one else is currently using, and receives zero utility.}. We also allow a subset of players to act collectively, \ie simultaneously changing their groups, to form a \emph{deviation}. A subset of size $k$ can form a \emph{$k$-deviation} if all players of the subset can simultaneously change their colors so that each of them strictly increases her utility (see Figure 1). 
\begin{figure}[htb]
	\begin{minipage}{0.55\linewidth}
\centerline{\includegraphics[width=1\linewidth]{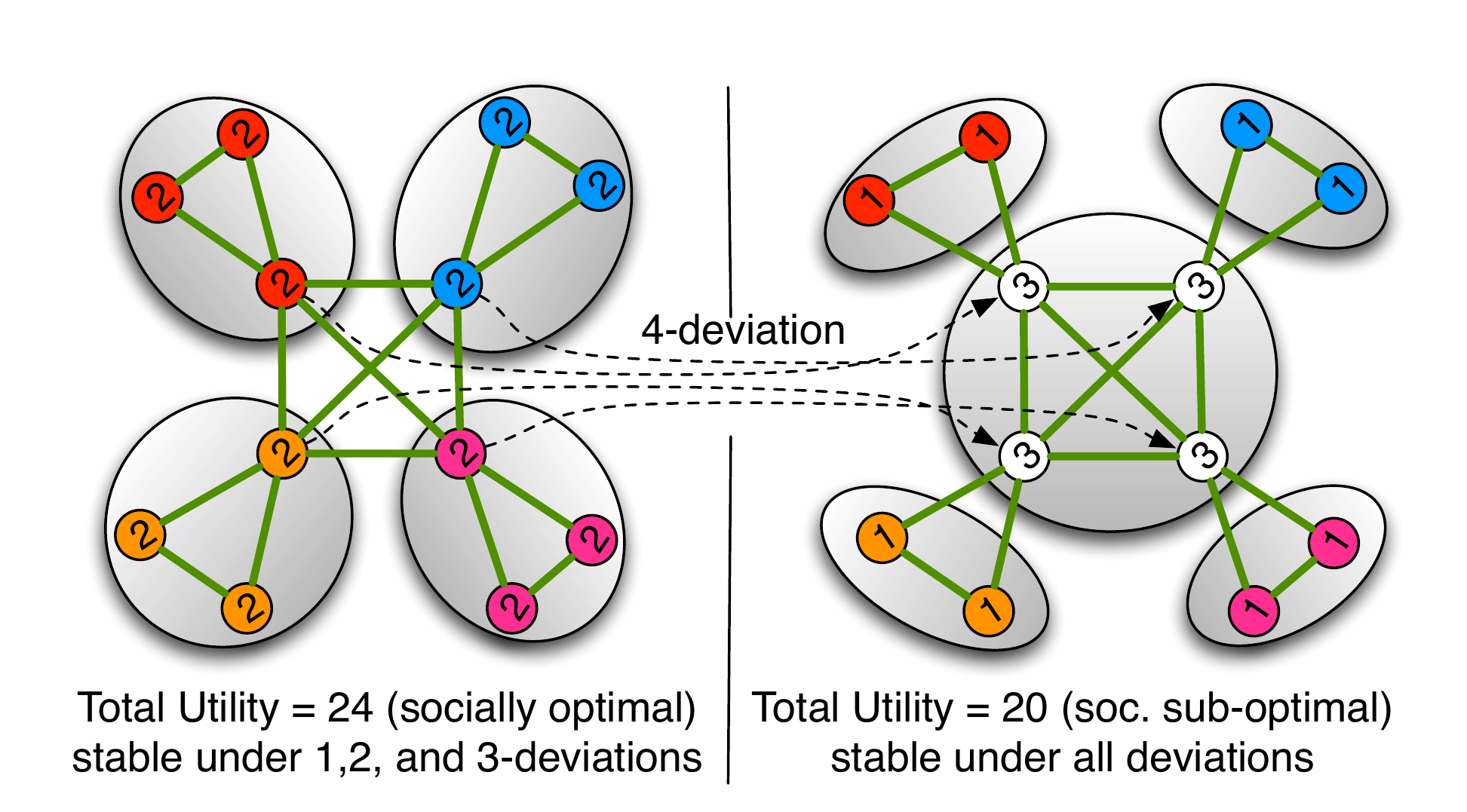}}
	\end{minipage}
	\begin{minipage}{0.45\linewidth}
		\caption{\textbf{A coloring game.}
		All solid edges have weight 1 and all pairs of nodes not connected by a solid edge are \emph{enemies} (connected with an edge of weight $-\infty$). Labels denote each node's utility. (Left) A partition of the graph into 4 groups, nodes all receive utility $2$. Note that the middle nodes, which we may call the ``connectors'', also have friends in other groups. (Right) Connectors can form a $4$-deviation and join a larger group, building a new partition where each of them receives utility $3$, while other nodes' utility decreases.
\label{fig:example}
}
	\end{minipage}
\end{figure}

The objective of this paper is to determine under which conditions a coloring game becomes \emph{stable} after a \emph{small} number of steps, and reaches an \emph{efficient} configuration. More precisely, for a fixed value of $k$ we call a partition \emph{$k$-stable} if no deviation of size at most $k$ exists\footnote{This is a $k$-strong Nash equilibrium of the game, called $k$-stable to align with \cite{kleinberg2013information}.}. When such $k$-stable partition exists, we wish to determine whether the game always finds one after a number of steps that grows at most \emph{polynomially} in the number of players. We also aims at comparing the \emph{total utility} received by players to the best partition where this quantity is maximized, as well as characterizing multiple variants of this game, including those where a user may join multiple groups.

\subsection{Related works and our contributions}

\emph{Vertex coloring} is a central optimization problem in Computer Science, with much of the work centering on minimizing globally the number of colors used. Without a complexity classes collapse, no polynomial time algorithm can approximate it. In a \emph{coloring game}, a solution to that problem is found through a selfish local search heuristic. Analyzing such games have proved in the last decade especially fruitful for three reasons: (1) they improve \emph{and} unify a series of well known polynomial heuristics~\cite{Panagopoulou:2008dr}, (2) they naturally extend to design sound and efficient distributed algorithms~\cite{Chatzigiannakis:2010vo} and (3) they offer insights into how computational problems are solved in practice by human subjects~\cite{Kearns:2006wj}. Games in the works aforementioned are equivalent to the \emph{uniform case} that we define and analyze first: In this case, all edges that do not denote enemies have the same weight. These previous works consider the case where players move alone (\ie forming only $1$-deviation). 

\textbf{Our contribution:} 
All these games are known to converge in polynomial time to a Nash Equilibrium (\ie a $1$-stable partition) using a potential argument~\cite{Panagopoulou:2008dr}, but the exponent as the number of players grows is not known to be tight. In fact, a non trivial lower bound was never established. An entire new analysis allows us to provide the \emph{exact} convergence time in closed form. This result prove in particular that the convergence is always arbitrarily \emph{faster} than the bound of the potential function can predict~(Section~\ref{ssec:k=12}).

\medskip
A renewed interest for coloring games recently emerged from a \emph{cooperative} approach. Extending game theory beyond the strict definition of a Nash Equilibrium, the analysis of cooperative games studies how multiple players can form self-interested alliances. \emph{Cooperative} coloring games involve deviations of larger size (up to size $k$), and a stronger termination condition stipulates that no deviation of size at most $k$ is profitable to its members. The attained configuration of the game can then show a considerable improvement\footnote{\eg \cite{Escoffier:2012vu} established that for $k=\infty$, the number of colors used when the game ends approximates the chromatic number up to a logarithmic factor, but finding such a configuration is NP hard.} but this comes at the price of a more complicated local search heuristic, which makes it challenging in general to prove a polynomial convergence time. All polynomial upper bounds to date rely on a potential argument which only applies when $k=2$ or $k=3$~\cite{kleinberg2013information}. Prior to us, non trivial lower bounds were never established. Finding a polynomial bound when $k$ is fixed and larger than $3$ is arguably the most important open problem of cooperative coloring games. It has been shown to be impossible to establish using a potential function~\cite{kleinberg2013information,Escoffier:2012vu}.  This again questions the power of the potential method in the first place.

\textbf{Our contribution:}
Here, we solve this open problem as we prove that these games with $k=4$ necessitates more than a polynomial number of steps to terminate. This also completes the characterization of a polynomial time convergence for cooperative coloring games in the uniform case. It reaffirms the potential method as a ``good'' estimate of time complexity. The proof, which closes an open problem from \cite{kleinberg2013information} and contradicts a previous conjecture from \cite{Escoffier:2012vu}, uses a surprising argument. The extreme pathologies that are responsible to delay the convergence are not derived as before from properties of independent sets. They appear even in the simplest graph from a series of bad ``scheduling'' decisions on subsets, which can be recursively nested to contradict any polynomial bound. This result motivates a widely open problem: can a ``smart scheduler'' be designed in a cooperative setting to avoid such pathological sequences, or is finding a 4-stable configuration a PLS complete problem?~(Section~\ref{subsec-superpolynomial-time}).

\medskip
Moving away from the uniform case, we next analyze coloring games with \emph{general weights}, in which edges that do not denote enemies may have different values. These games present new difficulties. Primarily, a $k$-stable partition is not guaranteed to exist (see Figure~\ref{fig:nonstable}), showing that the game may not always terminate.
Prior to our work, no result exists that guarantees the termination of the dynamics\footnote{The only result known \cite{kleinberg2013information} is a negative one for a variant of the game in which two nodes of different colors can ``gossip'', a deviation that results in connecting all other nodes in their respective groups into a single group, without the latter being able to prevent it. We analyze first the case where such deviations are ignored, which is more complex. Our results extends when gossip is allowed (Section~\ref{sec:extensions}).\label{fnot:gossip}}. A Nash Equilibrium always exists, although finding one in general is PLS complete for unrestricted weights. This motivates us to analyze the case where weights are chosen in a fixed subset, for which the game converges to a Nash Equilibrium in polynomial time. For the important case where all \emph{positive} weights are equal, while null and negative weights may take all possible values, we prove by another potential function argument that a $2$-stable partition is found in polynomial time. As a side note, this proof also yields sufficient conditions on the structure of the graph that guarantees that stability of higher order can be attained. The key open question is, again, is this result tight, or can games be shown to converge even when the potential function method fails?

\textbf{Our contribution:}
Our results close this question and also offer the first non trivial impossibility results on general cooperative coloring games. Indeed, the potential function method is ``as good as it gets'' for three complementary reasons: \emph{First}, as soon as $3$-deviations are allowed, and even when positive weights are all equal, we provide a constructive argument proving that not all games terminate. \emph{Second,} for all other sets of possible weights that are neither trivial nor trivially equivalent to the uniform case or this one, we show games exist that are not $2$-stable. \emph{Third}, we prove that whenever the conditions for a potential argument are not satisfied, not only is a game not guaranteed to converge, but deciding whether it does is NP-hard, even \emph{after} the set of weight has been preliminary fixed. Establishing this result requires us to build a general reduction between convergence from \emph{any} graph and a maximum independent set using only the fact that a single counterexample exists. We found many of these results surprising, especially as in most of those cases our proofs are based on 2-deviations that by definition do not diminish the total utility. ~(Section~\ref{sec:weighted}).
\begin{figure}[t]
	\begin{minipage}{0.55\linewidth}
\centerline{\includegraphics[width=0.8\linewidth]{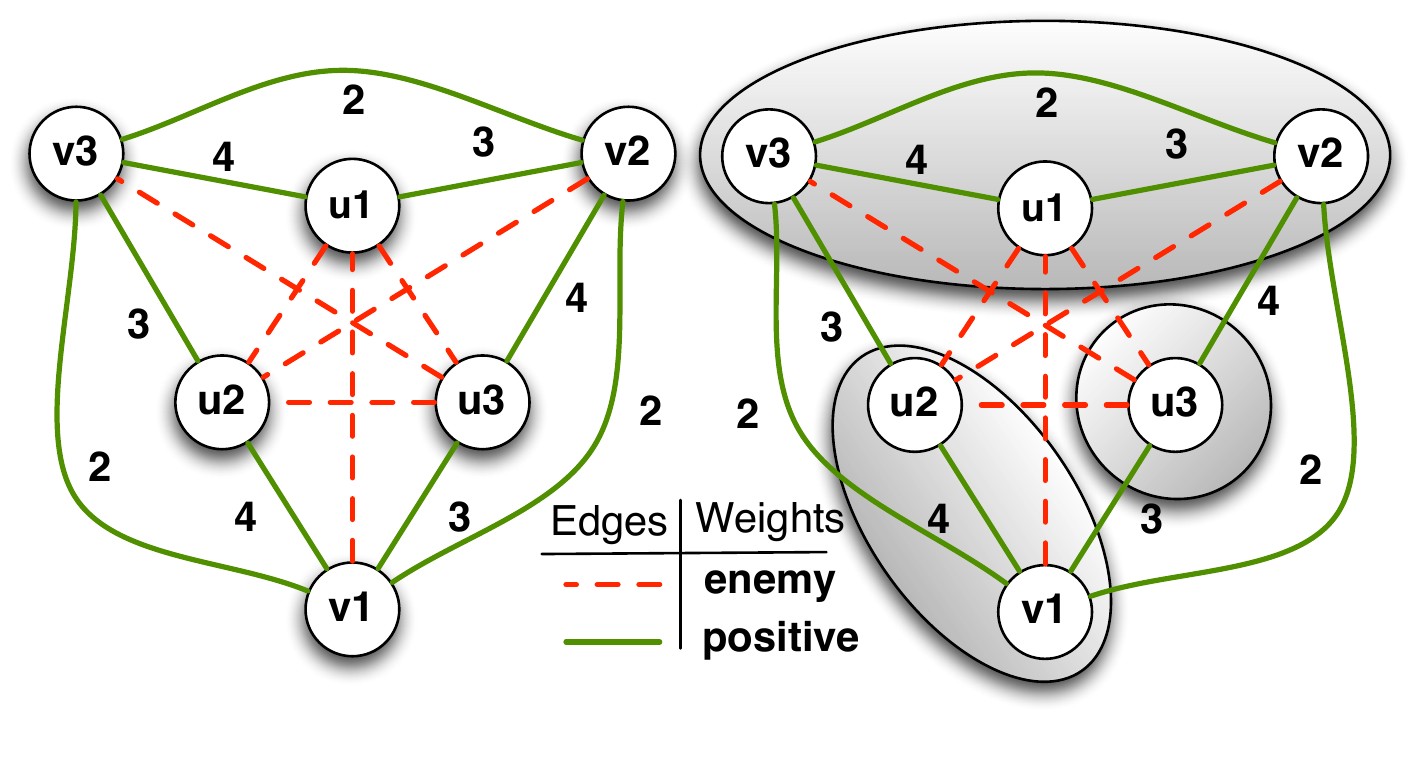}}
	\end{minipage}
	\begin{minipage}{0.45\linewidth}
		\caption{\textbf{A coloring games for which no $2$-stable partition exists.} (Left) A weighted coloring game: Solid green edges have a positive weight with the indicated value, dashed red edges denote enemies or, equivalently, edges of weight $-\infty$. (Right) No partition is $2$-stable: e.g., starting from these groups, nodes \textsf{v1} and \textsf{v2} both decide to move into the group containing \textsf{u3}, effectively ``rotating'' the groups.
\label{fig:nonstable}
}
	\end{minipage}
\end{figure}

\medskip
Together our results prove that coloring games converge in polynomial time only when a polynomial potential function exists. This has implications on the stability and efficiency of information sharing and group formation in general, motivating us to see if this claim is robust to slight changes of our assumptions. 

\textbf{Our contribution:}
We treat first the model from \cite{kleinberg2013information} that closely resembles ours except that it introduces another specific ``gossip'' deviation (see footnote~\ref{fnot:gossip} above). While this model is equivalent in the uniform case in which our results close previously open problems, we show that, in the weighted case, gossip introduces additional pathologies even with a unique positive weight and $k=2$. Similarly, we prove that asymmetry in the assigned weights creates more instability: It is then even NP-hard to decide if a Nash Equilibrium exists. 
~(Section~\ref{sec:gossip}-\ref{sec:asym}).

More interestingly, we extend our results to a ``multi-channel'' model in which users may participate to more than one group, which allows us to study the formation of groups beyond the partition case. This game resembles one in which Nash Equilibrium were empirically shown to reflect the formation of overlapping communities in social networks~\cite{Chen:2010wz}. We first show that the proof of a polynomial potential function with general weights extend from a partition to the case of overlapping groups. However, to our surprise, outside of this case, allowing group overlap can jeopardize stability! Even when focusing on \emph{the uniform case}, for which a $k$-stable \emph{partition} always exists, we show that not all graphs admit $3$-stable configurations when nodes can belong to two groups instead of a single one. 
~(Section~\ref{sec:multichannel}).

We then consider a model where the weights measuring the benefits and costs incurring by nodes when sharing a group are not only defined on pairs, but more generally on subsets. We extend the potential function argument of the weighted case to prove a sufficient condition on the termination of the game that depends on the property of the associated hypergraph.
~(Section~\ref{sec:multimodal}).

\medskip
Finally, we investigate the \emph{efficiency} of the equilibrium attained by the game as different size of deviations are allowed. Previous works measured overall efficiency in terms of the number of groups~\cite{Escoffier:2012vu}, or total utility produced~\cite{kleinberg2013information}. We focus on the later as we wish to study the outcomes of the games as experienced by the users, as opposed to minimizing colors used in a graph. Previous results focused either the ratio of the best $k$-stable configuration to the optimal (\ie price of $k$-stability), while we wish to study the worst case of a final configuration of the game when it exists (\ie~price of $k$-anarchy). It was only studied for large values of $k$, for which our results prove that convergence can be uncertain or long.

\textbf{Our contribution:}
Focusing on small values of $k$ for which convergence is guaranteed, we exhibit an interesting tension between stability and efficiency. The game using only 1-deviations can be converging to a stable case far from the optimal. However, we prove that allowing slightly larger deviations, while making convergence longer or more difficult, already significantly improves the efficiency of the groups that are formed. This provides another example in which a stronger form of stability is shown to positively impact the price of anarchy, as other recent works established in different cases~\cite{bachrach2013strong}.
An important theoretical consequence is to motivate finding conditions on graphs under which stability and efficiency can be reconciled. (Section~\ref{sec:efficiency}).

\medskip
\textbf{Additional related work:}
Beyond works of immediate relevance for this model of coloring games, coalitions formation have been studied in multiple settings including \emph{hedonic games} where agents form groups in order to enjoy each other's company \cite{Dreze:1980tv}. Indeed, a large number of coalition formation models, including some for overlapping coalitions, have been proposed in cooperative game theory where transfers are allowed between the players (see~\cite{Dreze:1980tv,Zick:2012ui,Chalkiadakis:2010ui,Augustine:2011wj}). In contrast player's utility in our model are \emph{non-transferable}, which leads to a very different analysis. This assumption resembles the \emph{cut} and \emph{party affiliation} games~\cite{Balcan:2009uc,Fabrikant:2004fxa,Christodoulou:2006cr} except that in those, the number of groups is fixed, and usually only Nash equilibria are considered. 

Hedonic games where players cannot make side payments to each other usually model users's preferences as a partial ordering on configurations, and focus on finding Nash Equilibrium or a configuration in the core~(\ie $k$-stable for $k=1$ or $k=\infty$). This problem is already known to be computationally more difficult, and the sufficient conditions for existence of a stable configurations~\cite{Bloch:2011ta} as well as the previous NP-hardness reduction~\cite{Ballester:2004iu} cannot be used to characterize these coloring games. Our results provide existence proofs and algorithm to cover $k$ stable configurations in a more specific case of users' preferences, but it also shows that this particular case is sufficiently rich to observe a sudden increase in the time and computation needed to converge, and that it is surprisingly sensitive to the exact value of $k$.


\section{Cooperative coloring games}
\label{sec:model}


\paragraph{The Social Network}
We model social group formation using an edge-weighted graph $G=(V,E,w)$, where $V$ denotes the set of $n$ users.
The weight $w_{uv}$ captures the utility both $u$ and $v$ receive when they choose the same color.
Without loss of generality, we can suppose that $G$ is a \emph{simple} graph (no loop exists, no multiple edges exist for a pair of nodes) and \emph{complete} (since missing edges can be added with a weight $0$).
We hence omit the set of edges $E$.
Note that we assume $w_{uv} = w_{vu}$ (see section \ref{sec:extensions} for the non-symmetrical case) and that\footnote{Occasionally, we include \emph{best friends} denoted with a weight $N$ such that $N\geq n\times |w_{uv}|$ for $w_{uv} \notin \{-\infty,N\}$.}
$w_{uv} \in \mathbb{Z} \cup \{ - \infty \}$.
We will also consider games for which the set of weights is constrained, we then denote it by $\mathcal{W}$ (\eg 
Figure~\ref{fig:example} is a graph for 
$\mathcal{W} = \{-\infty,1\}$, 
Figure~\ref{fig:nonstable} is a graph for 
$\mathcal{W} = \{-\infty,2,3,4\}$
).
Note that if the edge $\{u,v\}$ is a conflict ($w_{uv} = -\infty$), the users are called \emph{enemies} and will never choose the same color. 

We denote by $P = (X_1, X_2, \ldots, X_n)$ the partition with groups formed by users choosing the same color.
Note that some of these $n$ groups may be empty. Two examples of such partitions are shown in Figure~\ref{fig:example}.
For a partition $P$, the \emph{utility}
for the user $u$ is defined as $f_u(P) = \sum_{v \in X(u) \setminus \{u\}} w_{uv}$, where $X(u)$ denotes the sharing group in $P$ which $u$ belongs to.
The global utility, or social welfare, is $f(P) = \sum_{u \in V} f_u(P)$. Examples of utilities obtained are shown in Figure~\ref{fig:example}.

\paragraph{Dynamic of the game}
Each user in the network plays the game to maximize her individual utility. Initially, all users choose a different color, forming a collection of singletons.
Then we fix a constant parameter $k \geq 1$, and the coloring game starts.
At each round $i \geq 0$, we consider the current partition $P_i$ of the nodes.
Given a coalition $S$ with at most $k$ players, we say that $S$ is a \emph{$k$-deviation}, or $k$-set, when all the users in $S$ have an incentive to join the same\footnote{As noted in \cite{Escoffier:2012vu}, all deviations can be reduced to ones where everyone in the subset pick the same color.} (possibly empty) group in $P_i$, so that they \emph{all} increase their individual utility in the process.
When such a deviation exists, we may allow any \emph{one}\footnote{
Simultaneous parallel deviations can be handled with a loop avoiding rule \cite{Chatzigiannakis:2010vo} but we ignore them for simplicity.} of them to happen and, in so doing, we \emph{break} $P_i$, and we get a new partition $P_{i+1}$.
Note that $f_{u}(P_{i+1}) > f_{u}(P_{i})$ for all $u \in S$, and there exists a group $X$ of $P_{i+1}$ that contains all nodes of $S$.
Figure~\ref{fig:example} presents an update following a $4$-deviation.
The game stops if and only if there is no possible $k$-deviation.
We call a partition $P$ such that there is no $k$-deviation that breaks $P$ a \emph{$k$-stable} partition.
An algorithmic presentation of the game is given below.

\begin{algorithm}[h]
\textbf{Dynamic of a cooperative coloring game}
\vspace{1mm}
{\hrule height0.8pt depth0pt \kern0pt}
\vspace{.5mm}
\label{algorithm:dynamic}
	\textbf{Input:} A positive integer $k \geq 1$, 
	and a graph $G=(V,w)$.
\\\textbf{Output:} A partition $k$-stable for $G$.
\begin{algorithmic}[1]
\STATE Set $i = 0 $ and $P_{0}$ to be the partition composed of $n$ singletons groups.
\WHILE{there exists a $k$-deviation for $P_{i}$}
\STATE Set $i = i+1$ and let $P_{i}$ be the partition obtained after any $k$-deviation.
\ENDWHILE
\RETURN Partition $P_{i}$.
\end{algorithmic}
\end{algorithm}


\section{The uniform case: Are longest deviation sequences polynomial?}
\label{sec:unweighted}

A coloring game is said \emph{uniform} if, except for conflict edges, all edges have the same unit weight (\ie $\mathcal{W} = \{-\infty,1\}$). This game is entirely characterized by an unweighted and undirected \emph{conflict graph} $G^-=(V,E)$ that contains all the conflict edges (\ie those with weight $-\infty$).
The complementary graph of $G^{-}$ represents all the pairs of friends (with unit weight).
Note also that, for any user $u$, the individual utility $f_u(P)$ equals $|X(u)|-1$, which is how many nodes are in the group in addition to $u$. As shown in~\cite{Escoffier:2012vu,kleinberg2013information} $k$-stable partitions always exist for any value of $k$.
Let us first mention why this game converges for any $k\geq 1$. For a partition $P$ and any $i$ we define $\lambda_{i}(P)$ to be the number of groups of size $i$, and let $\Lambda(P) = (\lambda_{n}(P), \ldots, \lambda_{1}(P))$ be its \emph{partition vector}. One can see that a deviation \emph{of any size} produces a partition vector that is strictly higher in the lexicographical ordering. (for a more formal argument, see our Lemma~\ref{lem:convergence-k-stable} in the Appendix) 

We can then define $L(k,n)$ as the size of a longest sequence of $k$-deviations among all the graphs with at most $n$ nodes. One can prove (see Lemma~\ref{lem:graphe-vide} in the Appendix) that $L(k,n)$ is always attained in the graph containing no conflict edges, an instrumental observation for our next proofs. 

Prior to this work, no lower bound on $L(k,n)$ was known, and the analysis was limited to potential function that only applies when $k=1,2,$ and $3$~\cite{Escoffier:2012vu,kleinberg2013information}. The analysis of the game becomes much more difficult as soon as $4$-deviations are allowed. Table~\ref{tab:tab-graph-complexity} summarizes our contributions: 



\begin{table}[h]
\begin{center}
\begin{tabular}{|c|c|c|c|}
\hline
$k$ & Prior to our work & \multicolumn{2}{c|}{\textbf{Our results}} \\
\hline
1 & $O(n^{2})$ \cite{kleinberg2013information}& exact analysis, which implies $L(k,n) \sim \frac{2}{3}n^{3/2}$ & Theorem~\ref{lem:tight:k=1}  \\
\hline
2 & $O(n^{2})$ \cite{kleinberg2013information}& exact analysis, which implies $L(k,n) \sim \frac{2}{3}n^{3/2}$ & Theorem~\ref{lem:tight:k=2}  \\
\hline
3 & $O(n^{3})$ \cite{kleinberg2013information}& $\Omega(n^{2})$ & Theorem~\ref{thm:lower-bound-3}  \\
 & \cite{Escoffier:2012vu}&  & \\
\hline
$\geq 4$ & $O(2^{n})$ \cite{kleinberg2013information} & $\Omega(n^{\Theta(\ln(n))})$, $O(\exp({\pi\sqrt{{2n}/{3}}})/n)$ & Theorem~\ref{lem:lower-bound-4}  \\
\hline
\end{tabular}
\caption{Previous Bounds and results we obtained on $L(k,n)$.}
\label{tab:tab-graph-complexity}
\end{center}
\end{table}

\subsection{Exact analysis for k $\leq$ 2}
\label{ssec:k=12}

In~\cite{kleinberg2013information}, the authors proved the global utility increases for each $k$-deviation when $k \leq 2$. As that potential function is also upper bounded by $O(n^2)$, the system converges to a $2$-stable partition in at most a quadratic time.
We improve this result as we completely solve this case and give the exact (non-asymptotic) value of $L(k,n)$ when $k\leq 2$. It also proves that this potential function method is not tight. The gist of the proof is to re-interpret sequences of deviations in the Dominance Lattice.
This object has been widely used in theoretical physics and combinatorics to study systems in which the addition of one element (\eg a grain of sand) creates consequences in cascade (\eg the reconfiguration of a sand pile)~\cite{Goles1993321}. 
Let us define:
\begin{definition}[\cite{Brylawski1973201}]
An integer partition of $n \geq 1$, is a decreasing sequence of integers $Q = q_1 \geq q_2 \geq \ldots \geq q_n \geq 0$ such that $\sum_{i=1}^{n}q_{i} = n$. 
\end{definition}

Given any game with $n$ nodes, there are as many partition vectors as there are integer partitions of $n$.
Thus in the sequel, we will make no difference between a partition vector and the integer partition it represents.
If we denote the number of integer partitions by $p_{n}$, then we know the system reaches a stable partition in at most $p_{n}=\Theta((e^{\pi\sqrt{\frac{2n}{3}}})/n)$ steps\footnote{The upper-bound directly follows from Lemma~\ref{lem:convergence-k-stable} (\ie the lexicographical ordering argument).}.
This is already far less than $2^n$, which was shown to be the best upper bound that one can obtain for $k \geq 4$ when using an additive potential function~\cite{kleinberg2013information}. 
For $n\geq 6$ it can be seen that some partitions $P$ of the nodes may never be in the same sequence.
It is hence important to deal with a partial ordering instead of a total one. Brylawski proved in~\cite{Brylawski1973201} that a dominance ordering creates a lattice on integer partitions, where successors and predecessors can be defined by covering relation:
\begin{definition}(dominance ordering)
Given two integer partitions 
of $n \geq 1$,
denoted $Q = q_{1} \geq \ldots \geq q_{n}$ and $Q' = q_{1}' \geq \ldots \geq q_{n}'$, we say that $Q'$ dominates $Q$ if 
$\sum_{j=1}^{i}q_{j}' \geq \sum_{j=1}^{i}q_{j}$,
for all $1 \leq i \leq n$.
\end{definition}
\begin{definition}(covering)
Given two integer partitions $Q,Q'$ of $n \geq 1$, 
$Q'$ \emph{covers} $Q$ if and only if $Q'$ dominates $Q$ and no other integer partition $Q''$ satisfies $Q'$ dominates $Q''$ and $Q''$ dominates $Q$.
\end{definition}


\noindent
The main ingredient of our analysis is to exploit a strong relationship between this structure and $1$-deviations in a game, that holds as long as no conflict edges exist (see details of the proof in appendix).
\begin{lemma}[\cite{Brylawski1973201}]
\label{dominance-characterization}
Given $Q,Q'$,
$Q'$ covers $Q$ if and only if there are $j,k$ such that:
$(i)$ 
$q_j' = q_j+1$;
$(ii)$ 
$q_k' = q_k - 1$;
$(iii)$ 
for all $i \notin\{j,k\}$, we have $q_i' = q_i$;
$(iv)$ 
either $k=j+1$ or $q_j = q_k$. 
\end{lemma}
\begin{corollary}
\label{corollary:equiv-dominance-dev}
Assuming no conflict edges exist (\ie $G^{-}=(V,E)$ with $E=\emptyset$), let $Q,Q'$ be two integer partitions of $n$.
Then, $Q'$ dominates $Q$ if and only if there exist two partitions $P,P'$ of the nodes such that $\Lambda(P') = Q'$, $\Lambda(P) = Q$, and there is a valid sequence of $1$-deviations from $P$ to $P'$.
\end{corollary}

\begin{figure}[t]
\centerline{
\includegraphics[width=0.65\linewidth]{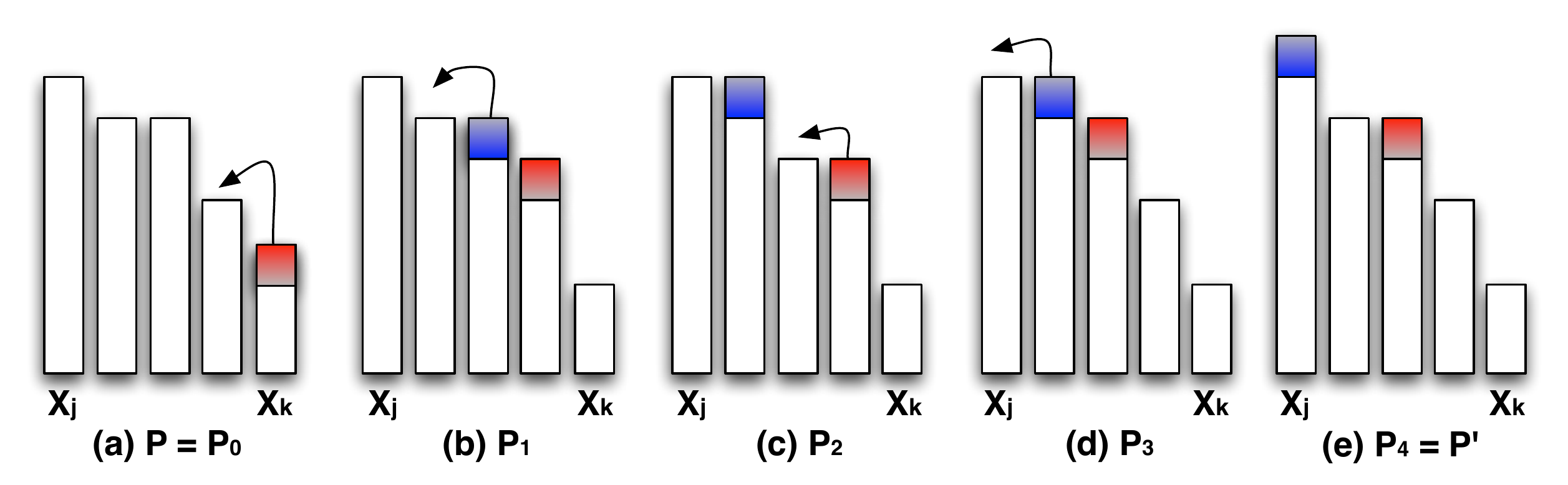}}
\caption{An example of decomposition for a $1$-deviation from group $X_k$ to group $X_j$.}
\label{fig:sandpile}
\end{figure}

In other words, any sequence of $1$-deviations, from a partition $P$ to another partition $P'$, can be decomposed into more elementary $1$-deviations, satisfying that if there is such an elementary deviation from the partition $P_i$ to the partition $P_{i+1}$, then the integer partition $\Lambda(P_{i+1})$ covers $\Lambda(P_i)$. Note that, by doing so, we may get another final partition $P'' \neq P'$ but it will have the same partition \emph{vector}, as seen on an example in Figure~\ref{fig:sandpile}.  
%
\begin{theorem}
\label{lem:tight:k=1}
Let $m$ and $r$ be the unique non negative integers such that $n = \frac{m(m+1)}{2} + r$, and $0 \leq r \leq m$.
We have $L(1,n) = 2 \binom{m+1}{3} + mr$.
This implies that $L(1,n) \sim \frac{2}{3} n\sqrt{n}$ as $n$ gets large.
\end{theorem}

Interestingly, we prove that $2$-deviations have a similar action on partition vectors, which implies:
\begin{theorem}
\label{lem:tight:k=2}
$L(2,n) = L(1,n)$.
\end{theorem}

\subsection{Lower bounds for k $>$ 2}
\label{subsec-superpolynomial-time}

Recall any sequence of $3$-deviations uses at most $O(n^{3})$ steps~\cite{kleinberg2013information}, by another application of the potential function method. In fact, before our work nothing proved that $L(3,n)$ is strictly larger than $L(2,n)$. We answer this question as we prove a quadratic lower bound when $k=3$ using a new technique.
The method heavily relies on specific sequences of deviations that we call \emph{cascades} discussed below. 

\begin{theorem}
\label{thm:lower-bound-3}
$L(3,n) = \Omega(n^{2})$.
\end{theorem}

The above result proves for the first time that deviations of multiple players can delay convergence and that the gap between $k=2$ and $k=3$ obtained from potential function is indeed justified. Using a considerable refinement of the cascade technique, we are able to prove a much more significant result: that $4$-deviations are responsible for a sudden complexity increases, as we now prove that no polynomial bounds exist for $L(4,n)$.

\begin{theorem}
\label{lem:lower-bound-4}
$L(4,n) = \Omega(n^{\Theta(\ln(n))})$.
\end{theorem}

The rest of this section is devoted to provide a brief sketch of the various steps used in the proof for Theorem~\ref{lem:lower-bound-4}. The actual proof uses these claims with additional observations to build a rather technical recursion on the largest sequence that occupies a large part of the appendix. As before, we assume no conflict edge exists since we know this provides a bound on the longest sequence that is attained.
Our approach relies on a \emph{vector representation} of partitions and deviations.
Indeed, recall any partition $P$ is represented by a vector $\Lambda=(\lambda_{n}, \ldots, \lambda_{i}, \ldots, \lambda_{1})$ where $\lambda_{i}$ represents the number of groups of size $i$,  $1 \leq i \leq n$.
The initial partition contains $n$ groups of size $1$ and the associated vector is $(0,\ldots,0,n)$.

Given a deviation from a partition $P$ to another partition $P'$, we then define the vector representation of this deviation as the difference $\Lambda(P') - \Lambda(P)$.
In particular, as we only consider $4$-deviations that consist in creating $1$ group of size $p$, removing $4$ groups of size $p-1$, creating $4$ groups of size $p-2$, and removing $1$ group of size $p-4$, with $p$ such that $5 \leq p \leq n/5$, we can formalize the vector representation $\delta[p]$ of such $4$-deviations in the following way:
$\delta[p]=\Lambda(P')-\Lambda(P)=(0,\ldots,0,1,-4,4,0,-1,0,\ldots,0)$.

We similarly formalize the $1$-deviations that we will use in our proofs as follows.
Recall that a $1$-deviation consists in creating $1$ group of size $p$, removing $1$ group of size $p-1$, removing $1$ group of size $q+1$, and creating $1$ group of size $q$ (if $q > 0$), for any $p, q$ such that $p \geq q + 2 \geq 2$ and $p + q \leq n$.
This yields either the vector representation $\alpha[p,q]= (0,\ldots,0,1,-1,0, \ldots, 0,-1,1,0,\ldots,0)$ or $\alpha[p,q]=(0,\ldots,0,1,-2,1,0,\ldots,0)$.
We also consider $1$-deviations such that $\alpha[p]=(0,\ldots,0,1,-1,0, \ldots, 0,-1)$ and such that $\alpha[2]=(0,\ldots,0,1,-2)$.

We are now ready to present a sequence for the coloring game with a super-polynomial size.
Without loss of generality, we assume that the number of nodes is $n = cL(L+1)/2$.
The values $c$ and $L$ will be defined later and we assume for the moment that $c$ is sufficiently large.
We first do a sequence of $1$-deviations to get the partition $P^{0}$ such that $\Lambda^{0}=(0,\ldots,0,\lambda_{L}=c,\ldots,\lambda_{1}=c)$ is the  vector of $P^{0}$.
Using our notations, we observe that to get the partition $P^{0}$ we do the sequence defined by ($\sum_{j=2}^{L} \alpha[j]$) exactly $c$ times.
Thus, one can check that $\Lambda^{0} = (0,\ldots,0,n) + c\sum_{j=2}^{L} \alpha[j]$.

In the sequel, we will denote any sequence of deviations by $\phi[l]$, where $l$ denotes the greatest index $j$ such that $\phi_j \neq 0$.
Equivalently, $l$ denotes the maximum size of a group having users involved into the deviation $\phi[l]$.
Observe that given any sequence of deviations defined by a vector $\phi[l]$, the vector $\phi[l-i]$, $i > 0$, represents the same sequence of deviations in which any $4$-deviation $\delta[p]$ is replaced by $\delta[p-i]$ and any $1$-deviation involving groups of sizes $p$ and $q$, now involves groups of sizes $p-i$ and $q-i$.
Let $t > 0$ and $T > 0$ be such that $2^{T-1} (2t^{2}+2) \leq L$.

\medskip
To go further with our vector approach, we need to check whether a given deviation is valid.
We actually remark that a deviation $\phi[l]$ is valid for the partition $P$ if there is no negative value in the vector $\Lambda(P) + \phi[l]$.
This yields the notion of \emph{balanced sequence}.


\begin{definition}
Given any integer $h > 0$, a sequence of deviations from a partition $P$ is $h$-balanced if, and only if, $\zeta = \Lambda' - \Lambda$ is such that $\zeta_{j} \geq -h$ for all $j$, $1 \leq j \leq n$, where $\Lambda$ is the partition vector of $P$ and $\Lambda'$ is the partition vector of any intermeriady partition obtained once one or several deviations in the sequence have been applied.
\end{definition}

Intuitively, it suffices to have $h$ users in every group of the initial partition so that a $h$-balanced sequence of deviations from this partition is valid.
In particular if a sequence of deviations from $P^{0}$ to another partition $P'$ is $c$-balanced, then this sequence is valid.
In the following, we construct a sequence of $1$-deviations and $4$-deviations and we will prove a sufficient value for $c$ to get the validity of this sequence, that is the sub-sequence from $P^{0}$ to \emph{any} partition of the sequence, is $c$-balanced.

Given a sequence of deviations, the \emph{size} of its corresponding vector is the length of the minimum-size sub-vector that contains all its non-zero values.

\begin{definition}
\label{def:symmetric-property}
A vector has the \emph{symmetric property} if, and only if, the minimum-size sub-vector that contains all non-zero values is a symmetric vector.
\end{definition}

Note that for a vector $\phi[L]$ with the symmetric property, $L$ is enough in that case to deduce the size of $\phi$.
Also, $\phi[L-i]$ defines the same sequence than $\phi[L]$ in which the sizes of the groups involved in the deviations have been decreased by $i$.

\begin{claim}
\label{claim:symmetry}
Let $\phi[L]$ be any vector of size $s$ with the symmetric property.
For any positive integers $r$ and $d$, $1+(r-1)d \leq s$, $\phi' = \sum_{h=0}^{r-1} \phi[L - h d]$ has the symmetric property.
\end{claim}

In order to prove Theorem~\ref{lem:lower-bound-4}, we construct the sequences of deviations defined by $\zeta^{i}[L]$ for all $i=1, \ldots T$.
The initial vector $\zeta^{1}[L]$ is defined according to Claim~\ref{claim:premiere-sequence} which is rather complex and hence we left in the appendix. To construct the vector $\zeta^{i+1}[L]$ from $\zeta^{i}[L]$, for all $i=1, \ldots, T-1$, we follow a particular construction that is shown valid by Claim~\ref{claim:sequence-generale}, and which is illustrated in Figure~\ref{fig:cascade}. The construction is composed of a repetition of the sequence defined by $\zeta^{i-1}$ a certain number of times (linear in $t$) shifting the "starting point" of each sequence by the same value (which Claim~\ref{claim:sequence-generale} allows to formalized). The construction then adds $1$-deviations in order to get a special property, called \emph{Good property}, defined below.

In the following, we denote by $s_{i}$ the size of the vector $\zeta^{i}[L]$.
We first define the notion of \emph{Good Property}.

\begin{figure}[t]
\centerline{
\includegraphics[width=0.65\linewidth]{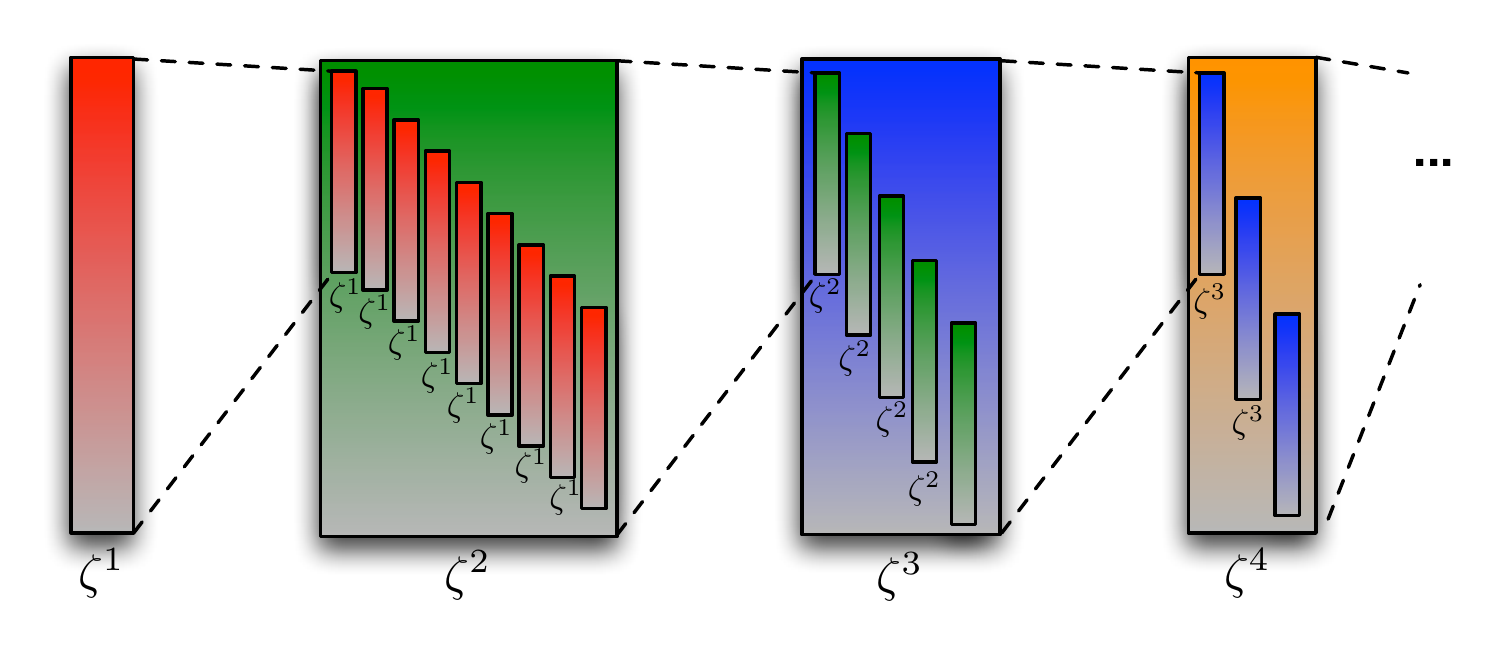}}
\caption{Long sequence using recursive cascades.}
\label{fig:cascade}
\end{figure}

\begin{definition}
\label{def:good}
Given any integer $i \geq 1$, $\zeta^{i}[L]$ has the \emph{Good Property} if $s_{i}$ is even, if $\zeta^{i}[L]$ has the symmetric property, and if there exists $t^{1}_{i}$, $t^{2}_{i}$ with $1 < t^{1}_{i} < t^{2}_{i} < 2 t^{1}_{i}$, $t^{2}_{i} \leq 2^{i+1}$, such that $\zeta^{i}[L]_{L} = 1$, $\zeta^{i}[L]_{L-t^{1}_{i}} = -1$, $\zeta^{i}[L]_{L-t^{2}_{i}} = -1$, $\zeta^{i}[L]_{L-s_{i}/2+1} = 1$, and $\zeta^{i}[L]_{j} = 0$ for other values of $j$.
\end{definition}

\begin{claim}
\label{claim:sequence-generale}
One can define recursively $( \zeta^{i}[L] )_{1 \leq i \leq T}$ such that $\zeta^1[L]$ is defined as for Claim~\ref{claim:premiere-sequence}, and for all $i = 1, \ldots, T-1$ there exists a sequence of $1$-deviations $\xi^{i+1}$ such that $\zeta^{i+1}[L] = \sum_{j=0}^{a_{i}} \zeta^{i}[L-j t^{1}_{i}] + \xi^{i+1}$, where $a_{i}$ is the largest even integer $j$ such that $L-jt^{1}_{i}-t^{2}_{i} > L-s_{i}/2+1$, and $\zeta^{i}[L]$ has the Good Property.
\end{claim}

\begin{claim}
\label{claim:c-balanced}
For all $i$, $1 \leq i \leq T$, the sequence of deviations from $P^{0}$ to $P^{i}$ defined in Claim~\ref{claim:sequence-generale}, is $(c_{1}+i-1)$-balanced where $c_{1}$ is the constant defined in Claim~\ref{claim:c-0-balanced}.
\end{claim}

Recall that $t, T > 0$ are such that $2^{T-1} (2t^{2}+2) \leq L$.
The size $s_{1}$ of $\zeta^{1}[L]$ is $2t^{2}+2$ and $s_{i} \leq s_{i+1} \leq 2 s_{i}$ for all $i$, $1 \leq i \leq T-1$.
Also, recall that $n = cL(L+1)/2$.
Consider $T = \log_{2}(t)+1$ and, without loss of generality, assume $2(t^{3}+t) = L$.
By Claim~\ref{claim:c-balanced}, $n = (c_{1}+T-1) (2t^{3}+2t) (2t^{3}+2t+1)/2 = (c_{1}+\log_{2}(t)) (2t^{3}+2t) (2t^{3}+2t+1)/2$ is sufficient to do the sequence $\zeta^{T}[L]$.
As $c_{1}$ is a constant, $n = O(t^{6} \log_{2}(t))$.
For all $i$, $1 \leq i \leq T$, let $|S_{i}|$ be the size of the sequence defined by $\zeta^{i}[L]$.

\begin{claim}
\label{claim:taille-sequence-1}
For all $i$, $1 \leq i \leq T-1$, $|S_{i+1}| \geq (\frac{s_{i}}{2^{i+2}} -6) |S_{i}|$.
\end{claim}

\begin{claim}
\label{claim:taille-sequence-11}
$|S_{T}| \geq O(t^{log_{2}(t)}) |S_{1}|$.
\end{claim}

As $n = O(t^{6} \log_{2}(t))$, Claim~\ref{claim:taille-sequence-11} proves  Theorem~\ref{lem:lower-bound-4}.

\section{The general case: Are games stable under deviations?}
\label{sec:weighted}

\def\MIS{{\sc Maximum Independent Set problem}\xspace}
\def\SD{{\sc k-Stable decision problem}\xspace}

The uniform case we analyzed so far may be described as a ``clique with enemies'': A pair of nodes may either be connected by an edge with unit weight: $w_{uv}=1$, or they are in conflict: $w_{uv}=-\infty$. 
More generally, a coloring game may be defined with weights taking values in a larger set $\mathcal{W}$.
Nodes then choose to interact with each others according to more complex preferences and the individual utility is not always related to the size of the sharing group.
We first show the following.

\begin{theorem}
\label{thm:stability:k=1}
For any weighted graph, if only $1$-deviations are allowed, a coloring game always converges in $O(w_pn^2)$ steps, where $w_p$ denotes the largest positive weight, if any, and $0$ otherwise. 
\end{theorem}

Theorem~\ref{thm:stability:k=1} proves that all these games admit a $1$-stable partition, or equivalently, a Nash-equilibrium. 
Note on the other hand that our dependency on the weights can be arbitrarily bad.
It is not that surprising, as the problem of computing a $1$-stable partition is PLS-complete (a simple reduction from cut games can be found, see~\cite{Johnson:1988ur}).
So we will always work with a fixed set of weights $\mathcal{W}$.
The proof of this result extends a potential function argument (see Lemma~\ref{lem:util-variations} in the Appendix).

We now show how this analysis becomes more complex as $k$-stable partition (with $k>1$) may not necessarily exist.
We look at the greatest values of $k$ for which $k$-stability is \emph{always} eventually attained, for all games where weights are chosen in a fixed subset $\mathcal{W}$. We denote this value $k(\mathcal{W})$, and establish it for various possible set of weights, which requires both to extend our previous potential function argument proof as well as prove counter-example of stability with minimum values of $k$.
We finally strengthen the aforementioned results by proving that deciding if a graph with weights in $\mathcal{W}$ admits a $k$-stable partition is either trivial (\ie it is true for all such graphs) or NP-complete.


\subsection{Games with a unique positive weight}
\label{subsec:-m01}

We first focus on the subsets of $\{-\infty,0,1\} \cup -\mathbb{N}$, which contains in particular $\mathcal{W} = \{-\infty,0,1\}$ which is the simplest set of weights extending the uniform case to accommodate \emph{indifferent} edges.
Whereas it does not differ that much from the set $\{-\infty,1\}$ previously studied in the uniform case, and which satisfies $k(\{-\infty,1\})=\infty$, we show its stability properties are radically altered.
Indeed we prove that in that case that $k(\mathcal{W}) = 2$ after exhibiting a surprising counterexample for $k=3$.

\begin{lemma}
\label{lem:2-stable}
Every coloring game with weights in $\{-\infty,0,1\} \cup -\mathbb{N}$ converges to a $2$-stable partition in $O(n^2)$ steps. This upper-bound on convergence time is tight.
\end{lemma}

\begin{lemma}
\label{lem:counter-example}
There exists $G=(V,w)$ with $\mathcal{W} = \{-\infty,0,1\}$ such that there is no $3$-stable partition.
\end{lemma}

Figure~\ref{fig:counter-examples} presents a graph with 18 vertices, that is does not admit any stable partition for $k=3$. Note in particular the presence of $4$ indifferent edges shown in dashed lines: Without these edges a $k$-stable partition would exist for all values of $k$. The proof illustrates the complexity of sequence of $3$ deviations which are able to create infinite sequence of deviation at constant utility. 
To keep the graph readable, we use conventions. 
(1) Some sets of nodes are grouped within a circle; an edge from another node to that circle denotes an edge to \emph{all} elements of this set.
(2) All nodes that are not connected by an edge on the Figure are enemies with weights $-\infty$. 
(3) Green solid edges represent edges with weight $1$, whereas blue dashed edges represent edges with weight $0$. 

\subsection{Game with general weights}

While the two results we obtained ($k(\mathcal{W})\geq 1$ in general, $k(\mathcal{W})=2$ when $\mathcal{W}$ contains a single positive weight) seem constrained, we now prove that these are the best results that one can hope for. 
%
%
%
%
%
Table~\ref{tab:tab-1} summarizes the most important tight values for $k(\mathcal{W})$.
In particular we prove $k(\mathcal{W})=\infty$ if and only if we have one of the following cases: $\mathcal{W}=\{-\infty,b\}$, which is equivalent to the uniform case; $\mathcal{W}$ has no negative (or no positive) elements, in which case the game is trivial; and $\mathcal{W} \subseteq -\mathbb{N} \cup \{N\}$, where $N$ denotes \emph{best friends} (see footnote 6 p.4) and stable groups are those connected by $N$-edges.

\begin{figure}%
\centering
\begin{minipage}[c]{0.38\textwidth}%
\centering
    \includegraphics[width=1\textwidth]{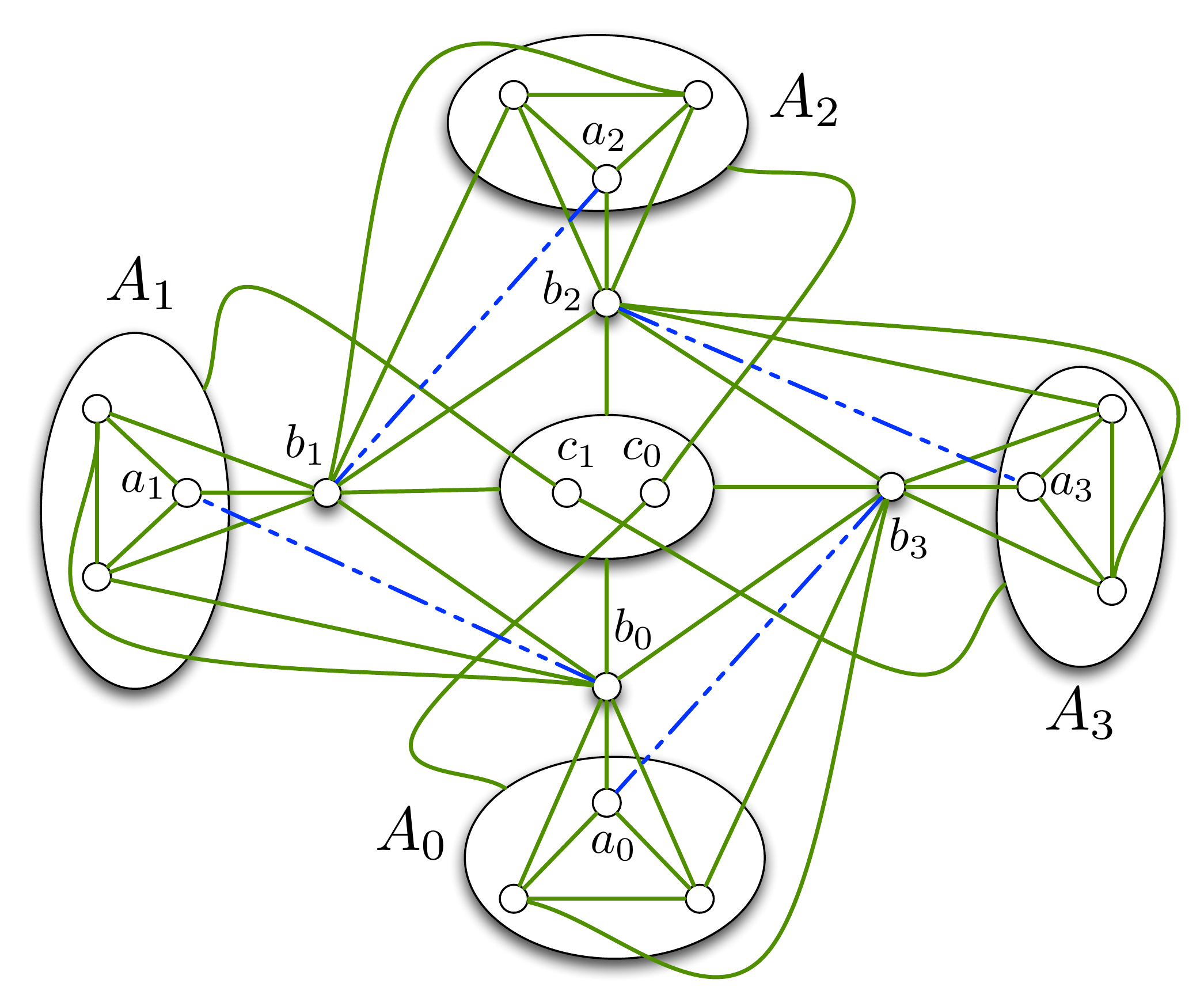}
\caption{Graph that does not admit a 3-stable partition.}
\label{fig:counter-examples}
\end{minipage}
\qquad
\parbox{0.5\textwidth}{
\begin{center}
\begin{tabular}{c|c}
$k(\mathcal{W})$ & $\mathcal{W}$ \\
\hline
$1$ & $\{-\infty,a,b\}$, $0 < a < b$\\
\hline
$2$ &  $\{-\infty,-\mathbb{N},0,1\}$\\
\hline
$\infty$ & $\{-\infty,b\}$, $b > 0$; \\
& $\mathcal{W}\subseteq \NatInt \cup \{N\}$; \\
& $\mathcal{W}\subseteq -\NatInt \cup \{-\infty\}$; \\
& $\mathcal{W} \subseteq -\NatInt\cup\{N\}$;
\\
\end{tabular}
\caption{Values of $k(\mathcal{W})$ for different $\mathcal{W}$.}
\label{tab:tab-1}
\end{center}
}
\end{figure}

\subsection{Intractability with conflict graphs}

Under general weights, we have proved that not all games in general have a $k$-stable partition, even for relatively simple set of weights and small values of $k$. On the other hand, one could argue that these results comes from pathological cases which could be ruled out after checking a property of the graph. We prove that cannot be the case: not only is the stability not guaranteed for a given game when $k > k(\mathcal{W})$, but it is computationally prohibitive to decide it. We first define:

\begin{definition}[\SD]
Let $k \geq 1$ and let $\mathcal{W}$ be the (constant) set of weights.
Given a graph $G=(V,w)$, does there exist a $k$-stable partition?
\end{definition}


\begin{theorem}
\label{thm:complexity-decision}
For $k \geq 1$ and $\mathcal{W}$ containing $-\infty$, either a $k$-stable partition always exists (\ie $k\leq k(\mathcal{W})$); or the \SD is NP-hard.
\end{theorem}

The previous result of NP-hardness in~\cite{kleinberg2013information} requires another kind of deviations in addition to the classical $k$-deviations we consider in our paper (see \emph{gossip deviations} in Section~\ref{sec:extensions}).
Moreover, the large positive weight $N$ is essential in their proof, whereas ours overrules these strong constraints.

\def\MIS{{\sc Maximum Independent Set problem}\xspace}

\section{Extensions of coloring games}
\label{sec:extensions}

All theoretical models of social dynamics should consider whether the overall behavior of the model is not too limited by some simplifying assumptions.
We prove this is not the case, as our results extend to account for various situations in the formation of social groups, including variants previously discussed and new ones.

\subsection{Gossiping}
\label{sec:gossip}


In this model, all $k$-deviations as previously defined are allowed, and in addition two nodes belonging to two different groups may ``gossip'', and in that case the two groups that they are part of are merged. Obviously, and as before this deviation will only takes place if the two nodes benefit from the merge, but what is unique here is that this deviation does not require that other nodes in these groups benefits from the merge. They may see a decrease in utility, but they are not given a choice to block the deviation, although they are in some sense seeing a modification of their group.

This actually turns out to be equivalent to our model in the uniform case. One can easily check that, in the uniform case, there is a valid $1$-deviation whenever there exists a gossip-deviation. Consequently all our results apply in that case, closing previously open problem with this model. In the weighted case, a first negative result has been proved in~\cite{kleinberg2013information}: when $\mathcal{W}=\{-\infty,1,N\}$, it is NP-hard to decide whether a given weighted graph admits a gossip-stable partition. This proof hence requires multiple positive weights, including a weight that grows with the number of nodes. We provide a stronger negative result:
\begin{theorem}
\label{NPKL}
Deciding whether there exists a $2$-stable, gossip-stable partition for a given graph $G=(V,w)$ is NP-hard, even with $\mathcal{W}=\{-\infty,0,1\}$.
\end{theorem}

Note that this proves that gossip makes this game strictly harder to decide: Gossip creates instability even when a unique and fixed positive weight exists. Without gossip, in such case a $2$-stable partition always exists.

\subsection{Asymmetry}
\label{sec:asym}

Studying directed graphs rather than undirected graphs is a natural generalization.
In this case, we may not have that $w_{uv} = w_{vu}$ for all nodes $u,v$.
However, even if modest generalization of the model, asymmetrical weights lead to intractability.
This can be seen with a simple digraph $D=(\{u,v\},w)$ such that $w_{uv} > 0$ whereas $w_{vu} < 0$.
Furthermore, the problem of deciding whether there exists a $1$-stable partition is NP-hard.
This result holds even when an upper-bound for the number of groups is known.

\begin{theorem}
\label{theorem:oneStableConfigurationProblemNP}
For any $k \geq 1$, deciding whether there exists a $k$-stable partition for a given digraph $D=(V,w)$ is NP-hard, even if the number of groups is upper-bounded by any constant $g \geq 2$.
\end{theorem}

\subsection{Multichannel model and overlapping groups}
\label{sec:multichannel}

One assumption of our model is that nodes form a partition, hence they are limited to a single channel to interact with their peers. 
In reality participants in social networks may engage in \emph{multiple} groups, although one can participate to only so many of them due to time.
This motivates us to extend partitions of the nodes $P$ into multisets, that we call 'configurations'.
We denote such a configuration $C$, and it is said to use $q$ channels if each user participates (at most) to $q$ groups.
As a user can always form singleton groups, we can assume that $|X(u)|=q$ for any $u$, where $X(u)$ now denotes the set of all groups in $C$ that contain $u$. 
In other words $C = (X_1,X_2,\ldots,X_{qn})$ with possibly empty groups.
Note that partitions of the nodes $P$ are exactly configurations with a single channel.

The utility of $u$ depends on the number of groups that $u$ shares with each peer:
\begin{equation}
	f_{u}(C) = \sum_{v \in V} h\left( \left| X(u)\cap X(v) \right| \vf w_{uv}\right)\vf
	\label{eq:userutility}
\end{equation}
where $h(g,w)$ is a function that measures the utility of sharing $g$ groups with a node with weight $w$.
Note that we assume, without loss of generality, that
\[
\begin{disarray}{l}
	h(0,.)=0\vf h(.,0)=0 \midwor{and} \forall w\in\RelInt \vf h(1,w)=w\vf \\
\forall g\in\NatInt \vf w\mapsto h(g,w) \midwor{is a non-decreasing function,} \\
\forall w\in\RelInt \vf g\mapsto w\cdot h(g,w) \midwor{is a non-decreasing function.}
\end{disarray}
\]
The last property simply ensures that $h(g,w)$ increases with $g$ when $w$ is positive, and decreases with $g$ when $w$ is negative.
As an example: $h(g,w)=s(g)\cdot w$, for $s$ a non-decreasing function that initially takes values $0$ and $1$, always satisfies these properties.
Especially, when $s(g)=\ind{g>0}$:
\vspace{-0.2cm}
\begin{equation}
	f_{u}(C) = \sum_{v \in V} \ind{ X(u)\cap X(v) \neq \emptyset} \cdot w_{uv} \ff 
	\label{eq:userutility1channel}
\end{equation}

Using the same potential function as before, \eg global utility, it follows that there always exists a $1$-stable configuration with $q$ channels.
It is now natural to wonder what is the behavior of $k_q(\mathcal{W})$ when the number of channels $q$ is higher than $1$.
While it may have none or positive effects, we show that it is not always the case.
Indeed, even for uniform games, Lemma~\ref{lem:contre-exemple-q=2} proves $k_2(\{-\infty,1\}) \leq 2$ with $2$ channels while $k_1(\{-\infty,1\}) = k(\{-\infty,1\}) = \infty$ with single channel. 
The idea is to emulate a weight $w_{uv} = 0$ by forcing nodes $u$ and $v$ to be in a predetermined group.

\begin{lemma}
\label{lem:contre-exemple-q=2}
Let us fix $h$ is as $h(0,w) = 0$, $h(g+1,w) = (1 + \varepsilon g)w$, where $\varepsilon > 0$ is arbitrarily small.
There exists a graph $G=(V,w)$ with $\mathcal{W}=\{-\infty,1\}$ such that any configuration with $2$ channels for $G$ is not $3$-stable.
\end{lemma}

In the case of general weights, a stability criterion depending on $q$ has yet to be found.
We prove in Theorem~\ref{the:chaotic-behavior} that any chaotic behavior may be observed between any two consecutive numbers of channels.

\begin{theorem}
\label{the:chaotic-behavior}
For any sequence of positive integers $q_{1}, q_{2}, \ldots, q_{p}$, and for any sequence $1 \leq i_{1} < i_{2} < \ldots < i_{l} \leq p$, there exists a graph $G=(V,w)$ such that there is a $2$-stable configuration with $q_j$ channels for $G$ if, and only if, $j \neq i_{t}$, for every $1 \leq t \leq l$.
\end{theorem}

\subsection{Multi-modal relationship}
\label{sec:multimodal}

Our model so far is heavily biased towards pairwise relationship, as the utility of a player depends on the sum of her interaction with all other members of the groups.
In reality, more subtle interaction occur: one may be interested to interact with either friend $u$ or $v$, but would not like to join a group where both of them are present. 
Our analysis also generalizes to this case.
We use a hypergraph based model and we prove that, for instance,  there always exists a $1$-stable partition

\section{Efficiency stability trade-off}
\label{sec:efficiency}

\def\MIS{{\sc Maximum Independent Set problem}\xspace}

Efficiency of configuration can be estimated through the sum of utility received by all the players.
On the positive side, we have a construction to show that, in our extended model with multichannel, using more channels adds in general no particular difficulty to that computation.

\begin{theorem}
\label{the:equivalence-maximum}
Assume $h$ concave.
Given a graph $G = (V,w)$ with $\setW \subseteq \{-\infty\} \cup \NatInt$, one can construct a graph $G'$ such that from a maximum partition of the nodes for $G'$, we deduce in linear time a maximum configuration with $q$ channels for $G$.
\end{theorem}

However, it has been proved in~\cite{kleinberg2013information}, that computing a maximum partition is NP-complete even with uniform games.
Hence, it does not only require coordination, but also has a prohibitive computational cost.
We first strengthen this result by showing that this optimum cannot even be \emph{approximated} within a sublinear multiplicative factor.

\begin{lemma}
\label{lem:mcp-not-multiplicative}
For every $1 > \epsilon > 0$, the problem of finding a maximum partition cannot be approximated within any $n^{\epsilon - 1}$-ratio in polynomial time, unless P=NP.
\end{lemma}

Note that in our extended model with multichannel, the total utility of the players can only improve as the number of channels $q$ gets larger, because more configurations are available. Unfortunately, we also prove that computing the minimum number of channels that guarantee a given threshold global utility is hard to approximate.

\begin{theorem}
\label{the:min-q-pas-APX}
Given $U \geq 1$, computing $q^{*}$, the minimum $q$ for which there exists a configuration $C$ satisfying $f(C) \geq U$, cannot be approximated within $n^{1-\varepsilon}$ for any $\varepsilon > 0$ in polynomial time, unless P=NP.
\end{theorem}

However, in spite of this complexity, we prove that collusion and multiple channels have in general a beneficial effect on the price of anarchy, which is the ratio of the worst stable equilibrium to this optimal.
Nash Equilibria obtained by the dynamic of the system satisfy a general bound that improves with $q$, in contrast with worst Nash Equilibria that are always arbitrarily bad.
Finally, we show that higher order stability, although it is significantly harder to obtain in general, potentially improves this bound significantly.

\subsection{Efficiency of Nash Equilibrium}
\label{sec:nashefficiency}

Since we proved there is few hope we can compute efficient configurations in polynomial time, our goal is now to study the efficiency of any stable configuration, and especially the ones that we can compute by using our dynamic system (extended to many channels).

We start by studying the efficiency of the largest class of stable configurations: Nash Equilibria ($k=1$).
Let us suppose $q$ to be fixed in the context.
For any graph $G$, we denote by $C^+$ (resp. $P^+$ when $q=1$) a maximum configuration with $q$ channels (resp. partition); the \emph{worst case} $C_k^-$ (resp. $P_k^-$) denotes a configuration with the minimum utility that is obtained in a $k$-stable configuration using $q$ channels (resp. partition).
We first show in Lemma~\ref{poa:infinite-ratio} that Nash Equilibria can be arbitrarily bad:

\begin{lemma}
\label{poa:infinite-ratio}
Assume here $q$ equals $1$.
Let $a,b > 0$ be two integers.
With $\mathcal{W} \supseteq \{-\infty,0,b\}$ or $\mathcal{W} \supseteq \{-a,b\}$, for any $R \geq 1$, there are graphs $G=(V,w)$ such that $f(P^+) \geq R$ and $f(P_1^-) = 0$.
\end{lemma}

On the other hand, such situation is an oddity as precised in Lemma~\ref{greedy:upper-bound-gal}.
We now denote by $C^{\textrm{d}-}_{1}$ the worst equilibrium obtained under the dynamic of the system.
Note that in the following, $w_p$ is the maximum weight of $\mathcal{W}$.

\begin{lemma}
\label{greedy:upper-bound-gal}
Given integers $q, n \geq 1$ and a set $\mathcal{W}$ such that $\mathcal{W}^+ \neq \emptyset$, for any graph $G=(V,w)$ with $\mathcal{W}$, we have $f(C^+)/f(C^{\textrm{d}-}_{1}) = O(h(q,w_p)n)$.
\end{lemma}

The upper-bound above, if it were tight, would imply that increasing of the number of channels might deteriorate the quality of the lowest computable Nash Equilibria.
On the contrary, we prove in Lemma~\ref{lemma:positive-impact} that the number of channels may have a positive impact for some  sets of weights.

\begin{lemma}
\label{lemma:positive-impact}
Given integers $q,n \geq 1$, and a set $\mathcal{W} = \mathcal{W}^+ \cup \{-\infty\}$ such that $\mathcal{W}^+ \neq \emptyset$, for any graph $G = (V,w)$ with $\mathcal{W}$, we have $f(C^+)/f(C^{\textrm{d}-}_{1}) = O((1 + 2n/q)h(q,w_p))$.
\end{lemma}

\subsection{Efficiency of higher order stability}
\label{sec:stabilityefficiency}

Increasing $k$ beyond 1 restricts further the set of stable configurations; they may even sometimes cease to exist. However, before this happens, we can expect the worst case configurations, and hence our general bounds, to improve.
We prove in Theorem~\ref{poa:upper-bound:kgeq2} that it is generally the case by considering single channel case (the existence of a stable configuration is not guaranteed, but the result holds if it exists). 
As we constrain ourselves to partitions of the nodes anew, we will denote the configurations $P$ instead of $C$ in the sequel.

\begin{theorem}
\label{poa:upper-bound:kgeq2}
Given $\mathcal{W}$, any graph $G = (V,w)$ with weights in $\setW$ satisfies, $f(P^+)/f(P_k^-) = O(\Delta_+)$
for $k\geq 2$, where 
$\Delta_+=\max_{v\in V} |\lset u \dimset w_{uv} > 0 \rset|$.
\end{theorem}

Note the result also holds for $k=1$ whenever neither $\{-\infty,0,b\} \subseteq \mathcal{W}$ nor $\{-a,b\} \subseteq \mathcal{W}$. 
Indeed, in this case, if there exist $u,v$ such that $f_u(P) = f_v(P) = 0$, whereas $w_{uv} > 0$, then $u$ can reach the group of $v$, or vice-versa, and so, there is a $1$-deviation.

In other words, we show that for $k\geq 2$ stable partitions are with a multiplicative factor of the optimal given by the maximum degree of a node in the friendship graph.

We define the \emph{price of $k$-anarchy} $p_a(n,k)$ as the greatest ratio $f(P^+)/f(P_k^-)$ obtained for a graph $G$ with $n$ nodes.
By considering a worst case degree, we have $p_a(n,k) = O(n)$. 
For a given $k$, this bound is tight: we can even construct a graph with $\setW\subseteq \NatInt$ and any other non trivial weights such that $f(P^+)/f(P_k^-) = \Omega(n)$.

\begin{lemma}
\label{poa:lower-bound-1-2}
Let $k$ be any integer $k \geq 2$.
Given positive integers $b,b'$ such that $b' < b$, and a non-negative integer $a$, there is a graph $G=(V,w)$:
\begin{itemize}
\item[-] with $\mathcal{W} \supset \{-a,b\}$, such that $\frac{f(P^+)}{f(P_k^-)}  \geq  \frac{ (k+1)\lfloor \frac{n}{(k+1)^2} \rfloor - 1}{k}$.
\item[-] with $\mathcal{W} \supset \{b,b'\}$, such that $\frac{f(P^+)}{f(P_k^-)} \geq 1 + \frac{b'}{b}\frac{kb(\lfloor \frac{n}{kb} \rfloor - 1)}{kb-1}$.
\end{itemize}
\end{lemma}



\label{totalpages}

\newpage
\bibliographystyle{abbrv}
\bibliography{Bibliographie/InfoSharing,Bibliographie/BehavPrivacy,Bibliographie/AllPapers,Bibliographie/IntegerPartitions}

\newpage

\section{Appendix}

\subsection{Appendix 1: proofs of Section~\ref{sec:unweighted}}

\begin{lemma}
\label{lem:convergence-k-stable}
For any $k \geq 1$, for any conflict graph $G^{-}$, the system converges to a $k$-stable partition.
\end{lemma}

\begin{proof}~[Lemma~\ref{lem:convergence-k-stable}]
Let $P_{i}, P_{i+1}$ be two partitions for $G^-$ such that $P_{i+1}$ is obtained from $P_i$ after a $k$-deviation for $P_i$. 
We prove $\Lambda(P_{i}) <_{L} \Lambda(P_{i+1})$ where $<_L$ is the lexicographical ordering.
By definition, for any $u \in S$, we have $f_{u}(P_{i}) < f_{u}(P_{i+1})$, where $S$ represents the set of nodes involving in the $k$-deviation ($|S| \leq k$).
Thus, we get $\Lambda(P_{i+1}) - \Lambda(P_{i}) = (0, \ldots, 0, 1, \ldots)$, and so $\Lambda(P_i) <_{L} \Lambda(P_{i+1})$.
Finally, as the number of possible vectors is finite, we obtain the convergence of the system.
\end{proof}

\begin{lemma}
\label{lem:graphe-vide}
For any $k, n \geq 1$, there is a sequence of $k$-deviations with size $L(k,n)$ in $G^{\emptyset}$, where $G^{\emptyset}=(V,E)$ is the conflict graph such that $|V|=n$ and $|E|=0$.
\end{lemma}

\begin{proof}~[Lemma~\ref{lem:graphe-vide}]
In order to prove the lemma, let $k, n \geq 1$, and let $G^-=(V,E)$ be any graph with $|V|=n$.
We consider a longest sequence of $k$-deviations for $G^-$.
Then, we can mimic this sequence for $G^-$ in order to get a valid sub-sequence for $G^{\emptyset}$.
Indeed, every $k$-deviation for $G^-$ is a valid $k$-deviation for $G^{\emptyset}$.
We get there is a sequence of $k$-deviations for $G^{\emptyset}$ that is at least as long as a longest sequence of $k$-deviations for $G^-$.
Because there is a finite number of conflict graphs with $n$ nodes, this clearly implies that there exists a sequence of $k$-deviations with maximum size for $G^{\emptyset}$.
\end{proof}

\begin{proof}~[Corollary~\ref{corollary:equiv-dominance-dev}]
First, assume $Q'$ covers $Q$.
By Lemma~\ref{dominance-characterization}, there exist $j,k$ such that  $q_j' = q_j+1$,  $q_k' = q_k - 1$, and for all $i$ such that $i \neq j,k$, $q_i' = q_i$.
Moreover, since $k=j+1$ or $q_j = q_k$, we get $q_j \geq q_k$.
Let $P$ be any partition of the nodes such that $\Lambda(P) = Q$.
We can write $P = ( X_1, \ldots, X_n )$ such that $|X_i| = q_i$ for all $1 \leq i \leq n$.
Let $v$ be any vertex in $X_k$.
Such a vertex exists because $|X_k| = q_k > 0$.
Then $v$ can move from $X_k$ to $X_j$, and it is a valid $1$-deviation because $|X_j| = q_j \geq q_k$.
In so doing, we get a new partition $P' = ( X_1', \ldots, X_n' )$ for the nodes such that $|X_j'| = |X_j|+1$, $|X_k'| = |X_k|-1$, and for all $i$ such that $i \neq j,k$, $|X_i'| = |X_i|$.
In other words, $\Lambda(P') = Q'$.

More generally, if $Q'$ dominates $Q$ then there exists a sequence $Q = Q_1, Q_2, \ldots, Q_p = Q'$, such that for every $1 \leq i \leq p-1$, $Q_{i+1}$ covers $Q_{i}$.
Therefore, we can iterate the process, and we get the expected result.\\

\noindent
Conversely, let $P = ( X_1, \ldots, X_n )$, be a partition for the nodes,
let $S = \{v\}$ be any $1$-deviation that breaks $P$,
and let $P' = ( X_1', \ldots, X_n')$ be the resulting partition for the nodes. 
Then $v$ leaves some group $X_{k}$ for another group $X_{j}$.
Furthermore, $|X_{j}| \geq |X_{k}|$ by the hypothesis.
So, we get either $|X_{j}| = |X_{k}|$, or $|X_{j}| > |X_{k}|$, hence $j \leq k-1$.
Let us denote $Q = \Lambda(P)$ and $Q' = \Lambda(P')$.
By a reordering of groups with equal size, we may assume that $X_j$ is the first group with size $|X_j|$, whereas $X_k$ is the last group with size $|X_k|$ in $P$.
Then:
\begin{itemize}
\item if $i \in \{ 1, \ldots, j-1 \}$, then $\sum_{l=1}^i |X_l'| = \sum_{l=1}^i |X_l|$;
\item if $i \in \{ j, \ldots, k-1 \}$, then $\sum_{l=1}^i |X_l'| = [\sum_{l=1}^i  |X_l|] + 1 $;
\item if $i \in \{ k, \ldots, n \}$, then $\sum_{l=1}^i  |X_l'| = \sum_{l=1}^i |X_l|$.
\end{itemize}
As a consequence, we have that $Q'$ dominates $Q$ by the hypothesis.
\end{proof}

\begin{proof}~[Theorem~\ref{lem:tight:k=1}]
We first recall that a \emph{chain} of integer partitions is a sequence of integer partitions $Q_1, \ldots, Q_p$ such that for all $1 \leq i \leq p-1$, $Q_{i+1}$ dominates $Q_i$. Especially, a \emph{covering chain} is a chain such that for all $1 \leq i \leq p-1$, $Q_{i+1}$ covers $Q_i$. It has been proven in~\cite{Greene:1986:LCL:11002.11003} the maximum length of such a chain in the Dominance Lattice is exactly $2 \binom{m+1}{3} + mr$. Moreover, it should be clear that any chain with maximum length is a covering chain.  

Note that to any sequence of partitions of the nodes $P_1, P_2, \ldots, P_p$ that is obtained via $1$-deviations, we can associate a chain of integer partitions with the same length $p$, by Corollary~\ref{corollary:equiv-dominance-dev}.
Hence $L(1,n)$ is upper-bounded by the size of a longest chain of integer partitions. 
Conversely, any covering chain of integer partitions $Q_1, Q_2, \ldots, Q_p$ can be associated to such a sequence of partitions of the nodes.
As a consequence, that maximum length is also a lower-bound. \\
We get $L(1,n) = 2 \binom{m+1}{3} + mr$
\end{proof}

\begin{proof}~[Theorem~\ref{lem:tight:k=2}]
Clearly, $L(2,n) = \geq L(1,n)$. \\
Let $G^{\emptyset}= (V,E)$ be the conflict graph such that $|V|=n$ and $|E|=0$.
Let $P$ be any partition of the nodes for $G^{\emptyset}$ such that there exists $S= \{u,v\}$ that breaks $P$.
We denote by $X'$ the group in $P$ both vertices in $S$ agree to join together.
If $|X'| \geq |X(u)|$ or $|X'| \geq |X(v)|$, then the $2$-deviation can be decomposed into $1$-deviations.
So, we suppose that $|X'| =  |X(u)| - 1 = |X(v)| - 1$. There are two cases:

a) Suppose $X(u) = X(v)$. 
Then, after the $2$-deviation, we replace $X(u),X'$, by $X(u) \setminus \{u,v\}, X' \cup \{u,v\}$; that is we get two groups of size $|X(u)| - 2$, $|X(u)| + 1$, respectively, instead of two groups of size  $|X(u)| - 1$, $|X(u)|$, respectively. Thus, if a vertex of $X'$ breaks $P$ by joining group $X(u)$ we get the same partition vector.

b) Now suppose $X(u) \neq X(v)$.  
Then, after the $2$-deviation, we replace $X(u),X(v),X'$, by $X(u) \setminus \{u\},X(v) \setminus \{v\}, X' \cup \{u,v\}$; that is we get three groups of size $|X(u)| - 1$, $|X(u)| - 1$, $|X(u)| + 1$ , respectively, instead of three groups of size $|X(u)| - 1$, $|X(u)|$, $|X(u)|$, respectively. Thus, if a vertex of $X(v)$ breaks $P$ and joins group $X(u)$, then we also get the same partition vector.

Finally, any partition vector that is gotten from a $2$-deviation may also be gotten from a sequence of $1$-deviations.
Consequently, $L(2,n) = L(1,n) = \theta{(n\sqrt{n})}$.
\end{proof}

\begin{proof}~[Theorem~\ref{thm:lower-bound-3}]
Let us consider the conflict graph $G^{\emptyset}=(V,E)$ such that $|V|=n$ and $|E|=0$.
Recall that any partition $P$ is represented by $\Lambda=(\lambda_{n}, \ldots, \lambda_{i}, \ldots, \lambda_{1})$ where $\lambda_{i}$ represents the number of groups of size $i$,  $1 \leq i \leq n$.
We here only consider $3$-deviations that consist in creating $1$ group of size $p$, removing $3$ groups of size $p-1$, creating $3$ groups of size $p-2$, and removing $1$ group of size $p-3$, for some $p$ such that $4 \leq p \leq n/4$.
In other words, we focus on $3$-deviations from $P$ to $P'$ such that $\Lambda'-\Lambda=(\ldots,0,1,-3,3,-1,0,\ldots)$ where $\Lambda$ and $\Lambda'$ are the vectors of partitions $P$ and $P'$, respectively.
For any $p$, $4 \leq p \leq n/4$, we define the vector $\gamma[p]$ of size $n$ such that:
\begin{equation*}
\begin{cases}
\gamma[p]_{j} = 0 & j = p+1, \ldots n \\
\gamma[p]_{p} = 1 \\
\gamma[p]_{p-1} = -3 \\
\gamma[p]_{p-2} = 3 \\
\gamma[p]_{p-3} = -1 \\
\gamma[p]_{j} = 0 & j = 1, \ldots p-4.
\end{cases}
\end{equation*}

Given the vector $\Lambda$ of a partition $P$, the $3$-deviation considered here is such that $\Lambda' = \Lambda + \gamma[p]$, for some $p$, $4 \leq p \leq n/4$, where $\Lambda'$ represents the new partition after the deviation.

The number of nodes is assumed to be $n = O(cL^{2})$, where $c$ denotes a large constant integer.
We first do a sequence of $1$-deviations to get the partition $P$  such that $\Lambda=(0,\ldots,0,\lambda_{L}=c,\ldots,\lambda_{1}=c)$ represents the vector of $P$.

Let $t = \frac{L-1}{4}$.

We do $t+1$ consecutive $3$-deviations $\gamma[L], \ldots, \gamma[L-t]$ starting from $P$.
We get $\Lambda^{1} = \Lambda + \gamma^{1}[L]$ where $\Lambda^{1}$ is the vector of the resulting partition $P^{1}$ and $\gamma^{1}[L] = \sum_{i=L-t}^{L} \gamma[i]$ is such that:
\begin{equation*}
\begin{cases}
\gamma^{1}[L]_{j} = 0 & j = L+1, \ldots, n \\
\gamma^{1}[L]_{L} = 1 \\
\gamma^{1}[L]_{L-1} = -2 \\
\gamma^{1}[L]_{L-2} = 1 \\
\gamma^{1}[L]_{j} = 0 & j = L-t, \ldots, L-3 \\
\gamma^{1}[L]_{L-t-1} = -1 \\
\gamma^{1}[L]_{L-t-2} = 2 \\
\gamma^{1}[L]_{L-t-3} = -1 \\
\gamma^{1}[L]_{j} = 0 & j = 1, \ldots, L-t-4.
\end{cases}
\end{equation*}

We repeat the sequence defined by $\gamma^{1}$ starting from $L-1, L-2, \ldots, L-t+2$.
That corresponds to $t-1$ consecutive sequences $\gamma^{1}[L], \ldots, \gamma^{1}[L-t+2]$ starting from $P$.
We get $\Lambda^{2} = \Lambda + \gamma^{2}[L]$ where $\Lambda^{2}$ is the vector of the resulting partition $P^{2}$ and $\gamma^{2}[L] = \sum_{i=L-t+2}^{L} \gamma^{1}[i]$ is such that:
\begin{equation*}
\begin{cases}
\gamma^{2}[L]_{j} = 0 & j = L+1, \ldots, n \\
\gamma^{2}[L]_{L} = 1 \\
\gamma^{2}[L]_{L-1} = -1 \\
\gamma^{2}[L]_{j} = 0 & j = L-t+2, \ldots, L-2 \\
\gamma^{2}[L]_{L-t+1} = -1 \\
\gamma^{2}[L]_{L-t} = 1 \\
\gamma^{2}[L]_{L-t-1} = -1 \\
\gamma^{2}[L]_{L-t-2} = 1 \\
\gamma^{2}[L]_{j} = 0 & j = L-2t+1, \ldots, L-t-3 \\
\gamma^{2}[L]_{L-2t} = 1 \\
\gamma^{2}[L]_{L-2t-1} = -1 \\
\gamma^{2}[L]_{j} = 0 &  j = 1, \ldots, L-2t-2.
\end{cases}
\end{equation*}

We repeat the sequence $\gamma^{2}$ starting from $L-1, L-2, \ldots, L-t+4$.
That corresponds to $t-3$ consecutive sequences $\gamma^{2}[L], \ldots, \gamma^{2}[L-t+4]$ starting from $P$.
We get $\Lambda^{3} = \Lambda + \gamma^{3}[L]$ where $\Lambda^{3}$ is the vector of the resulting partition $P^{3}$ and $\gamma^{3}[L] = \sum_{i=L-t+4}^{L} \gamma^{2}[i]$ is such that:

\begin{equation*}
\begin{cases}
\gamma^{3}[L]_{j} = 0 & j = L+1, \ldots, n \\
\gamma^{3}[L]_{L} = 1 \\
\gamma^{3}[L]_{j} = 0 & j = L-t+4, \ldots, L-1 \\
\gamma^{3}[L]_{L-t+3} = -1 \\
\gamma^{3}[L]_{L-t+2} = -1 \\
\gamma^{3}[L]_{L-t+1} = 0 \\
\gamma^{3}[L]_{L-t} = -1 \\
\gamma^{3}[L]_{j} = 0 & j = L-3t+4, \ldots, L-t-1 \\
\gamma^{3}[L]_{L-3t+3} = 1 \\
\gamma^{3}[L]_{L-3t+2} = 0 \\
\gamma^{3}[L]_{L-3t+1} = 1 \\
\gamma^{3}[L]_{j} = 0 & j = L-4t+1, \ldots, L-3t \\
\gamma^{3}[L]_{L-4t} = -1 \\
\gamma^{3}[L]_{j} = 0 & j = 1, \ldots, L-4t-1.
\end{cases}
\end{equation*}

We finally repeat the sequence $\gamma^{3}$ from $L-1, L-2, \ldots, L-t+1$.
That corresponds to $t$ consecutive sequences $\gamma^{3}[L], \ldots, \gamma^{3}[L-t+1]$ starting from $P$.
We get $\Lambda^{4} = \Lambda + \gamma^{4}[L]$ where $\Lambda^{4}$ is the vector of the resulting partition $P^{4}$ and $\gamma^{4}[L] = \sum_{i=L-t+1}^{L} \gamma^{3}[i]$.
Note as $\gamma^{3}[i]$ has a constant number of non zero values, then $\gamma^{4}_{j}[L] \geq -c$ for any $j$, $1 \leq j \leq n$, where $c \geq 4$. So, the sequence is indeed valid when we start from $P$, which means we have enough groups of each size to do it.

Set $t = \theta(L)$ such that $t \leq \frac{L-1}{4}$.
We get that the total number of $3$-deviations is $\theta(t^{4}) = \theta(L^{4})$.
Recall that the number of nodes is $\theta(cL^{2})=\theta(L^{2})$.
Thus $L(3,n) = \Omega(n^{2})$.
\end{proof}

\paragraph{Proof of super-polynomial lower bound for k $\geq$ 4 (Thm~\ref{lem:lower-bound-4}).}

We first formalize the deviations we will use in the proof.
We only consider $4$-deviations that consist in creating $1$ group of size $p$, removing $4$ groups of size $p-1$, creating $4$ groups of size $p-2$, and removing $1$ group of size $p-4$, with $p$ such that $5 \leq p \leq n/5$.
For all $p$, $5 \leq p \leq n/5$, we define the vector $\delta[p]$ of size $n$ such that for all $j = 1, \ldots, n$,
$\delta[p]_{p} = 1$,
$\delta[p]_{p-1} = -4$
$\delta[p]_{p-2} = 4$
$\delta[p]_{p-4} = -1$,
and $\delta[p]_{j} = 0$ otherwise.
Given the vector $\Lambda$ of a partition $P$, the $4$-deviation considered here is such that $\Lambda' = \Lambda + \delta[p]$, for some $p$, $5 \leq p \leq n/5$, where $\Lambda'$ represents the vector of the new partition after the $4$-deviation defined by $\delta[p]=\Lambda'-\Lambda=(0,\ldots,0,1,-4,4,0,-1,0,\ldots,0)$.
In other words, we focus on $4$-deviations from $P$ to $P'$ such that $\Lambda'-\Lambda=(0,\ldots,0,1,-4,4,0,-1,0,\ldots,0)$.


We now formalize the $1$-deviations as follows.
Recall that a $1$-deviation consists in creating $1$ group of size $p$, removing $1$ group of size $p-1$, removing $1$ group of size $q+1$, and creating $1$ group of size $q$ (if $q > 0$), for any $p, q$ such that $p \geq q + 2 \geq 2$ and $p + q \leq n$.
Hence for any $p, q$ such that $p > q + 2 \geq 3$ and $p + q \leq n$, we define the vector $\alpha[p,q]$ of size $n$ such that for all $j = 1, \ldots n$,
$\alpha[p,q]_{p} = 1$,
$\alpha[p,q]_{p-1} = -1$,
$\alpha[p,q]_{q+1} = -1$
$\alpha[p,q]_{q} = 1$, 
and $\alpha[p,q]_{j} = 0$ otherwise.
For any $p, q$ such that $p = q + 2 \geq 3$ and $p + q \leq n$, we define the vector $\alpha[p,q]=\alpha[p,p-2]$ of size $n$ such that for all $j = 1, \ldots, n$,
$\alpha[p,q]_{p} = 1$,
$\alpha[p,q]_{p-1} = -2$,
$\alpha[p,q]_{p-2=q} = 1$,
and $\alpha[p,q]_{j} = 0$ otherwise.
Given the vector $\Lambda$ of a partition $P$, the $1$-deviation considered above is such that $\Lambda' = \Lambda + \alpha[p,q]$, for some $p, q$ such that $p \geq q + 2 \geq 3$ and $p + q \leq n$, where $\Lambda'$ represents the vector of the new partition $P'$ after the $1$-deviation defined by $\alpha[p,q]$.
In other words, the $1$-deviation from $P$ to $P'$ is such that $\Lambda'-\Lambda=(0,\ldots,0,1,-1,0, \ldots, 0,-1,1,0,\ldots,0)$ or $\Lambda'-\Lambda=(0,\ldots,0,1,-2,1,0,\ldots,0)$.

The definition above does not take into account the special case when $q=0$.
So, for any $p$ such that $p \geq 3$ and $p < n$, we define the vector $\alpha[p]$ of size $n$ such that for all $j = 1, \ldots, n$,
$\alpha[p]_{p} = 1$,
$\alpha[p]_{p-1} = -1$,
$\alpha[p]_{1} = -1$,
and $\alpha[p]_{j} = 0$ otherwise.
We define the vector $\alpha[2]$ of size $n$ such that for all $j = 1, \ldots, n$,
$\alpha[2]_{2} = 1$,
$\alpha[2]_{1} = -2$.
and $\alpha[2]_{j} = 0$ otherwise.
Given the vector $\Lambda$ of a partition $P$, the $1$-deviation considered above is such that $\Lambda' = \Lambda + \alpha[p]$, for some $p$, $2 \leq p < n$, where $\Lambda'$ represents the vector of the new partition after the $1$-deviation defined by $\alpha[p]$.
In other words, the $1$-deviation from $P$ to $P'$ considered here is such that $\Lambda'-\Lambda=(0,\ldots,0,1,-1,0, \ldots, 0,-1)$ or $\Lambda'-\Lambda=(0,\ldots,0,1,-2)$.

For any positive integers $p, q, d$ such that $p-d \geq q+d$ and $p + q \leq n$, we define the vector $\alpha[p,p-d,q+d,q]=\sum_{j=0}^{d-1} \alpha[p-j,q+j]$.
If $p-d > q+d$, we get for all $j = 1, \ldots, n$:
\begin{equation*}
\begin{cases} 
\alpha[p,p-d,q+d,q]_{p} = 1,\\
\alpha[p,p-d,q+d,q]_{p-d} = -1,\\
\alpha[p,p-d,q+d,q]_{q+d} = -1,\\
\alpha[p,p-d,q+d,q]_{q} = 1,\\
\text{and }\alpha[p,p-d,q+d,q]_{j} = 0\text{ otherwise}.
\end{cases}
\end{equation*}
If $p-d = q+d$, we get for all $j = 1, \ldots, n$:
\begin{equation*}
\begin{cases}  
\alpha[p,p-d,q+d,q]_{p} = 1,\\
\alpha[p,p-d,q+d,q]_{p-d=q+d} = -2,\\
\alpha[p,p-d,q+d,q]_{q} = 1,\\
\text{and }\alpha[p,p-d,q+d,q]_{j} = 0\text{ otherwise}.
\end{cases}
\end{equation*}

We are now ready to present a sequence for the coloring game with a superpolynomial size.
Without loss of generality, we assume that the number of nodes is $n = cL(L+1)/2$.
The values $c$ and $L$ will be defined later and we assume for the moment that $c$ is sufficiently large.
 If in truth $cL(L+1)/2 < n < (c+1)L(L+1)/2$, then we only consider $cL(L+1)/2$ nodes.
More precisely, the $n - cL(L+1)/2$ other nodes will never move and will stay in their $n - cL(L+1)/2$ respective singleton groups.
Thus we can only consider the vectors and partitions induced by the $cL(L+1)/2$ considered nodes.
We first do a sequence of $1$-deviations to get the partition $P^{0}$ such that $\Lambda^{0}=(0,\ldots,0,\lambda_{L}=c,\ldots,\lambda_{1}=c)$ is the  vector of $P^{0}$.

Using our notations, we observe that to get the partition $P^{0}$ we do the sequence defined by $\sum_{j=2}^{L} \alpha[j]$ $c$ times.
Thus, one can check that $\Lambda^{0} = (0,\ldots,0,n) + c\sum_{j=2}^{L} \alpha[j]$.

In the sequel, we will denote any sequence of deviations by $\phi[l]$, where $l$ denotes the greatest index $j$ such that $\phi_j \neq 0$.
Observe that given any sequence of deviations defined by a vector $\phi[l]$, the vector $\phi[l-i]$, $i > 0$, represents the same sequence of deviations in which any $4$-deviation $\delta[p]$ is replaced by $\delta[p-i]$ and any $1$-deviation involving groups of sizes $p$ and $q$, now involves groups of sizes $p-i$ and $q-i$.

Let $t > 0$ and $T > 0$ be such that $2^{T-1} (2t^{2}+2) \leq L$.

\begin{claim}
\label{claim:premiere-sequence}
There exists a sequence of $1$-deviations and $4$-deviations from $P^{0}$ to a partition $P^{1}$ with vector $\Lambda^{1}$, where $\zeta^{1}[L] = \Lambda^{1}-\Lambda^{0}$ is such that for all $j = 1, \ldots, n$,
$\zeta^{1}[L]_{L} = 1$,
$\zeta^{1}[L]_{L-2} = -1$,
$\zeta^{1}[L]_{L-3} = -1$,
$\zeta^{1}[L]_{L-t^{2}} = 1$,
$\zeta^{1}[L]_{L-t^{2}-1} = 1$,
$\zeta^{1}[L]_{L-2t^{2}+2} = -1$,
$\zeta^{1}[L]_{L-2t^{2}+1} = -1$,
$\zeta^{1}[L]_{L-2t^{2}-1} = 1$,
and $\zeta^{1}[L]_{j} = 0$ otherwise.
\end{claim}

\begin{proofclaim}~[Claim~\ref{claim:premiere-sequence}]
Set $\phi[L] = (\sum_{j=0}^{t^{2}-1}\delta[L-j]) + \alpha[L-1,L-2,L-2,L-3] + \alpha[L-t^{2},L-t^{2}-1,L-t^{2}-1,L-t^{2}-2] + \alpha[L-t^{2}-3,L-t^{2}-4,L-t^{2}-5,L-t^{2}-6] + \alpha[L-1,L-3,L-3,L-5]$. 
Recall it means we do $t^2$ consecutive $4$-deviations, namely a cascade, then we adjust the resulting partition using $4$ sequences of $1$-deviations.
We get:
\begin{equation*}
\begin{cases}
\phi[L]_{j} = 0 & j = L+1, \ldots, n \\
\phi[L]_{L} = 1 \\
\phi[L]_{L-1} = -1 \\
\phi[L]_{L-2} = -1 \\
\phi[L]_{j} = 0 & j = L-3,L-4 \\
\phi[L]_{L-5} = 1 \\
\phi[L]_{j} = 0 & j = L-t, \ldots, L-6 \\
\phi[L]_{L-t^{2}-1} = 1 \\
\phi[L]_{j} = 0 & j = L-t^{2}-3,L-t^{2}-2 \\
\phi[L]_{L-t^{2}-4} = -1 \\
\phi[L]_{L-t^{2}-5} = -1 \\
\phi[L]_{L-t^{2}-6} = 1 \\
\phi[L]_{j} = 0 & j = 1, \ldots, L-t^{2}-7.
\end{cases}
\end{equation*}

We can then construct $\zeta^{1}_{j}[L] = (\sum_{j=0}^{t^{2}-5} \phi[L-j]) + \alpha[L-4,L-t^{2}+1,L-t^{2}-2,L-2t^{2}+3] + \alpha[L-t^{2}+4,L-t^{2}+2,L-t^{2}-3,L-t^{2}-5]$.
\end{proofclaim}

\begin{claim}
\label{claim:c-0-balanced}
There is a constant $c_{1}$ such that the sequence from $P^{0}$ to $P^{1}$ defined by $\zeta^{1}[L]$ in Claim~\ref{claim:premiere-sequence}, is $c_{1}$-balanced.
\end{claim}

\begin{proofclaim}~[Claim~\ref{claim:c-0-balanced}]
First, note that any $k$-deviation, $k \leq 4$, changes a constant number of groups in the partition.
Thus, for any $j$, $0 \leq j \leq t^{2}-5$, the sequence defined by $\phi[L-j]$ is balanced for some constant because, for any $i$, $1 \leq i \leq n$, the number of deviations that change groups of size $i$ is constant.
Furthermore, the vector $\phi[L-j]$ contains a constant number of non-zero values.
Thus for any $i$, $1 \leq i \leq n$, the sequence defined by $\sum_{j=0}^{t^{2}-5} \phi[L-j]$ changes a constant number of groups of size $i$.
We conclude there exists a constant $c_{1}$ such that the sequence induced by $\zeta^{1}[L]$ is $c_{1}$-balanced.
\end{proofclaim}

\begin{proofclaim}~[Claim~\ref{claim:symmetry}]
The vector $\phi'$ has size $s' = (r-1)d+s$.
Let us first suppose that $r$ is even.
$\phi' = \sum_{h=0}^{r/2-1} \phi[L- h d] +  \sum_{h=r/2}^{r-1} \phi[L - h d]$.
We prove that, for any $r'$, $0 \leq r' \leq r/2-1$, $\phi^{r'} = \sum_{h=0}^{r'} \phi[L - h d] +  \sum_{h=r-1-r'}^{r-1} \phi[L - h d]$ has the symmetric property.
Note that the size of $\phi^{r'}$ is $s' = (r-1)d+s$ for any $r'$, $0 \leq r' \leq r/2-1$.

By induction on $r'$.
Suppose $r'=0$.
We have $\phi^{0} = \phi[L]+\phi[L-(r-1)d]$.
For any $j$, $L-(r-1)d-s+1 \leq j \leq L$, $\phi^{0}_{j}=\phi_{j}+\phi_{j+(r-1)d}$.
As $\phi$ has the symmetric property and has size $s$, $\phi^{0}_{j}=\phi_{L-s+1+L-j} + \phi_{L-s+1+L-j-(r-1)d}$.
Furthermore, $\phi^{0}_{L-s'+1+L-j} = \phi_{L-s'+1+L-j} + \phi_{L-s'+1+L-j+(r-1)d}$.
Recall that $s' = (r-1)d+s$.
Thus, $\phi^{0}_{L-s'+1+L-j} = \phi_{L-(r-1)d-s+1+L-j} + \phi_{L-s+1+L-j}$.
We get that, for any $j$, $L-s'+1 \leq j \leq L$, $\phi^{0}_{j} = \phi^{0}_{L-s'+1+L-j}$, and so $\phi^{0}$ has the symmetric property.

Suppose it is true for $r'$, $0 \leq r' \leq r/2 - 1$, we prove it is also true for $r'+1$.
By the induction hypothesis, vector $\phi^{r'}$ has the symmetric property.
We have $\phi^{r'+1} = \phi^{r'} + \phi[L-(r'+1) d] + \phi[L-(r-r'-2) d]$.
First, if $j > L-(r'+1)d$, then $\phi^{r'+1}_{j} - \phi^{r'}_{j} = 0$ by construction.
Also, $\phi^{r'+1}_{j}-\phi^{r'}_{j} = 0$ for any $j$, $L-s'+1 \leq j \leq L-(r-r'-2) d-s$.
For any $j$, $L-(r-r'-2) d-s+1 \leq j \leq L-(r'+1)d$, $\phi^{r'+1}_{j}-\phi^{r'}_{j} = \phi_{j-(r'+1)d}+\phi_{j-(r-r'-2) d}$.
We apply the proof of case $r'=0$ to show that $\phi[L-(r'+1) d] + \phi[L-(r-r'-2) d]$ has the symmetric property.
To conclude, it is sufficient to show that $L- (L-(r-r'-2) d-s+1) = L-(r'+1)d - (L-s'+1)$, that is the sub-vector of non-zero values of $\phi[L-(r'+1) d] + \phi[L-(r-r'-2) d]$ is in the middle of $\phi'$.
Indeed, $(r-r'-2) d+s-1) = -(r'+1)d +s'-1$, and so $(r-r'-2) d+s-1) = -(r'+1)d +s'-1$.
As $s'=(r-1)d+s$, $rd-r'd-2d+s-1=-r'd-d+rd-d+s-1$, and so $L- (L-(r-r'-2) d-s+1) = L-(r'+1)d - (L-s'+1)$.
Thus, $\phi'$ has the symmetric property.

Finally, for any $r'$, $0 \leq r' \leq r/2-1$, $\phi^{r'}$ has the symmetric property.
We get that $\phi'$ has the symmetric property too.

Consider now that $r$ is odd.
Using the previous induction, we prove that the vector $\sum_{h=0}^{(r-1)/2-1} \phi[L-h d] +  \sum_{h=(r-1)/2+1}^{r-1} \phi[L-h d]$ has the symmetric property.
We now prove that $\phi' = \sum_{h=0}^{(r-1)/2-1}  \phi[L-h d] +  \sum_{h=(r-1)/2+1}^{r-1}  \phi[L-h d] + \phi[L-d(r-1)/2]$ has the symmetric property.
To do that we prove that the non-zero part of the vector $\phi[L-d(r-1)/2]$ is exactly in the middle of the vector $\sum_{h=0}^{(r-1)/2-1}  \phi[L-h d] +  \sum_{h=(r-1)/2+1}^{r-1}$.
In other words, the non-zero sub-vector of $\phi[L-d(r-1)/2]$ starts at index $L-d(r-1)/2$ and ends at index $L-d(r-1)/2-s+1$, and we show that $L-(L-d(r-1)/2) = (L-d(r-1)/2-s+1)-(L-s'+1)$.
Indeed, $d(r-1)/2 = -d(r-1)/2-s+s'$, and so $d(r-1)=s'-s$.
Recall that $s' = (r-1)d+s$.
Thus, $\phi'$ has the symmetric property.
\end{proofclaim}

\begin{proofclaim}~[Claim~\ref{claim:sequence-generale}]
By induction on $i$.
The sequence $\zeta^{1}$ defined in Claim~\ref{claim:premiere-sequence} has the Good Property.
Consider $i$, $1 \leq i \leq T-1$.
Suppose $\zeta^{i}[L]$ has the Good Property.
We prove that there exists a sequence of $1$-deviations $\xi^{i+1}$ such that $\zeta^{i+1}[L] = \sum_{j=0}^{a_{1}} \zeta^{i}[L-j t^{1}_{i}] + \xi^{i+1}$ has the Good Property.

Set $\varphi^{i+1} = \sum_{j=0}^{a_{i}} \zeta^{i}[L-j t^{1}_{i}]$.
By construction, for all $j$, $L-t^{2}_{i+1} \leq j \leq L$, we get 
$\varphi^{i+1}_{L} = 1$, $\varphi^{i+1}_{L-t^{1}_{i+1}} = -1$, $\varphi^{i+1}_{L-t^{2}_{i+1}} = -1$, and $\varphi^{i+1}_{j} = 0$ otherwise, where $t^{1}_{i+1}=t^{2}_{i}$ and $t^{2}_{i+1}=2t^{1}_{i}$.
We get $1 < t^{1}_{i+1} < t^{2}_{i+1}$ because $1 < t^{2}_{i} < 2t^{1}_{i}$ by the  induction hypothesis.
Also, $t^{2}_{i+1} < 2 t^{1}_{i+1}$ because $2t^{1}_{i} < 2 t^{2}_{i}$.
Finally, $t^{2}_{i+1} \leq 2^{i+1}$ because $t^{2}_{i+1} < 2t^{1}_{i+1} = 2t^{2}_{i} \leq 2^{i}2$.

For now we abuse the notation setting $s_{i+1}$ the size of the vector $\varphi^{i+1}$.
We will prove later that $\varphi^{i+1}$ and $\zeta^{i+1}$ have same size.
Furthermore $s_{i+1}=a_{i}t^{1}_{i}+s_{i}$.
As $a_{i}$ and $s_{i}$ are even integers, then $s_{i+1}$ is even.
Note that $s_{i+1} \leq 2 s_{i}$.

By Claim~\ref{claim:symmetry}, $\varphi^{i+1}$ has the symmetric property: $\varphi^{i+1}_{L-j}=\varphi^{i+1}_{L-s_{i+1}+1+j}$ for all $j$, $0 \leq j \leq s_{i+1}-1$.

Moreover, by the choice of $a_{i}$, we get that:
\begin{equation*}
\varphi^{i+1}_{j} \in 
\begin{cases}
\{-1,0\} & \mbox{if $L-\frac{s_{i}}{2}+2 \leq j \leq L-t_{i+1}^{2}-1$} \\
\{0,1\} & \mbox{if $L-s_{i}+1+t_{i}^{1} \leq j \leq L-\frac{s_{i}}{2}+1$} \\
\{-1,0\} & \mbox{if $L-s_{i+1}+1+t_{i+1}^{1} \leq j \leq L-s_{i}+t_{i}^{1}$.}
\end{cases}
\end{equation*}

\noindent
Furthermore:
$\sum_{j=L-\frac{s_{i}}{2}+2}^{L-t_{i+1}^{2}-1} \varphi^{i+1}_{j}
+\sum_{j=L-s_{i+1}+1+t_{i+1}^{1}}^{L-s_{i}+t_{i}^{1}} \varphi^{i+1}_{j} \\ = -\sum_{j=L-s_{i}+1+t_{i}^{1}}^{L-\frac{s_{i}}{2}+1} \varphi^{i+1}_{j} -2$.




We now show that $\varphi^{i+1}_{L-s_{i+1}/2}=\varphi^{i+1}_{L-s_{i+1}/2+1}=1$.
As $\zeta^{i}[L]_{L-s_{i}/2+1} = 1$, then $\varphi^{i+1}_{L-s_{i}/2+\frac{a_{i}}{2}(t^{1}_{i}-1)+1}=1$.
Also, $\varphi^{i+1}_{L-s_{i}/2+\frac{a_{i}}{2}(t^{1}_{i}-1)+2}=1$.
As $a_{i}$ is even, we get $\varphi^{i+1}_{L-s_{i+1}/2+1}=1$ because $\frac{s_{i}}{2}+\frac{a_{i}}{2}=\frac{s_{i+1}}{2}$.
Recall that $s_{i+1}=a_{i}(t^{1}_{i}-1)+s_{i}$.
Thus, $\varphi^{i+1}_{L-s_{i+1}/2}=1$.

We now prove that there exists a sequence of $1$-deviations $\xi^{i+1}$ such that $\zeta^{i+1}[L] = \varphi^{i+1} + \xi^{i+1}$ has the Good Property.
Consider any $j_{1}, j_{2} > 0$, $L-\frac{s_{i}}{2}+2 \leq j_{1} \leq L-t_{i+1}^{2}-1$, $L-\frac{s_{i+1}}{2}+2 \leq j_{2} \leq L-\frac{s_{i}}{2}+1$, such that $\varphi^{i+1}_{j_{1}} = -1$, $\varphi^{i+1}_{j_{2}} = 1$.
As $\varphi^{i+1}$ has the symmetric property (Claim~\ref{claim:symmetry}), then there exists a $1$-deviation defined by $\alpha[j_{1},j_{2},j_{3},j_{4}]$, with $j_{1}-j_{2} = j_{3}-j_{4}$, such that $\varphi^{i+1}_{j_{1}} = \varphi^{i+1}_{j_{4}} = - 1$ and $\varphi^{i+1}_{j_{2}} = \varphi^{i+1}_{j_{3}} = 1$.

This deviation clearly keeps the symmetric property.
We can repeat that, with a different $4$-tuple of indices, until having the Good Property and so the vector $\zeta^{i+1}[L]$.
\end{proofclaim}

\begin{proofclaim}~[Claim~\ref{claim:c-balanced}]
We prove the result by induction on $i$.
Claim~\ref{claim:c-0-balanced} proves the result for $i=1$.
Suppose it is true for all $i' \leq i \leq T-1$.
We prove that it is true for $i+1$.
Recall that we first build $\varphi^{i+1} = \sum_{j=0}^{a_{i}} \zeta^{i}[L-j t^{1}_{i}]$ where $a_{i}$ is the largest even integer $j$ such that $L-jt^{1}_{i}-t^{2}_{i} > L-s_{i}/2+1$.

For all $b$, $0 \leq b \leq a_{i}$, consider $\mu^{b} = \sum_{j=0}^{b} \zeta^{i}[L-j t^{1}_{i}]$.
We prove that the sequence of deviations from $\zeta^{i}$ to $\mu^{b}$ is $(c_{1}+i)$-balanced.
By induction on $b$.
It is true for $b=0$.
Indeed by the first induction hypothesis, the sequence $\zeta^{i}[L]$ is $(c_{1}+i-1)$-balanced.
Suppose it is true for all $b' \leq b \leq a_{i}-1$.
We prove it is true for $b+1$.
By Claim~\ref{claim:sequence-generale}, $\mu^{b}$ is $1$-balanced because $\zeta^{i}[L]$ is $1$-balanced (the values of the vector belong to the set $\{-1,0,1\}$).
By definition $\mu^{b+1}=y^{b}+\zeta^{i}[L-(b+1)t^{1}_{i}]$.
By the first induction hypothesis, each sequence of $\zeta^{i}[L-(b+1)t^{1}_{i}]$ is $(c_{1}+i-1)$-balanced.
Since $\mu^{b}$ is $1$-balanced, the sequence from $\mu^{b}$ to $\mu^{b+1}$, is $(c_{1}+i)$-balanced.
Thus, $\varphi^{i+1}$ is $(c_{1}+i)$-balanced.

To conclude, it remains to prove that the sequence from $\varphi^{i+1}$ to $\zeta^{i+1}$, is $(c_{1}+i)$-balanced.
Consider the last $1$-deviations described in the proof of Claim~\ref{claim:sequence-generale}.
Every deviation consists in replacing values $-1$ and $1$ by $0$.
Thus, $\zeta^{i+1}$ is $(c_{1}+i)$-balanced.
\end{proofclaim}

\begin{proofclaim}~[Claim~\ref{claim:taille-sequence-1}]
$|S_{i+1}| = a_{i} |S_{i}|$ with $a_{i}$ is the largest even integer $j$ such that $L-jt^{1}_{i}-t^{2}_{i} > L-s_{i}/2+1$.
As $t^{2}_{i} \leq 2^{i+1}$, $a_{i}$ is the largest even integer such that $j < \frac{s_{i}}{2^{i+2}} - \frac{1}{2^{i}} - 1$.
We get $a_{i} \geq \lfloor \frac{s_{i}}{2^{i+2}} - \frac{1}{2^{i}} - 1 \rfloor -2$, and so $|S_{i+1}| \geq (\frac{s_{i}}{2^{i+2}}-6) |S_{i}|$.
\end{proofclaim}

\begin{proofclaim}~[Claim~\ref{claim:taille-sequence-11}]
By Claim~\ref{claim:taille-sequence-1}, for all $i$, $1 \leq i \leq T-1$, $|S_{i+1}| \geq (\frac{s_{i}}{2^{i+2}}-6) |S_{i}|$.
As $s_{i} \leq s_{i+1}$ for all $i$, $1 \leq i \leq T-1$, then $|S_{T}| \geq (\frac{s_{1}}{2^{T+1}}-6)^{T-1} |S_{1}|$.
Thus, $|S_{T}| \geq (\frac{2t^{2}+2}{2^{\log_{2}(t)+2}}-6)^{\log_{2}(t)} |S_{1}|$, and so $|S_{T}| \geq O(t^{log_{2}(t)}) |S_{1}|$.
\end{proofclaim}

\subsection{Appendix 2: proofs of Section~\ref{sec:weighted}}

We start this section with a tool to lower-bound the variations of the global utility in a more general setting:
\begin{lemma}
\label{lem:util-variations}
Given a graph $G=(V,w)$, let $P$ be a partition of the nodes, let $S = \{u_1, \ldots, u_k\}$ be a $k$-set that breaks $P$, and let $P'$ be the resulting partition.
Then:

\centering
$f(P') - f(P) \geq 2[ k - \sum_{1\leq i,j \leq k}w_{u_i u_j} + \sum_{\{u_i,u_j \mid X(u_i) = X(u_j)\}}w_{u_i u_j}]$. 
\end{lemma}

\begin{proof}
In the sequel, we will call $X'$ the group in $P$ that all the vertices in $S$ agree to join with.
As $S$ breaks the partition $P$, the utility of each vertex increases by at least $1$.
So, we get as the variation of the utility for the $k$-set: $\sum_{i=1}^k [f_{u_i}(P') - f_{u_i}(P)] \geq k$.

For every $1 \leq i \leq k$, we then define $\delta_i = \sum_{v \in X(u_i) \setminus S}w_{u_iv}$ and $\sigma_i = \sum_{u_j \in X(u_i)} w_{u_ju_i}$;
we define $\delta_i' = \sum_{v \in X'}w_{u_iv}$ and $\sigma_i' = \sum_{j=1}^k w_{u_iu_j}$ in a similar way.
Then, $f_{u_i}(P) = \delta_i + \sigma_i$ and $f_{u_i}(P') = \delta_i' + \sigma_i'$; so, we get by summation that
$\sum_{i=1}^k f_{u_i}(P)  = \sum_{i=1}^k \delta_i + 2\sum_{\{u_i,u_j \mid X(u_i) = X(u_j)\}}w_{u_iu_j}$,\\
$\sum_{i=1}^k f_{u_i}(P') = \sum_{i=1}^k \delta_i' + 2\sum_{1\leq i, j \leq k}w_{u_iu_j}$.\\
Note that we have a factor $2$ for any occurence of $w_{u_iu_j}$, as it is counted once for $u_i$ and once for $u_j$.

Furthermore, the variation of the global utility includes that of the nodes in $S$, that of the nodes in $X'$, plus that of the nodes in $X(u_i)\setminus S$ for every $1\leq i \leq k$.
In other words, we get by symmetry that:\\ 
$f(P') - f(P) = \sum_{i=1}^k [f_{u_i}(P') - f_{u_i}(P)] + \sum_{i=1}^k [\delta_i' - \delta_i]$\\
$= 2\sum_{i=1}^k [f_{u_i}(P') - f_{u_i}(P)] - 2\sum_{1\leq i, j \leq k}w_{u_iu_j} +  2\sum_{\{u_i,u_j \mid X(u_i) = X(u_j)\}}w_{u_iu_j}$\\
$\geq 2k - 2\sum_{1\leq i, j \leq k}w_{u_iu_j} +  2\sum_{\{u_i,u_j \mid X(u_i) = X(u_j)\}}w_{u_iu_j}$.
\end{proof} 

\begin{proof}~[Theorem~\ref{thm:stability:k=1}]
Let $P,P'$ be two partitions of the nodes for the same graph $G=(V,w)$.
We assume that there is a user $u$ such that $\{u\}$ breaks $P$, and the resulting partition is $P'$.
Then, by Lemma~\ref{lem:util-variations}, we get that $f(P') - f(P) \geq 2$.
So, the global utility increases at each step of the game, and it is upper bounded by $O(w_pn^2)$.
\end{proof}

\subsection{Proof of Corollary~\ref{lem:counter-example-2}}

Let us first show that $k(\mathcal{W}) \geq 2$.
In fact, Theorem~\ref{thm:stability:ktree} is a global stability result, which is more precise and uses structural properties of the graphs.
Given a graph $G=(V,w)$, let us define the "friendship graph" $G^{+} = (V,E^{+})$ of  $G$, where $E^{+} = \{ uv \in E: w_{uv} > 0\}$.
We remind that the \emph{girth} of a graph is the length of its shortest cycle.
By definition, an acyclic graph has infinite girth.
\begin{theorem}
\label{thm:stability:ktree}
Given an integer $k \geq 1$ and a graph $G = (V,w)$ with $\mathcal{W} \subseteq \{-\infty,0,1\} \cup -\mathbb{N}$, there exists a $k$-stable partition for $G$ if the girth of the friendship graph $G^{+}$ is at least $k+1$.
Furthermore, the system always reaches a $k$-stable partition in $O(n^2)$ steps in that case.
\end{theorem}

\begin{proof}
As in the case $k=1$, we will prove the global utility strictly increases after any $k$-deviation.
Let $P$ be any partition for $G$, such that there is no vertices $u$ and $v$ with $w_{uv} = -\infty$ and $X(u) = X(v)$.
If $P$ is $k$-stable, we are done.
Otherwise there exists a subset $S = \{u_1,u_2, \ldots, u_{|S|}\}$, $|S| \leq k$, such that $S$ breaks $P$.
Let $P'$ be the resulting partition for $G$ after the $k$-deviation.
By Lemma~\ref{lem:util-variations} we have:
 $f(P') - f(P) \geq 2[ |S| - \sum_{\{u_i,u_j \mid X(u_i) \neq X(u_j)\}}w_{u_iu_j}]$\\
$\geq  2[ |S| - \sum_{\{u_i,u_j \mid X(u_i) \neq X(u_j)\}}\max{(0,w_{u_iu_j})}]$\\
$\geq  2[ |S| - |E^+ \cap S \times S|]$\\
$\geq 2$, because the girth is at least $k+1$ by the hypothesis.
\end{proof}

Particularly, if $G^{+}$ is cycle-free, then we get there is a $k$-stable partition for $G$, for any $k \geq 1$;
if $G^{+}$ is triangle-free, then there always exists a $3$-stable partition for $G$.
Furthermore, as the girth of any graph is at least $3$, then there always exists a $2$-stable partition. 
Theorem~\ref{thm:stability:ktree} also shows that such a coloring game (with \emph{fixed} set of weights $\mathcal{W}$) takes at most $O(n^2)$ steps to converge. 
Using another sequence construction, one can prove that this bound is indeed tight in the general case.
\begin{lemma}
\label{lem:n-square}
There exists $G=(V,w)$ with $\mathcal{W} = \{0,1\}$ such that there is a sequence of $\theta(n^2)$ $1$-deviations for $G$.
\end{lemma}

\begin{proof}
Let any $p \geq 1$.
We build $G$ as follows.
Without loss of generality assume $n=3p$.
Let $V=V_{1} \cup V_{2} \cup V_{3}$ with $|V_{1}|=|V_{2}|=|V_{3}|=p$.

For any $u, v \in V_{2} \cup V_{3}$, $w_{uv} = 1$.
For any $u, v \in V_{1}$, $w_{uv} = 0$.
For any $u \in V_{1}$, for any $v \in V_{2}$, $w_{uv} = 1$.

We now set $V_{1}=\{u_{1}, u_{2}, \ldots, u_{p}\}$ and $V_{2}=\{v_{1}, v_{2}, \ldots, v_{p}\}$.
Consider the partition $P$ such that any node of $V_{1} \cup V_{2}$ forms a singleton group and there is group formed by all the nodes of $V_{3}$.

Sequentially, each node of $V_{1}$ reaches the group of $v_{1}$.
The utility of $v_{1}$ is now $p$.
Then $v_{1}$ reaches the group composed of all nodes of $V_{3}$.
The utility of $v_{1}$ is now $p+1$, and the utility of each $u_{i}$ is $0$.
The number of $1$-deviations is $p+1$.
We repeat the same process for $v_{2}$ using $p+1$ $1$-deviations.
And so on for each $v_{i}$, $3 \leq i \leq p$.

Finally, the number of $1$-deviations is $(p+1)p = \theta(n^2)$.
\end{proof}

As it was announced, the combination of Theorem~\ref{thm:stability:ktree} and Lemma~\ref{lem:n-square} yields a proof of Lemma~\ref{lem:2-stable}.
We now define the notion of twin vertices and we prove a useful lemma.

\begin{definition}
Let $G=(V,w)$ be a graph, and $u,u' \in V$ be two users.
We say that $u$ and $u'$ are \emph{twins} if and only if $w_{uu'} > 0$ and for all vertex $v \in V \setminus \{u,u'\}$ $w_{uv} = w_{u'v}$. 
\end{definition}

\begin{lemma}
\label{lemma:group}
Given a graph $G$, and a $1$-stable partition $P$ for $G$, then $X(u) = X(u')$ for all vertices $u,u'$ such that $u$ and $u'$ are twins.
\end{lemma}

\begin{proof}
By contradiction.
Suppose $X(u) \neq X(u')$.
By symmetry, we can assume $f_u(P) \leq f_{u'}(P)$.
Thus, $\{u\}$ breaks the partition, as vertex $u$ joins the group of $u'$ because, in this new partition $P'$ for the nodes, we get $f_u(P') = f_{u'}(P) + w_{uu'} > f_u(P)$.
A contradiction.
Hence $X(u) = X(u')$.
\end{proof}

Finally, we show that there exists a graph $G$ such that any partition for $G$ is not $3$-stable ($k(\{-\infty,0,1\}) \leq 2$).
The proof is all the more counter-intuitive as it implies there are infinite sequences of deviations that do not change the global utility, even in graphs with only four edges weighted by $0$!

\begin{proof}~[Lemma~\ref{lem:counter-example}]
The set of vertices consists of four sets $A_i$, $0 \leq i \leq 3$, each of equal size $h \geq 2$ and with a special vertex $a_i$, plus four vertices $b_i$, $0 \leq i \leq 3$, and two vertices $c_0$ and $c_1$.
In what follows, indices are taken modulo $2$ for $c_j$, $j \in \{0,1\}$, and they are taken modulo $4$ everywhere else.
Figure~\ref{fig:counter-examples} represents the example with $h=3$. 

The friendship graph $G^+$ here consists of all the edges with weight $1$; the set $E^+$ contains:
\begin{enumerate}
\item all the edges between nodes in $A_i$ ($0\leq i \leq 3$);
\item edges between $b_i$ and $A_i$ ($0\leq i \leq 3$);
\item edges between $b_i$ and $A_{i+1} \setminus \{a_{i+1}\}$ ($0\leq i \leq 3$);
\item edges between $b_i$ and $b_{i-1}$ and $b_{i+1}$ ($0\leq i \leq 3$);
\item edges between $c_0$ and all the $b_i$, and edges between $c_1$ and all the $b_i$;
\item edges between $c_0$ and $A_0 \cup A_2$, and edges between $c_1$ and $A_1 \cup A_3$.
\end{enumerate} 
Moreover, there are four edges with weight $0$, namely the edges $b_i a_{i+1}$.
All the remaining edges have weight $-\infty$. 
That is two nodes in different $A_i,A_{i'}$ are enemies; a user $b_i$ is enemy of $b_{i+2}$ and of the nodes in $A_{i+2}$ and $A_{i+3}$; $c_0$ and $c_1$ are enemies; $c_0$ is enemy of the nodes in $A_1$ and $A_3$, and $c_1$ is enemy of the nodes in $A_0$ and $A_2$. 

We now assume there exists a $3$-stable partition $P$ for $G$.

\begin{claim}
\label{claim:subgrp}
Every $A_i$ is a subgroup in $P$.
\end{claim}

\begin{proofclaim}~[Claim~\ref{claim:subgrp}]
First, all the nodes in $A_i \setminus \{a_i\}$ are in the same group by Lemma~\ref{lemma:group}, as all the vertices in this subset are pairwise twins. 
Now assume that $a_i$ is not in the same group as $A_i \setminus \{a_i\}$.
Note that the nodes in $A_i$ have the same enemies, they all have $b_i$ and $c_i$ as friends (recall that the index for $c_i$ is taken modulo $2$), but the nodes in $A_i \setminus \{a_i\}$ also have $b_{i-1}$ as a friend whereas $w_{a_ib_{i-1}} = 0$.  
So, we have $X(a_i) \subseteq \{a_i,b_i,b_{i-1},c_i\}$.
This implies all the nodes in $X(a_i)$ have at least one friend and no enemy in the group in $P$ that contains $A_i \setminus \{a_i\}$.

If $|X(a_i)| \leq 3$ we can merge this group with the one that contains $A_i \setminus \{a_i\}$ using a $3$-deviation.
Otherwise, $P$ contains groups $\{a_i,b_i,b_{i-1},c_i\}$ and $A_i \setminus \{a_i\}$.
If $h \geq 4$, then $|A_i \setminus \{a_i\}| \geq 3$ and so, $a_i$ has an incentive to join $A_i \setminus \{a_i\}$.
Otherwise, we have $|A_i \setminus \{a_i\}| \leq 2$, hence we can merge the two groups using a $2$-deviation.
\end{proofclaim}

\begin{claim}
\label{claim:contre-2}
$b_i$ is in the same group as $A_i$ or $A_{i+1}$.
\end{claim}

\begin{proofclaim}~[Claim~\ref{claim:contre-2}]
Suppose it is not the case.
Then $X(b_i)$ contains at most two other nodes: one of $b_{i-1}$ and $b_{i+1}$ (together enemies), and one of $c_0$ and $c_1$ (enemies).
If $|X(b_i)| \leq 2$ or the group containing $A_i$ has size at least $3$, then $b_i$ has an incentive to leave her group.
So, we assume $|X(b_i)| = 3$, $h=2$ and $A_i$ is a group of size $h=2$ in $P$.
There are two cases.
If $X(b_i) = X(c_i)$ then we can merge $X(b_i)$ with $A_i$ using the $2$-deviation $\{b_i,c_i\}$.
Otherwise, $c_{i-1}$ is in $X(b_i)$, and $X(b_i)$ also contains either $b_{i-1}$ or $b_{i+1}$; hence we can break $P$ using the vertex in $\{b_{i-1},b_{i+1}\} \cap X(b_i)$ with $c_{i-1}$. 
\end{proofclaim}

\begin{claim}
\label{claim:contre-3}
There exists an $i$ such that $A_i, b_i$ and $b_{i-1}$ are in the same group. 
\end{claim}

\begin{proofclaim}~[Claim~\ref{claim:contre-3}]
Again, we show the claim by contradiction.
We distinguish two cases:

Case $1$: $b_{i-1}$ is with $A_i$, but not $b_i$.
So, as the claim is supposed to be false, $b_i$ is with $A_{i+1}$, $b_{i+1}$ is with $A_{i+2}$, and $b_{i+2}$ is with $A_{i+3}$.
Either we have $X(b_{i-1}) = X(c_i)$, hence $b_i$ has an incentive to join $X(b_{i-1})$, or $X(b_{i-1}) \neq X(c_i)$, hence $\{b_i,c_i\}$ joins group $X(b_{i-1})$.

Case $2$: $b_i$ is with $A_i$, but not $b_{i-1}$.
So, as the claim is supposed to be false, $b_{i-1}$ is with $A_{i-1}$, $b_{i+1}$ is with $A_{i+1}$, and $b_{i+2}$ is with $A_{i+2}$.
Note that $c_i$ is either in $X(b_i)$ or in $X(b_{i+2})$.
Without loss of generality, suppose it is in $X(b_{i+2})$.
Either we have $X(b_{i-1}) = X(c_{i-1})$, hence $b_i$ has an incentive to join $X(b_{i-1})$, or $X(b_{i-1}) \neq X(c_{i-1})$, hence $\{b_i,c_{i-1}\}$ joins group $X(b_{i-1})$.
\end{proofclaim}

\begin{claim}
\label{claim:contre-4}
There exists an $i$, such that $b_{i}$, $b_{i-1}$ and $c_{i}$ are in the same group as $A_{i}$ in $P$.
Moreover, such a group is unique in $P$, due to the conflict graph in $G$. 
\end{claim}

The proof of Claim~\ref{claim:contre-4} is straightforward by using the previous claim.
By symmetry, we assume that $P$ contains the group $\{b_0,b_3,c_0\} \cup A_0$.

Case $1$: the group containing $A_2$ also contains $b_1$ and $b_2$.
Then either $c_1$ is with $A_1$ and a $1$-deviation of $b_1$ to $A_1 \cup \{c_1\}$ breaks $P$, or $c_1$ is not with $A_1$ and a $2$-deviation of $\{b_1,c_1\}$ to $A_1$ breaks $P$.

Case $2$: the group containing $A_2$ only contains $b_2$.
Then either $c_1$ is with $A_3$ and a $2$-deviation of $\{b_2,b_3\}$ to $A_3$ breaks $P$, or $c_1$ is not with $A_3$ and a $3$-deviation of $\{b_2,b_3, c_1\}$ to $A_3$ breaks $P$.

Case $3$: the group containing $A_2$ only contains $b_1$.
In that case, a $1$-deviation from $b_1$ to the group containing $A_1$ breaks $P$.

Case $4$: $A_2$ is a group in $P$.
This implies $b_2$ is with $A_3$.
Then either $c_1$ is with $A_3$ and a $1$-deviation of $b_3$ to $A_3$ breaks $P$, or $c_1$ is not with $A_3$ and a $2$-deviation of $\{b_3, c_1\}$ to $A_3$ breaks $P$.

Finally, there does not exist a $3$-stable partition $P$ for $G$.
\end{proof}

\begin{corollary}
\label{lem:counter-example-2}
Let $\mathcal{W} = \{-\infty,0,1\}$.
Then $k(\mathcal{W}) = 2$.
\end{corollary}

\subsection{Counter-examples for stability}

Counter-examples with a different set of weights are presented in the sequel.
First, observe the weights $2,3,4$ in the counter-example of Figure \ref{fig:nonstable} can be replaced by any weights $w_1,w_2,w_3$ such that $w_1 < w_2 < w_3$, and $w_1+w_2 > w_3$. 
Especially, they can be replaced by $b,b+1,b+2$, for any $b \geq 2$.  
We now show through a counter-example, no $2$-stable partition exists in general when $\mathcal{W}$ contains two positive elements, even when the zero is not present.

\begin{lemma}
\label{lemma:extent}
Let $a,b$ be two positive integers such that $a < b$.
There is a graph $G=(V,w)$ with $\mathcal{W} = \{-\infty,a,b\}$ such that there is no $2$-stable partition for $G$.
\end{lemma}

\begin{proof}[Lemma~\ref{lemma:extent}]
Users are partitioned into four sets $V_1 = \{x_1,x_2,x_3\}$, $V_2 = \{y_1,y_2,y_3\}$, $V_3=\{z_1,z_2,z_3\}$ and $\{u_1,u_2,u_3\}$.
Each of these sets is a clique with weight $b$.
Moreover, all the edges that lie between $V_1$ and $V_2$, $V_2$ and $V_3$, and $V_3$ and $V_1$ have weight $-\infty$.
We also set that for all vertices $x_i \in V_1$, for all vertices $y_j \in V_2$, and for all vertices $z_l \in V_3$, $w_{x_iu_1} = w_{y_ju_2} = w_{z_lu_3} = b$, whereas $w_{x_iu_3} = w_{y_ju_1} = w_{z_lu_2} = -\infty$.
Finally:
\begin{itemize}
\item $w_{u_1z_1} = w_{u_1z_2} = b$; $w_{u_1z_3} = a$;
\item $w_{u_2x_1} = w_{u_2x_2} = b$; $w_{u_2x_3} = a$;
\item $w_{u_3y_1} = w_{u_3y_2} = b$; $w_{u_3y_3} = a$.
\end{itemize}
 
We now assume there exists a $2$-stable partition $P$ for $G$.

\begin{claim}
$V_1, V_2$ and $V_3$ are subgroups in $P$.
\end{claim}

\begin{proofclaim}
By symmetry, it suffices to show the claim for $V_1$.
First, we have that $\{x_1,x_2\}$ is a subgroup in $P$ by Lemma~\ref{lemma:group}.
Remark that all the vertices in $V_1$ have the same enemies in the graph.
Moreover, the sum of the positive weights that lie between $x_3$ and $V \setminus V_1$ equals $b+a$, which is lower than $2b$.
Hence $V_1$ is a subgroup of $P$ because we suppose $P$ to be $1$-stable.
\end{proofclaim}

Then, we can replace subsets $V_1,V_2,V_3$ with three nodes $v_1,v_2,v_3$ in the obvious way.
In so doing, we get the variation of Figure \ref{fig:nonstable} with weights $w_1 = b$, $w_2 = 2b+a$ and $w_3 = 3b$.
Moreover, it is important in the counter-example of Figure \ref{fig:nonstable} to remark that one can always find a $2$-deviation using nodes $u_1,u_2$ and $u_3$ only.
As a consequence, such a $2$-deviation is valid for $G$ and so, there does not exist any $2$-stable partition for $G$. 
\end{proof}

Other results can be expressed in a convenient way using the subset of non-negative weights $\mathcal{W}^+$ and the subset of non-positive weights $\mathcal{W}^-$.
These two subsets entirely cover the set of weights $\mathcal{W}$, but they may not be disjoint whenever $0 \in \mathcal{W}$.
\begin{corollary}
\label{cor:mn}
Suppose $\{-\infty,N\}$ is a strict subset of $\mathcal{W}$.
Then, $k(\mathcal{W}) = 1$ if $\mathcal{W}^+ \setminus \{N\} \neq \emptyset$, and $k(\mathcal{W}) \leq 2$ otherwise.
\end{corollary}

\begin{proof}~[Corollary~\ref{cor:mn}]
First assume there exists a weight $a \in \mathcal{W}^+ \setminus \{N\}$.
Then we can apply Lemma~\ref{lemma:extent} with the subset $\{-\infty,a,N\}$.

Otherwise, we have by Lemma~\ref{lem:counter-example} there exists a graph $G_0 = (V_0,w^0)$ with $\{-\infty,0,1\}$ such that there does not exist any $3$-stable partition for $G_0$.
Let $a \in \mathcal{W} \setminus \{-\infty,N\}$ be any weight.
Note that $a$ may be positive or non positive for this construction.
We define $G_1 = (V_0,w^1)$ with the same conflict edges as for $G_0$, such that for all vertices $u,v \in V_0$,  $w^1_{uv} = N$ if $w^0_{uv} = 1$,  $w^1_{uv} = a$ if $w^0_{uv} = 0$.
As there does not exist any $3$-stable partition for $G_0$, there does not exist any $3$-stable partition for $G_1$ either.
\end{proof}

\begin{lemma}
\label{lem:counter-example-3}
Let $a,b$ be two positive integers.
There exists a graph $G=(V,w)$ with $\mathcal{W} = \{-a,b\}$ such that any partition for $G$ is not $2(b+a+1)$-stable.
\end{lemma}

\begin{proof}[Lemma~\ref{lem:counter-example-3}]
Our aim in this case is to emulate the variation of Figure~\ref{fig:nonstable} with the three weights $w_1 = b+a+1$, $w_2 = 2(b+a) + 3$, and $w_3 = 3(b+a+1)$.
As $b+a+1 < 2(b+a) + 3 < 3(b+a+1)$, and $b+a+1 + 2(b+a) + 3 = 3(b+a+1) +1 >  3(b+a+1)$, this change of weights does not make $2$-stable partitions exist for the graph.

To do that, let $t_1, t_2, t_3$ be any three positive integers that satisfy the following properties:
\begin{itemize}
\item $t_3 \geq t_2 \geq t_1$;
\item $t_1 \geq 5(b+a+1)/\max (b-a, 1)$;
\item $\max ( t_1b, (t_1+1)a + 2(b+a) + 3 ).a \geq (3(b+a)+2)b + 1$;
\item $\max ( t_2b, (t_2+1)a + 2(b+a) + 3 ).a \geq [(t_1+5)(b+a) + 4]b + 1$;
\item $\max ( t_3b, (t_3+1)a + 2(b+a) + 3 ).a \geq [(t_2+5)(b+a) + 4]b + 1$.   
\end{itemize}
We can note the possible values of those three integers are infinite.
Indeed, the value of $t_1$ can be any value that is greater than every of its (finite) lower bounds.
Furthermore, if the value of $t_1$ is arbitrarily fixed, then we can also choose for $t_2$ any value that is sufficiently large and, finally, the value of $t_3$ can be chosen in the same way, depending on the values of $t_1$ and $t_2$.

For every $1 \leq i \leq 3$, we also define the size $s_i = t_i(b+a) + 3(b+a+1) = (t_i+1)(b+a) + 2(b+a) + 3$.\\

\noindent
We now define the graph $G=(V,w)$. 
The vertices consist in three subsets $V_1,V_2$ and $V_3$, plus three other subsets $U_1,U_2$ and $U_3$.
All of them are cliques with weight $b$.
Furthermore, $U_i$ has size $b+a+1$ and $V_i$ has size $s_i$, for every $1 \leq i \leq 3$. 
There are also two distinct subcliques $V_{i,m} \subset V_{i,p} \subset V_i$, with $|V_{i,m}| = t_ib$, $|V_{i,p}| = (t_i+1)b$.

We set $w_{uu'} = b$ for any $u,u' \in U_1 \cup U_2 \cup U_3$, and in addition, for every $1 \leq i \leq 3$:
\begin{itemize}
\item $w_{u_iv_i} = b$ for every $u_i \in U_i$, $v_i \in V_i \setminus V_{i,m}$;
\item $w_{u_iv_{i+1}} = b$, for every $u_i \in U_i$, $v_{i+1} \in V_{i+1} \setminus V_{i+1,p}$ (with the convention $V_{i+1} = V_1$ when $i=3$).
\end{itemize}
Every other weight equals $-a$.\\

\noindent
By contradiction.
Now assume there is a $2(b+a+1)$-stable partition of the nodes $P$.
By Lemma~\ref{lemma:group}, it comes that subsets $U_1,U_2,U_3, V_{1,m}, V_{2,m}, V_{3,m}, V_{1} \setminus V_{1,p},  V_{2} \setminus V_{2,p},  V_{3} \setminus V_{3,p}$ are subgroups in $P$.

\begin{claim}
For every $1 \leq i \leq 3$:
\begin{itemize} 
\item if $a \geq b$ then $|V_{i,p}| + 2(b+a+1) \leq | V_{i} \setminus V_{i,p}|$, 
\item if $b > a$ then $| V_{i} \setminus V_{i,m}| + 2(b+a+1) \leq |V_{i,m}|$.
\end{itemize} 
\end{claim}

\begin{proofclaim}
By the hypothesis, we have $|V_{i,p}| = (t_i+1)b$, hence $| V_{i} \setminus V_{i,p}| = s_i -  (t_i+1)b = (t_i+1)(b+a) + 2(b+a) + 3 - (t_i+1)b  =  (t_i+1)a + 2(b+a) + 3$; furthermore, $|V_{i,m}| = t_ib$, and so, $| V_{i} \setminus V_{i,m}| = s_i - t_ib =   t_i(b+a) + 3(b+a+1) - t_ib =  t_ia + 3(b+a+1)$. 

First if $a \geq b$, then $|V_{i,p}| + 2(b+a+1) =  (t_i+1)b + 2(b+a+1) < (ti+1)a + 2(b+a) + 3 = | V_{i} \setminus V_{i,p}|$, and the claim is proved. Otherwise, the equation $| V_{i} \setminus V_{i,m}| + 2(b+a+1) \leq |V_{i,m}|$ is equivalent to the following lines:
\[  t_ia + 3(b+a+1) +  2(b+a+1) \leq  t_ib\]
\[ 5(b+a+1) \leq t_i(b-a) \]
\[ 5(b+a+1)/(b-a) \leq t_i\]
and the claim is also proved, as $t_i \geq t_1 \geq  5(b+a+1)/\max (b-a, 1)$. 
\end{proofclaim}

In the sequel, we will denote $V_i' =  V_{i} \setminus V_{i,p}$ if $a \geq b$, and $V_i' = V_{i,m}$ if $b > a$.
We recall that $V_i'$ is a subgroup in $P$ by the previous claim.

\begin{claim}
For any $j < i$, none of the vertices in $V_j$ is in the same group in $P$ as vertices of $V_i'$.
\end{claim}

\begin{proofclaim}
Indeed, any vertex in $V_j$ has at most $s_j - 1 + 2(b+a+1) =  (t_j+1)(b+a) + 2(b+a) + 3 +  2(b+a+1) - 1 = (tj+5)(b+a) +4$ friends in its group.
Moreover, $[(tj+5)(b+a) +4]b < \max ( t_ib, (t_i+1)a + 2(b+a) + 3 ).a = |V_i'|a$ by the hypothesis.
\end{proofclaim}

\begin{claim}
Subsets $U_1$ and $V_3'$ are not in the same group in $P$.
\end{claim}

\begin{proofclaim}
By contradiction.
Suppose $U_1 \cup V_3'$ is a subgroup in $P$.
By the previous claim, there can be no vertex of $V_1 \cup V_2$ in this group.
Thus, any vertex of $U_1$ has at most $2(b+a+1) + b+a = 3(b+a) + 2$ friends in its group.
Furthermore, as $|V_3'|a = \max ( t_3b, (t_3+1)a + 2(b+a) + 3 ).a \geq (3(b+a)+2)b + 1$, we have that any vertex of $U_1$ has a negative utility, which is a contradiction because $P$ is $1$-stable.
\end{proofclaim}

In the same way, one can check that $U_2$ and $V_1'$ are not in the same group in $P$,  $U_3$ and $V_2'$ are not in the same group in $P$.

\begin{claim}
For any $1 \leq i \leq 3$, $V_i$ is a subgroup in $P$.
\end{claim}

\begin{proofclaim}
By our first claim, $|V_3 \setminus V_3'| + 2(b+a+1) \leq |V_3|'$, and by our third claim, there can be no vertex $u$ in the same group as $V_3'$ in $P$ such that $w_{uv_3} = -a$ for any vertex $v_3 \in V_3$. 
Consequently, $V_3$ is a subgroup in $P$, as any vertex of $V_3 \setminus V_3'$ that would not be in the same group as $V_3'$ would break the partition. 
\end{proofclaim}

It directly follows that $V_2$ is a subgroup in $P$, hence $V_1$ is a subgroup in $P$.

Consequently, subsets $V_1,V_2,V_3,U_1,U_2$ and $U_3$ can be replaced with vertices $v_1,v_2,v_3,u_1,u_2$ and $u_3$ in the obvious way.
We get the variation of Figure~\ref{fig:nonstable} that we announced at the beginning of the proof.
Observe we can mimic any $2$-deviation involving nodes $u_1,u_2$ or $u_3$ in this reduced graph by a $2(b+a+1)$-deviation involving subsets $U_1,U_2$ or $U_3$ in $G$. 
So, there does not exist any $2(b+a+1)$-stable partition for $G$.
\end{proof}

\begin{theorem}
\label{lem:car}
$k(\mathcal{W})=\infty$ $\Leftrightarrow$ $\mathcal{W} = \mathcal{W}^{+}$, or $\mathcal{W} = \mathcal{W}^{-}$, or $\mathcal{W} = \{-\infty, b\}$, or $\mathcal{W} = \mathcal{W}^{-} \cup \{N\}$ and $-\infty \notin \mathcal{W}$.
\end{theorem}

\begin{proof}~[Theorem~\ref{lem:car}]
First, suppose that $\mathcal{W}$ is not of one of the forms we propose.
Let $a,b$ be two positive integers.
If $-a,b \in \mathcal{W}$ then $k(\mathcal{W}) \leq 2(a+b) + 1$ by Lemma~\ref{lem:counter-example-3}.
So, we will assume either $-\infty$ is the only negative weight in $\mathcal{W}$, or $-\infty \in \mathcal{W}$ and $N$ is the only positive weight in $\mathcal{W}$.
In the first case, let $w_p$ denote the largest positive weight in $\mathcal{W}$; then, either there exists a positive integer $a < w_p$ such that $a \in \mathcal{W}$, hence $k(\mathcal{W}) = 1$ by Lemma~\ref{lemma:extent}, or $\mathcal{W} = \{-\infty,0,w_p\}$ and so, $k(\mathcal{W}) = 2$ by Corollary~\ref{lem:counter-example-2}.
In the second case, $k(\mathcal{W}) \leq 2$ by Corollary~\ref{cor:mn}.  
We deduce that $k(\mathcal{W})$ is upper-bounded in all the other cases, by the lemmas that are in Section~\ref{sec:weighted}.\\

Now, for the 'if part', we start the proof with a preliminary remark.
Let $G = (V,w)$ be any graph, and let $P$ be any $1$-stable partition for $G$.
For any node $u \in V$, $f_{u}(P) \geq \sum_{v \in V}{\min(w_{uv},0)}$ and $f_{u}(P) \leq \sum_{v \in V}{\max(w_{uv},0)}$.\\

\noindent
Let $\mathcal{W} = \mathcal{W}^{+}$.
By the first remark, $f_{u}(P) \leq \sum_{v \in V}{\max(w_{uv},0)}$ for any node $u \in V$ and for any stable partition $P$ for $G$.
In our case, $\sum_{v \in V}{\max(w_{uv},0)} = \sum_{v \in V}{w_{uv}}$ because $w_{uv} \geq 0$ for any $uv \in E$.
Let $P = (V)$ be the partition for $G$ composed of $1$ group of size $n$.
By construction, for any node $u \in V$, $f_{u}(P) = \sum_{v \in V}{w_{uv}}$.
Thus $P$ is a $k$-stable partition for $G$.
\\

\noindent
Let $\mathcal{W} = \mathcal{W}^{-}$.
By the very first remark, $f_{u}(P) \leq \sum_{v \in V}{\max(w_{uv},0)}$ for any node $u \in V$ and for any stable partition $P$ for $G$.
In our case, $\sum_{v \in V}{\max(w_{uv},0)} = \sum_{v \in V}{0} = 0$ because $w_{uv} \leq 0$ for any $uv \in E$.
Let us denote $V=\{u_{1}, \ldots, u_{n}\}$.
Let $P=(\{u_{1}\}, \ldots,\{u_{n}\})$ be the partition for $G$ composed of $n$ groups of size $1$.
$P$ is a $k$-stable partition for $G$ because $f_{u}(P) = 0$ for any $u \in V$.\\

\noindent
The case $\mathcal{W} = \{-\infty, b\}$ is proved in~\cite{kleinberg2013information}.\\

\noindent
Let $\mathcal{W} = \mathcal{W}^{-} \cup \{N\}$ such that $-\infty \notin \mathcal{W}$.
Let $G'=(V,E')$ be the graph induced by the set of arcs $E' = \{uv, w_{uv}=N \}$.
Let $G_{1}, \ldots, G_{t}$ be the different connected components of $G'$.
Given a node $u \in V$, let $\Gamma_{G'}(u)$ be the set of neighbors of $u$ in graph $G'$.
Formally $\Gamma_{G'}(u) = \{v, uv \in E'\}$.
By the preliminary remark, we get that, for any stable partition $P$ for $G$ and for any node $u \in V$, $f_{u}(P) \leq N|\Gamma_{G'}(u)|$.

\noindent
Let $P= (V(G_{1}), \ldots, V(G_{t}))$ be the partition for $G$ composed of the nodes of the different maximal connected components of $G$.
By Definition of $N$, $f_{u}(P) \geq (N-1)|\Gamma_{G}(u)|$.

\noindent
We now prove that $P$ is $k$-stable.
By contradiction.
Suppose that there exists a $k$-deviation.
Let $P'$ be the partition for $G$ obtained after the $k$-deviation.
There are two cases.
\begin{itemize}
\item A subset $S \subset V(G_{i})$, for some $i$, $1 \leq i \leq t$, $|S| \leq k$, forms a new group.
In that case, we prove that at least one node $u \in S$ is such that $f_{u}(P') \leq f_{u}(P)$.
Indeed, consider a node $u \in S$ such that $\Gamma_{G'}(u) \cap S \neq \Gamma_{G'}(u)$.
Note that there always exists such a node because $|S| < |V(G_{i})|$ and $V(G_{i})$ is a maximal connected component of $G$.
Thus, we get $f_{u}(P') \leq (N-1)|\Gamma_{G'}(u)|$ while $f_{u}(P) \geq \max((N-1),0)|\Gamma_{G'}(u)|$.
So $f_{u}(P') \leq f_{u}(P)$.
A contradiction.
\item A subset $S \subset V$, $|S| \leq k$, reaches a group $V(G_{j})$, for some $i$, $1 \leq i \leq t$.
If $V(G_{i}) \cap S \neq \emptyset$, then we claim $V(G_{i}) \subset S$.
Otherwise, we would have $f_{P'}(u) \leq f_{u}(P)$ for some $u \in V(G_{i})$.
Indeed, consider a node $u \in S$ such that $\Gamma_{G'}(u) \cap S \neq \Gamma_{G'}(u)$.
So $f_{u}(P') \leq (N-1)|\Gamma_{G'}(u)|$ while $f_{u}(P) \geq (N-1)|\Gamma_{G'}(u)|$.
Thus, if $V(G_{i}) \cap S \neq \emptyset$, then $V(G_{i}) \subset S$.
Consider a set $V(G_{i})$, for some $i$, such that $V(G_{i}) \subset S$.
In that case, $f_{u}(P) \leq f_{u}(P')$ for any $u \in V(G_{i})$ because $w_{uv} \leq 0$ for any $v \in V \setminus V(G_{i})$.
A contradiction.
\end{itemize}
There does not exist a $k$-deviation for $P$ and so $P$ is $k$-stable.
\end{proof}

\subsubsection{Proof of Theorem~\ref{thm:complexity-decision}}

First, given a set of weights $\mathcal{W}$ such that $\mathcal{W}^+ \neq \emptyset$, we recall that $w_p \in \mathcal{W}$ is defined as the largest positive weight in $\mathcal{W}$.
Observe that $w_p$ may equal the arbitrary value $N$ in the general case.
Furthermore, as the set $\mathcal{W}$ is fixed, we can assume without loss of generality the value $N$ is a fixed polynomial of $n$; hence $w_p$ is either constant or polynomial.

Then we present two different ways to increase the minimum utility of the nodes in a graph.
In the general case, that minimum utility equals $0$; otherwise, the partition of the nodes is not even $1$-stable.
However, we can improve it with external cliques, in a way that keeps the stability properties of the graph safe.
Given a graph $G = (V,w)$ and a positive integer $t$, we build the graph $\tilde{G_t}$ by adding to the set of vertices $n$ cliques of $t$ nodes.
We further impose that all the nodes in the same $t$-clique have weight between them and between the vertices in $V$ equal to $w_p$, whereas they have a conflict edge with any other vertex in the remaining $t$-cliques.  
By doing so, we intuitively increase the minimal utility of the nodes to $w_p t$ for any stable partition of the graph.
Note that the transformation also keeps the stability properties of the graph safe, which can be formalized as follows:

\begin{definition}
\label{def:transfo}
Let $\mathcal{W}$ be any fixed set of weights such that $\mathcal{W}^+ \neq \emptyset$.
Given a graph $G=(V,w)$ and a positive integer $t$, for every $1 \leq i \leq n$ we can define the clique graph $K_t^i = (V_i,w^i)$ with $t$ nodes and the set of weights $\{w_p\}$.
The graph $\tilde{G_t} = (\tilde{V},\tilde{w})$ is then built as follows:
\begin{itemize}
\item $\tilde{V} = V \cup \bigcup_{i=1}^n V_i$;
\item for every $1 \leq i,j \leq n$, for all $u_i \in V_i$, and for all $u_j \in V_j$, $\tilde{w}_{u_i u_j} = w_p$ if $i=j$, $\tilde{w}_{u_i u_j} = -\infty$ otherwise;
\item  for all $u,v \in V$, $\tilde{w}_{uv} = w_{uv}$;
\item  for all $1 \leq i \leq n$, for all $u_i \in V_i$ and for all vertex $v \in V$, $\tilde{w}_{u_i v} = w_p$.
\end{itemize}
\end{definition}

\begin{lemma}
\label{lemma:keep-stability}
Let $\mathcal{W}$ be any fixed set of weights such that $\mathcal{W}^+ \neq \emptyset$, and $k$ be a positive integer.
Given a graph $G=(V,w)$, and a positive integer $t$ such that $t > n$, there exists a $k$-stable partition of the nodes for $G$ if and only if there exists a $k$-stable partition for $\tilde{G_t}$.
\end{lemma}

\begin{proof}~[Lemma~\ref{lemma:keep-stability}]
Let $P = ( X_1, \ldots, X_n )$ be a $k$-stable partition of the nodes for $G$.
We claim that $P' = ( X_1 \cup V(K_t^1), \ldots, X_n \cup V(K_t^n) )$ is a $k$-stable partition for $\tilde{G_t}$.
By contradiction.
Assume that $P'$ is not $k$-stable, and let $S$ be a $k$-set that breaks $P'$.
First, we claim that for every $i$, $1 \leq i \leq n$, $S$ cannot be a subset of $V(K_t^i)$.
This is clear as all the vertices in $K_t^i$ have enemies in every non-empty group of $P'$ that is not their own group; in addition, they only have friends in their group in $P'$, and all of them are already in the same group.
Then it follows that $S' = S \cap V(G) \neq \emptyset$.
Moreover, we claim that $S'$ breaks $P$.
Indeed, after the $k$-deviation, the vertices in $S'$ end in the same group as a (possibly empty) subset of $K_t^i$ for some $1 \leq  i \leq n$. 
Let $1 \leq j \leq n$ such that $X_j \cap (V(G) \setminus S') = \emptyset$. 
Such an integer $j$ exists, because there are at most $n$ non-empty groups in $P$. 
Then it follows that if we move $S'$ into the group $X_j$, it is still a valid $k$-deviation for $P'$ and so, we can mimic this $k$-deviation in $P$.
A contradiction.
Therefore $P'$ is $k$-stable.\\ 

\noindent
Conversely, assume that there is no $k$-stable partition for $G$, whereas there exists a $k$-stable partition $P' = ( X_1', \ldots, X_n')$ for $\tilde{G_t}$.
Particularly, for all $1 \leq i \leq n$, we get that $K_t^i$ is a subgroup of $P'$ by Lemma~\ref{lemma:group}.
Furthermore, we claim that for all vertex $u \in V$, there exists an integer $i \in \{ 1, \ldots, n\}$ such that $K_t^i \subset X(u)$.
This directly comes from the fact that $t > n$ and that $w_{p}$ is the largest positive weight in $\mathcal{W}$.
Let $P = ( X_1' \cap V(G), \ldots, X_n' \cap V(G) )$ be the underlying partition for $G$ that one can deduce from $P'$.
By the hypothesis, there exists a $k$-deviation $S$ that breaks $P$.
However, we can mimic this deviation in $P'$, which is a contradiction because $P'$ is $k$-stable.
Consequently, there does not exist any $k$-stable partition for $\tilde{G_t}$.
\end{proof}

It is useful to notice that Lemma~\ref{lemma:keep-stability} still holds if we allow to add $n' \geq n$ cliques $K_t^i$ in graph $\tilde{G_t}$.

Another way to increase the minimum utility of the nodes is to substitute every vertex in the graph by cliques.
By this second process, we get the graph $\bar{G_{\alpha}}$.
Formally:

\begin{definition}
Let $\mathcal{W}$ be any set of weights such that $\mathcal{W}^+ \neq \emptyset$.
Given a graph $G=(V,w)$ and a positive integer $\alpha$, for every vertex $u \in V$ we can define the clique graph $K_u = (V_u,w^u)$ with $\alpha$ nodes and the set of weights $\{w_p\}$.
The graph $\bar{G_{\alpha}} = (\bar{V},\bar{w})$ is then built as follows:
\begin{itemize}
\item $\bar{V} = \bigcup_{u \in V} V_u$;
\item for every $u \in V$, for all $u_1,u_2 \in V_u$, $\bar{w}_{u_1 u_2} = w_p$;
\item  for all $u,v \in V$, for all $u_i \in V_u$ and for all $v_j \in V_v$, $\bar{w}_{u_i v_j} = w_{uv}$.
\end{itemize}
\end{definition}

We are now able to prove Theorem~\ref{thm:complexity-decision}:

\begin{proof}~[Theorem~\ref{thm:complexity-decision}]
Suppose there exists a graph $G_0 = (V_0,w^0)$ with $\mathcal{W}$ such that there does not exist any $k$-stable partition for $G_0$.
We get there are positive weights, because otherwise the partition with only singleton groups is $k$-stable for $G_0$.
Moreover, we can assume that there exists a node $x_0 \in V_{0}$ such that the removal of $x_0$ makes the existence of a $k$-stable partition for the gotten subgraph.
Indeed, otherwise, we remove nodes sequentially until obtaining this property.

Let us choose $P^0$ a $k$-stable partition for $G_0 \setminus x_0$.
By the hypothesis, the partition $P^0 \cup \{x_0\}$ is not $k$-stable for $G_0$.
Let us define $f_0$ as the maximum utility the vertex $x_0$ can get with only one $k$-deviation that breaks $P^0 \cup \{x_0\}$.
Note that we can always assume $f_0 > 0$ by using the transformation of Definition~\ref{def:transfo} if needed.
We define two other constants, namely $\alpha = \lceil \frac{f_0}{w_p} \rceil$ and $c_0 = 2n_0+1$, with $n_0 = |V_0|$.

We can now prove the NP-hardness by using a polynomial reduction for the \MIS.
Let $G = (V,E)$ be a graph, and let $c \geq c_0$ be an integer.
We identify $G$ to the conflict graph of some network $D_G$ with the set of weights $\{-\infty,w_p\}$.
Let $t$ be the integer $\lfloor \alpha c - \frac{f_0}{w_p} \rfloor$.
We can consider the graphs $G_1 = \tilde{G_0}_t = (V_1,w^1)$ and $G_2 = \bar{G}_{\alpha} = (V_2,w^2)$.
Observe that $t >  \alpha c - \frac{f_0}{w_p} - 1 \geq \alpha c - n_0 -1 \geq c - n_0 -1 \geq n_0$, because $f_0 \leq n_0w_p$.
So we can apply Lemma~\ref{lemma:keep-stability} to $G_1$.

We build the graph $H_G = (V_H,w^H)$ as follows:
\begin{itemize}
\item $V_H = V_1 \cup V_2$;
\item for all $i \in \{1,2\}$, for all $u_i,v_i \in V_i$, $w^H_{u_i v_i} = w^i_{u_i v_i}$;
\item for all $u_1 \in V_1$, for all $v_2 \in V_2$, $w^H_{u_1 v_2} = w_p$ if $u_1 = x_0$, and  $w^H_{u_1 v_2} = -\infty$ otherwise. 
\end{itemize} 

\vspace{-0.25cm}
\begin{figure}[t]
\centerline{
\includegraphics[width=0.5\linewidth]{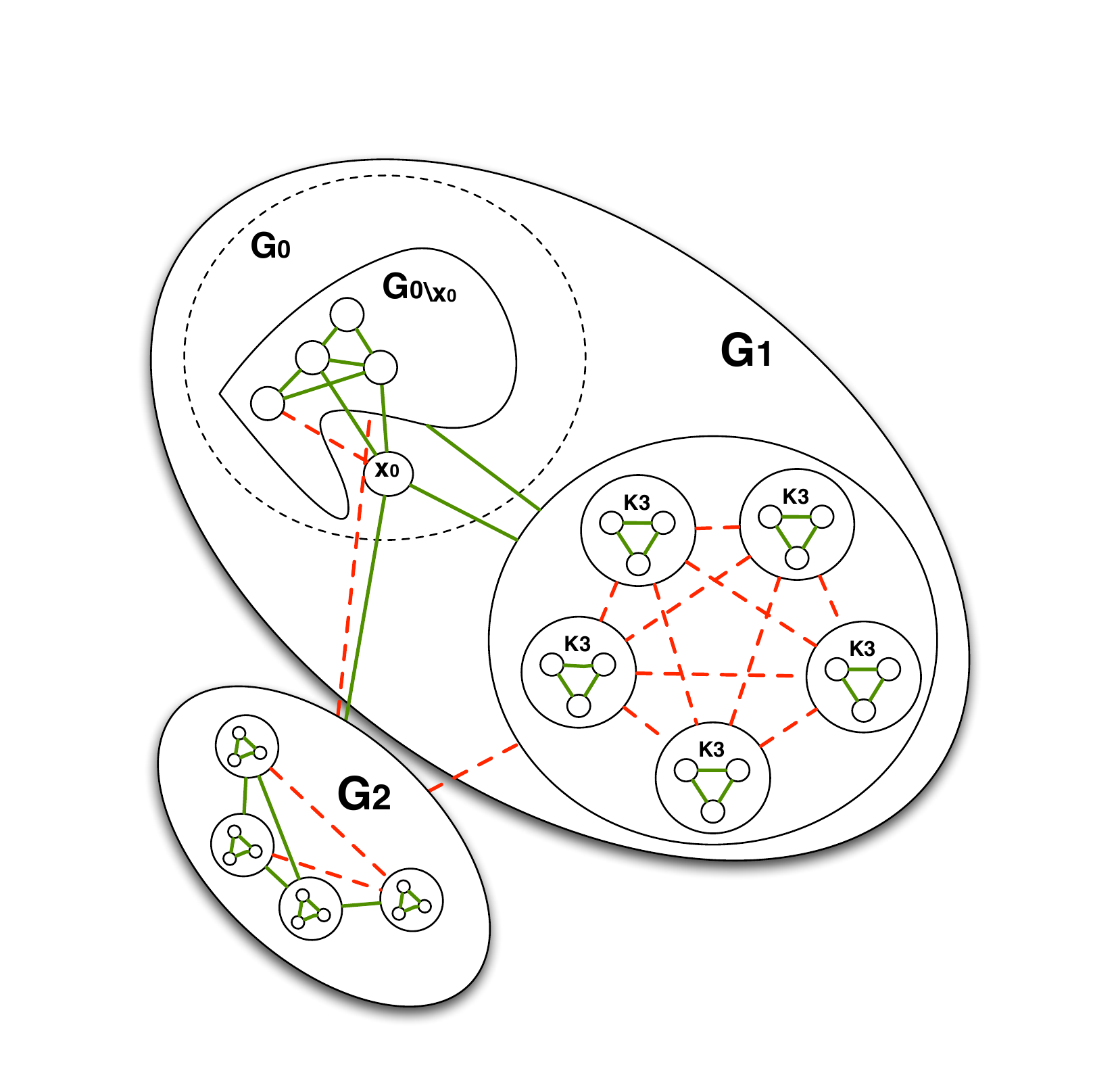}}
\caption{The transformation of an input.}
\label{fig:blabla}
\end{figure}

The transformation above is illustrated in Figure~\ref{fig:blabla}.
First assume that any independent set of $G$ has a size lower than $c$.
Then, there can be no group in a stable partition with more than $\alpha (c - 1)$ vertices of $V_2$. 
Furthermore, $\alpha(c-1) = \alpha c - \lceil \frac{f_0}{w_p} \rceil \leq \alpha c - \frac{f_0}{w_p} -1 < t$, and so, the minimum utility reachable by $x_0$ in $G_1$ is greater than its maximum utility in $G_2$.
As a consequence, $P_0$ can be partitioned into a $1$-stable partition for $G_1$ and a $1$-stable partition for $G_2$.
Moreover, there is no $k$-stable partition for $G_0$ by the hypothesis, hence there is no $k$-stable partition for $G_1$ either by  Lemma~\ref{lemma:keep-stability}.
Therefore, there is no $k$-stable partition for $H_G$.

Conversely, assume that there exists an independent set of $G$ with size at least $c$.
By~\cite{kleinberg2013information}, there exists a $k$-stable partition $P^{\alpha}$ for $G_2 \cup \{x_0\}$ such that there is a group $X \in P^{\alpha}$ that is a maximum independent set of the conflict graph of $G_2 \cup \{x_0\}$.
Furthermore, $x_0 \in X$, because $x_0$ has no conflict within the graph.
We can also observe that $G_1 \setminus x_0$ is $\tilde{(G_0 \setminus x_0)}_t$ with an extra clique graph $K_t^{n_0}$ added.
Hence, both these graphs have the same stable properties, that is if $P^0 = ( X_1, \ldots, X_{n_0-1} )$ is a $k$-stable partition for $G_0 \setminus x_0$, then we can associate to it a $k$-stable partition $P_1^t$  for $G_1 \setminus x_0$, by Lemma~\ref{lemma:keep-stability}.
We then claim that $P^H = P^{\alpha} \cup P_1^t$ is a $k$-stable partition for $H_G$.
Indeed, we have that the utility of $x_0$ in $P^{\alpha}$ is at least $w_p \alpha c$.
Moreover, the maximum utility $x_0$ can get after a $k$-deviation that breaks $P_1^t$ is $f_0 + w_p t = w_p ( t + \frac{f_0}{w_p}) \leq w_p \alpha c$.

We can conclude the NP-hardness, as our transformation is polynomial, and the \MIS is NP-complete~\cite{Kar72}. 
\end{proof}

\subsection{Appendix 3: Extensions of coloring games}
\label{annexe:extensions}


\subsubsection{Asymmetry and Gossiping}



\begin{proof}~[Theorem~\ref{theorem:oneStableConfigurationProblemNP}]
We use in our reduction the multi-way partitioning number problem, which is NP-complete~\cite{Korf2009}.
Furthermore, we constrain ourselves to $g=2$ in order to ease the proof.
This is enough so that we give the intuition of the general case.
Note that the multi-way partitioning number is also known as the Partition problem when $g=2$.
So, in order to prevent any confusion, we will use the word 'partition' only for the Partition problem, that is we call the partitions of the nodes 'configurations' in the sequel.\\

Consider an instance $S$ of Partition problem.
Let us call $T = \sum_{s\in{S}}s$.
We define $D$ as follows. 
The set of vertices is $V = S \cup {\{z,u,v\}}$.
For all $s_{1}, s_{2} \in{S}$, with $s_1 \neq s_2$, we set $w_{s_{1}s_{2}} = - s_{2}$, $w_{s_{1}u} = w_{s_{1}v} = T - s_{1} + 1$, $w_{us_{1}} = w_{vs_{1}} = 0$, and finally, $w_{s_{1}z} = - w_{zs_{1}} = s_{1}$. 
We also set $w_{uv} = w_{vu} = -\infty$, that $w_{uz} = w_{vz} = 0$, and that $w_{zu} = w_{zv} = T + 1$.
Such a digraph $D$ can be built in polynomial time.

First, in any stable configuration for $D$, we get that $u$ and $v$ must be in two different groups because they are mutual enemies.
Furthermore, for all vertex $x\notin{\{u,v\}}$, $w_{ux} = w_{vx} = 0$ and so, the other vertices do not matter for $u$ and $v$. 
Hence for any configuration $P$ such that $X(u) \neq X(v)$, and for any $k$, there is no subset that contains $u$ or $v$ and that may break the configuration.

Moreover, we recall that $w_{zs} < 0$ for any $s \in S$. 
Consequently, in any configuration $P$ for $D$ such that $X(z) \neq X(u)$ and $X(z) \neq X(v)$, we get that $f_{z}(P) \leq 0$, whereas in any configuration $P'$ such that $X(z) = X(u)$ or $X(z) = X(v)$, we get that $f_{z}(P') \geq{1}$ because $\sum_{s \in {S}}w_{zs} = -T$ and $w_{zu}=w_{zv}=T+1$.

Hence, for any stable configuration for $D$, we conclude that either $z$ is in the same group $X_0$ as $u$, or $z$ is in the same group $X_1$ as $v$. Without any loss of generality, suppose $z$ in the same group $X_0$ as $u$.
We claim that any vertex $s$ in $S$ must be also in the same group $X_0$ as $u$ and $z$ or in the same group $X_1$ as $v$.
Indeed, otherwise we would have that $f_{s}(P) \leq 0$, since $w_{ss'} \leq 0$ for any $s'\in{S}$; however such a configuration $P$ is not $1$-stable, because if vertex $s$ joins the group $X_1$, then $f_{s}(P') > 0$, with $P'$ being the new configuration, due to the fact that $w_{sv} = T - s + 1 > 0$, and $T - s = - [ \sum_{s'\in{S}} w_{ss'} ]$.
Thus, there can be only two groups in a stable configuration of $D$.

Let $P = (X_0, X_1)$ be a configuration for $D$ such that $X(u) \neq X(v)$ and $S_i = \{ s_i \in {S} | s_i \in {X_i} \}, i \in{\{0, 1\}}$.
We can assume that $S_0, S_1$ is a partition of $S$.
Let us call $T_i = \sum_{s_i\in{S_i}}s_i$.
By definition of $w$, for any $i \in{\{0,1\}}$, for all vertex $s_i$ of $S_i$, we get that $f_{s_{i}}(P) = T + 1 - \sum_{s'_i\in{S_i}}s'_i + c_{s_i} = T - T_i + 1 +  c_{s_i} = T_{1-i} + 1 + c_{s_i}$, where the value of $ c_{s_i}\in{\{0,s_i\}}$ depends on whether $z$ is in $X_i$. 
In the same way, one can prove that if $z$ is in $X_j$ for some $j \in {\{0,1\}}$ then $f_{z}(P) = T_{1-j} + 1$.

It is important to notice that there is no subset of $V$ that breaks the configuration $P$ in creating a new group. Indeed, $u$ and $v$ cannot be part of such a subset, and it cannot contain $z$ either, because $z$ must be in the group of $u$ or in the group of $v$.
Furthermore, in any subset of $S$ that is taken as a new group, all of the individual utilities are non positive, whereas they were positive in the original groups.

Let $s \in{S}$ be a vertex of $X_i$ for some $i$, such that $T_i \leq{T_{1-i}}$. We get that $f_{s}(P) \geq{ T_{1-i} + 1}$.
Consequently, $s$ has no incentive in moving into $X_{1-i}$, whatever it is a collective move or not, because in such a new configuration $P'$ we get $f_{s}(P')  \leq{ T_i' + s + 1} \leq{ (T_i - s) + s + 1 = T_i + 1 } \leq{ T_{1-i} + 1}$.
We prove in a similar way that vertex $z$ has no incentive in leaving the group $X_i$ when  $T_i \leq{T_{1-i}}$. 
Therefore if $S_0,S_1$ is a solution for the Partition problem, then $T_0 = T_1$, that is $T_i \leq T_{1-i}$ for all $i\in{\{0,1\}}$ and so, for every $k\geq{1}$ we get that $P$ is a $k$-stable configuration for $D$.

Conversely, let us suppose that  $S_0,S_1$ is not a solution for the Partition problem, that is $T_i < T_{1-i}$ for some $i$, and let us show that $P$ is not a $1$-stable configuration for $D$.
Obviously, $z$ moves into $X_{i}$ provided $z$ belongs to $X_{1-i}$.
So we now assume that $z$ is in $X_{i}$. 
Then, any singleton $\{s_{1-i}\}$ with $s_{1-i} \in{S_{1-i}}$ breaks the configuration, because in the new configuration $P'$ where vertex $s_{1-i}$ has moved into group $X_i$ we get that $f_{s_{1-i}}(P') = T_{1-i} + 1$, that is strictly greater than $f_{s_{1-i}}(P) = T_i + 1$.

Finally, for any $k \geq 1$, there is a $k$-stable configuration $P$ for the digraph $D$ if and only if there is a solution to the Partition problem for the set $S$.
\end{proof}



\begin{proof}~[Theorem~\ref{NPKL}]
For our reduction, we use 3-coloring problem, that consists in determining whether there exists a proper $3$-coloring of a given graph $G=(V,E)$, a well known NP-complete problem~\cite{Dailey1980}.
Let $G$ be an instance for the 3-coloring problem.
We define the graph $G_c=(V,w)$ as follows.
For any vertex $v$ of $G$ there are five vertices $v_1,v_2,v_{c_1},v_{c_2},v_{c_3}$ in $V$.
We say that $v_{c_i}$ is a colored vertex, and that it has the color $i$.
There are also three special vertices in $V$, namely $c_1,c_2$ and $c_3$. 

For any vertex $v$ of $G$ we set $w_{v_1v_2} = 1$,  $w_{v_1v_{c_1}} = w_{v_1v_{c_2}} = w_{v_1v_{c_3}} = 1$,  $w_{v_2v_{c_1}} = w_{v_2v_{c_2}} = w_{v_2v_{c_3}} = 1$, and $w_{v_{c_1}v_{c_2}} = w_{v_{c_1}v_{c_3}} = w_{v_{c_2}v_{c_3}} = -\infty$.
For any vertices $u,v$ in $G$, and for any $i,j\in{\{1,2,3\}}$, $i \neq j$, we also set $w_{c_iv_{c_i}} = w_{c_iu_{c_i}} = 1$, $w_{c_ic_j} = -\infty$, and $w_{u_{c_i}c_j} = w_{u_{c_i}v_{c_j}} = -\infty$.
In the case when $u$ and $v$ are adjacent in $G$, we set the weights $w_{v_iu_j} = -\infty$, for any $i,j\in{\{1,2\}}$.
Finally, every other weight equals $0$.

We have that, for any $i$, and for any vertex $v$ of $G$, $v_{c_i}$ and $c_i$ gossip when they are not in the same group, because $w_{v_{c_i}c_i} = 1$ and their enemies are the same. 
Furthermore, vertices $v_1$ and $v_2$ gossip when they are not in the same group too.
We now suppose there exists a gossip-stable partition $P$ such that none of the vertices $v_{c_i}$, for $i\in{\{1,2,3\}}$, is in the same group as $v_1$ and $v_2$,
and we claim that there is a contradiction. 
Indeed, all the colored vertices that are in the same group as $v_1$ and $v_2$, if any, have the same color $i$, because otherwise there would be enemies in the group and $P$ would not be stable.
Thus, there is at least one vertex $v_{c_i}$ that breaks the partition by joining the group of $v_1$ and $v_2$, because $w_{v_{c_i}v_1}+ w_{v_{c_i}v_2} = 2$ whereas $w_{v_{c_i}c_i} = 1$. 
A contradiction.

Finally, there can be only three groups in any gossip-stable partition for $G_c$.
Moreover, we claim those three groups are equivalent to a proper $3$-coloring of $G$.
Indeed, suppose $P$ is a gossip-stable partition for $G_c$, with only three groups.
We have that $X(c_1),X(c_2),X(c_3)$ are pairwise different, because these three nodes are enemies.
Thus we can define a $3$-coloring of $G$ in the following way: for every vertex $v$ of $G$, $c(v) = i$ if and only if $X(v_1)$ ($= X(v_2)$) $= X(c_i)$, with $1 \leq i \leq 3$.
This $3$-coloring is proper, because otherwise there would be enemies in at least one group in $P$. 

Conversely, suppose that we have a proper $3$-coloring $c$ of $G$.
We define the partition  $P = ( X_1, X_2, X_3 )$ as follows.
For every $1 \leq i \leq 3$, $c_i \in X_i$, and any colored vertex $v_{c_i}$ is in $X_i$ too.
Furthermore, for all vertex $v$ of $G$, we have that $v_1,v_2 \in X_{c(v)}$.
By construction, there is no enemies in any of these three groups.
Moreover, we claim that there is no gossiping between any two vertices that are in different groups.
Indeed, there can only be gossiping between two vertices $v_j$ and $v_{c_i}$, with $j \in \{1,2\}$, $i \in \{1,2,3\}$, and $v_j \notin X_i$; however, the very definition of the partition $P$ implies that $v_{c_i}$ has, at least, one enemy in the group of $v_j$, namely $c(v)$, and so, there is no gossiping at all.
To conclude, we claim that any subset $S$, $|S| \leq 2$, does not break $P$.
Indeed, the vertices in $S$ can only found a new group in $P$, by the very definition of the weights $w$ in $G_c$.
Furthermore, for all vertex $u \in G_c$, we have that $f_{u}(P) \geq 1$.
Hence, there is no vertex in $G_c$ that can strictly increase its utility by a $2$-deviation, and so, $P$ is a $2$-stable, gossip stable partition.   
\end{proof}

\subsubsection{Multichannels}

\begin{proof}~[Lemma~\ref{lem:contre-exemple-q=2}]
We use the counter-example $G_{q=1}=(V_{q=1},w^{q=1})$ of Lemma~\ref{lem:counter-example} for the proof.
Note that, for each node $u \in V_{q=1}$, there exists at most one node $v \in V_{q=1}$ such that $w^{q=1}_{uv} = 0$.

In the sequel, $K$ denotes a clique of order $p > 2|V_{q=1}|$ with weight $1$.
We now assume that $\varepsilon < 1/p$ for the rest of the proof, and we construct the graph $G_{q=2}=(V_{q=2},w^{q=2})$ as follows.
We first start from $G_{q=1}$.
For any node $u \in V_{q=1}$ such that $w^{q=1}_{uv} \neq 0$ for any $v \in V_{q=1}$, we add a copy of $K$, denoted $K_{u}=(V_{u},E_{u})$.
We set $w^{q=2}_{ux} = 1$ for any $x \in V_{u}$.
For any edge $uv$ with $w^{q=1}_{uv} = 0$, we add another copy of $K$, denoted $K_{(u,v)}=(V_{(u,v)},E_{(u,v)})$.
We set $w^{q=2}_{ux} = w^{q=2}_{vx} = 1$ for any $x \in V_{(u,v)}$.
Furthermore, we set $w^{q=2}_{uv} = 1$.
Finally, the conflict graph consists of all the remaining edges.\\

\noindent
First, we get that for any $u,v \in V_{q=1}$ such that $w^{q=1}_{uv} = 0$, for any $x,y \in V_{(u,v)}$, $|X(x) \cap X(y)| = 2$  in any $3$-stable configuration $C$ with $2$ channels for $G_{q=2}$.

Indeed, we have that any $x \in V_{(u,v)}$ cannot be in a group containing any node $y \in V_{q=2} \setminus (V_{(u,v)} \cup \{u,v\})$.

We first prove that in any $3$-stable configuration $C$ with $2$ channels for $G_{q=2}$, we have $|X(x) \cap X(y)| \neq 1$, for any $x,y \in V_{(u,v)}$.
By contradiction.
Assume there exist $x,y \in V_{(u,v)}$ such that $|X(x) \cap X(y)| = 1$.
Without loss of generality, assume that $f_x(C) \geq f_y(C)$.
Then, there exists a $1$-deviation: $y$ reaches the group $X_i$ such that $x \in X_i$ and $y \notin X_i$.
In that case, $f_y(C') > f_y(C)$, where $C'$ is the configuration for $G_{q=2}$ after the previous $1$-deviation.
Indeed, $f_y(C') = f_x(C) + \epsilon$.

Thus, we assume that, in any $3$-stable configuration $C$ with $2$ channels for $G_{q=2}$, we have $|X(x) \cap X(y)| \neq 1$, for any $x,y \in V_{(u,v)}$.
Furthermore, if there exist $x,y \in V_{(u,v)}$ such that $|X(x) \cap X(y)| = 0$, then there exist two groups $X_{i}$ and $X_{j}$ such that $x \in X_{i}$, $x \in X_{j}$, and $X_{i} = X_{j}$, and there exist two groups $X_{i'}$ and $X_{j'}$ such that $y \in X_{i'}$, $y \in X_{j'}$, and $X_{i'} = X_{j'}$.
Thus, $f_x(C) = (1+\epsilon)(|X_{i}|-1)$, and $f_y(C) = (1+\epsilon)(|X_{i'}|-1)$.
Without loss of generality, suppose $f_x(C) \geq f_y(C)$.
Then, there exists a $1$-deviation: $y$ leaves $X_{i'}$ and reaches $X_{i}$.
We get $f_y(C') = |X_{i}| + |X_{i'}| - 1 > f_y(C)$ because $\epsilon < \frac{1}{p}$.
Thus, in any $3$-stable configuration $C$ with $2$ channels for $G_{q=2}$, we have for any $x,y \in V_{(u,v)}$, $|X(x) \cap X(y)| = 2$, and so, the claim is proved.\\

\noindent
Then, by the choice of $0 < \varepsilon < 1/p$ with $p > 2|V_{q=1}|$, there exist two groups $X_{i}$ and $X_{j}$ such that $X_{i} = V_{(u,v)}$ and $X_{j} = V_{(u,v)} \cup \{u,v\}$.
Indeed, otherwise there exists a $1$-deviation for $C$, because by construction $|V_{(u,v)}| > f_{u}(C')$, with $C'$ being any configuration with $2$ channels for $G_{q=1}$.
We obtain the same result for copies $K_{u}$.
To summarize, after setting the groups from the previously necessary conditions for having a $3$-stable configuration with $2$ channels for $G_{q=1}$, each node $u \in V \setminus V_{q=1}$ belongs to exactly $2$ groups, and each node $u \in V_{q=1}$ belongs to exactly $1$ group.

Thus, there exists a $3$-stable configuration for $G_{q=2}$ with $2$ channels if, and only if, there exists a $3$-stable configuration for $G_{q=1}$ with single channel.
Indeed, any two nodes $u,v \in V_{q=1}$ such that $w^{q=1}_{uv} = 0$, are already in a same group.
Thus, by the choice of the function $h$, this mimics a weight $\varepsilon$ between these two nodes for determining the other groups, that is quite the same as a weight $0$, because $\varepsilon$ is arbitrarily small (see Corollary \ref{cor:mn}).
Finally, there does not exists a $3$-stable configuration for $G_{q=2}$ with $2$ channels because there does not exist a $3$-stable configuration for $G_{q=1}$ with single channel by Lemma~\ref{lem:counter-example}.
\end{proof}



\begin{proof}~[Theorem~\ref{the:chaotic-behavior}]

\begin{claim}
\label{lemma:counter-symmetric-1}
There exists a graph $G=(V,w)$ with $\mathcal{W}=\{-\infty,-4,2,6,7\}$.such that any partition of the nodes for $G$ is not $2$-stable.
\end{claim}

\begin{proofclaim}
Let $G_1=(V,w)$ be the graph built as follows.
Let $V=\{u_{1},u_{2},u_{3},u_{4}\}$.
We set $w_{u_{1}u_{2}} = w_{u_{1}u_{3}} = -\infty$, $w_{u_{1}u_{4}}=7$, $w_{u_{2}u_{4}}=6$, $w_{u_{3}u_{4}}=2$, and $w_{u_{2}u_{3}}=-4$.

For any $k \geq 1$, any $k$-stable partition $P$ for $G_1$ is such that $f_{u}(P) \geq 0$ for any $u \in V$.
Thus, for any $k \geq 1$, the following partitions for $G_1$ are not $k$-stable:
\begin{itemize}
\item $P=(\{u_{1}\},\{u_{2},u_{3},u_{4}\})$ as $f_{u_3}(P) = -2 < 0$;
\item $P=(\{u_{1}\},\{u_{2},u_{3}\},\{u_{4}\})$ as $f_{u_3}(P) = -4 < 0$;
\item $P=(\{u_{1},u_{4}\},\{u_{2},u_{3}\})$ as $f_{u_3}(P) = -4 < 0$;
\item any partition $P=(X_{1},X_{2},X_{3},X_{4})$ ($X_{i}$ may be empty for some $i$, $1 \leq i \leq 4$) such that $u_{1},u_{2} \in X_{i}$ and/or $u_{1},u_{3} \in X_{i}$ for some $i$, $1 \leq i \leq 4$.
\end{itemize}

Thus, it remains to prove the result for the other partitions for $G_1$.
\begin{itemize}
\item Consider the partition $P=(\{u_{1}\},\{u_{2}\},\{u_{3}\},\{u_{4}\})$ for $G_1$.
$P$ is not $1$-stable because $u_{4}$ can reach group $\{u_{1}\}$.
Indeed, $f_{u_4}(P') = f_{u_4}(P) + 7 = 7$ where $P'=(\{u_{1},u_{4}\},\{u_{2}\},\{u_{3}\})$ is the resulting partition after this $1$-deviation.

\item Consider the partition $P=(\{u_{1},u_{4}\},\{u_{2}\},\{u_{3}\})$ for $G_1$.
$P$ is not $2$-stable because $u_{2}$ and $u_{4}$ can reach group $\{u_{3}\}$.
Indeed, $f_{u_2}(P') = f_{u_2}(P) + 6 - 4 = f_{u_2}(P) + 2 = 2$ and $f_{u_4}(P') = f_{u_4}(P) -7 + 6 + 2 = f_{u_4}(P) + 1 = 8$ where $P'=(\{u_{1}\},\{u_{2},u_{3},u_{4}\})$ is the resulting partition for $G_1$ after this $2$-deviation.

\item Consider the partition $P=(\{u_{1}\},\{u_{2},u_{4}\},\{u_{3}\})$ for $G_1$.
$P$ is not $1$-stable because $u_{4}$ can reach group $\{u_{1}\}$.
Indeed, $f_{u_4}(P') = f_{u_4}(P) -6 + 7 = f_{u_4}(P) + 1 = 7$ where $P'=(\{u_{1},u_{4}\},\{u_{2}\},\{u_{3}\})$ is the resulting partition for $G_1$ after this $1$-deviation.

\item Consider the partition $P=(\{u_{1}\},\{u_{2}\},\{u_{3},u_{4}\})$ for $G_1$.
$P$ is not $1$-stable because $u_{4}$ can reach group $\{u_{1}\}$.
Indeed, $f_{u_4}(P') = f_{u_4}(P) -2 + 7 = f_{u_4}(P) + 5 = 7$ where $P'=(\{u_{1},u_{4}\},\{u_{2}\},\{u_{3}\})$ is the resulting partition for $G_1$ after this $1$-deviation.
\end{itemize}

Finally, any partition $P$ for $G_1$ is not a $2$-stable partition for $G_1$.
\end{proofclaim}

\begin{claim}
\label{cor:base}
Let us fix $h$ is as $h(0,w) = 0$, $h(g+1,w) = (1 + \varepsilon g)w$, where $\varepsilon > 0$ is arbitrarily small.
Given any integer $q \geq 1$, there exists a graph $G_q=(V_q,w^q)$ with $\mathcal{W}=\{-\infty,-4,2,6,7,N\}$ such that any configuration $C$ with $q$ channels for $G_q$ is not $2$-stable, whereas there always exists a $2$-stable configuration $C'$ with $q' \neq q$ channels for $G_q$.
\end{claim}

\begin{proofclaim}
We arbitrarily select $\epsilon < 1/N$.
Let $G_1 = (V_1,w^1)$ be the graph that is presented in the proof of Claim~\ref{lemma:counter-symmetric-1} (we re-use the same notations for the vertices in $G_1$).
We build the graph $G_q$ as follows.
First, the set of vertices is $V_q = V_1 \cup \{x_1, \ldots, x_{q-1}\}$.
For any $u,v \in V_1$, $w^q_{uv} = w^1_{uv}$.
Furthermore, for any $1 \leq i < j \leq q-1$, $w^q_{x_i u_4} = w^q_{x_j u_4} = N$, $w^q_{x_i x_j} = w^q_{x_i v} = w^q_{x_j v} = -\infty$ for any $v \in V_1 \setminus \{u_4\}$.

We have that, for any $q' < q$, the following configuration
is $2$-stable with $q'$ channels in the graph $G_q$:

$
\begin{array}[c]{rl}
C_{q'} = & ( \{u_4,x_1\}, \ldots, \{u_4,x_{q'}\}, \{u_1\}^1, \ldots, \{u_1\}^{q'}, \{u_2\}^1, \ldots, \{u_2\}^{q'}, \{u_3\}^1, \ldots, \{u_3\}^{q'}, \\ & \{x_1\}^1, \ldots, \{x_1\}^{q'-1}, \ldots, \{x_{q'}\}^1, \ldots, \{x_{q'}\}^{q'-1}, \\ & \{x_{q'+1}\}^1, \ldots, \{x_{q'+1}\}^{q'}, \ldots, \{x_{q-1}\}^1, \ldots, \{x_{q-1}\}^{q'}) \ff
\end{array}
$

However, in any stable configuration $C_q$ with $q$ channels for $G_q$, we necessarily have $\{u_4,x_i\}$ as a group in $C_q$ (without repetitions), for any $1 \leq i \leq q-1$. As there does not exist any $2$-stable configuration with single channel for $G_1$, there does not exist any $2$-stable configuration with $q$ channels for $G_q$.

If $q' = q+1$, then the following configuration is $2$-stable with $q'$ channels for the graph $G_q$:

$
\begin{array}[c]{rl}
	C_{q'} = & ( \{u_4,x_1\}, \ldots, \{u_4,x_{q-1}\},\{u_4,u_1\}, \{u_4,u_2\}, \{u_1\}^1, \ldots, \{u_1\}^{q'-1}, \{u_2\}^1, \ldots, \{u_2\}^{q'-1}, \\ & \{u_3\}^1, \ldots, \{u_3\}^{q'}, \{x_1\}^1, \ldots, \{x_1\}^{q'-1}, \ldots, \{x_{q-1}\}^1, \ldots, \{x_{q-1}\}^{q'-1})\ff
\end{array}
$ 

Finally, if $q' \geq q+2$, then the following configuration is $2$-stable with $q'$ channels in $G_q$:

$
\begin{array}[c]{rl}
C_{q'} = & ( \{u_4,x_1\}^1, \ldots, \{u_4,x_1\}^{q'-q-1}, \\ & \{u_4,x_2\}, \ldots, \{u_4,x_{q-1}\},\{u_4,u_1\}, \{u_4,u_2\}, \{u_4,u_3\}, \{u_1\}^1, \ldots, \\& \{u_1\}^{q'-1}, \{u_2\}^1, \ldots, \{u_2\}^{q'-1}, \{u_3\}^1, \ldots, \{u_3\}^{q'-1}, \{x_1\}^1, \ldots, \\ & \{x_1\}^{q+1}, \{x_2\}^1, \ldots, \{x_2\}^{q'-1}, \ldots, \{x_{q-1}\}^1, \ldots, \{x_{q-1}\}^{q'-1}) \ff
\end{array}
$ 

\end{proofclaim}

For any $1 \leq t \leq l$, let $G_{q_{i_t}}$ be the graph in the proof of Claim~\ref{cor:base} such that there exists a $2$-stable configuration with $q'$ channels for $G_{q_{i_t}}$ if, and only if, $q' \neq q_{i_t}$.
By convention, $G_0$ is any graph such that there always is a $2$-stable configuration with $q$ channels for $G_0$; for instance, we may assume that all the weights for $G_0$ are positive.

Then we build a graph $G$ from $G_0, G_{q_{i_1}}, \ldots G_{q_{i_l}}$, such that any vertices $u,v$ that are not in the same graph $G_q$ are enemies (e.g. $w_{uv} = -\infty$).
By doing so, every graph $G_q$ is independent of the other ones, and so, there is a $2$-stable configuration with $q'$ channels for $G$ if, and only if, there exists a $2$-stable configuration with $q'$ channels for every $G_{q_{i_t}}$, $1 \leq t \leq l$. 
\end{proof}

\subsubsection{Multi-modal relationship}
\label{annexe-multi-modal-annexe}

Our model so far is heavily biased towards pairwise relationship, as the utility of a player depends on the sum of her interaction with all other members of the groups. In reality, more subtle interaction occur: one may be interested to interact with either friend $u$ or $v$, but would not like to join a group where both of them are present. 

Our analysis also generalize to this case. 
We define a coloring game over a weighted undirected hypergraph $H = (V,w)$, with $w$ a weight function from $\mathcal{P}_t(V)$ to $\mathbb{N}$, where $\mathcal{P}_t(V)$ denotes the subsets in $V$ of size at most $t$.
Note that $t = \infty$ is allowed by our model.
The utility of a node $u$ in partition $P$ now equals: 
\(\sum_{ \{u\} \subseteq A \subseteq X(u), \\|A| \leq t} w_{A}.\)

We prove in Lemma~\ref{lemma:hypergraph-1-stable} that $1$-stability is guaranteed. 
Interestingly, the proof uses a different potential function than the global utility: $\phi (P) = \sum_{i=1}^n \sum_{A \subseteq X_i,\\ |A| \leq t} w_{A}$.

\begin{lemma}
\label{lemma:hypergraph-1-stable}
For every hypergraph $H = (V,w)$, there exists a $1$-stable partition.
\end{lemma}

\begin{proof}~[Lemma~\ref{lemma:hypergraph-1-stable}]
Let $P$ be a partition of the nodes for $H$, let $u$ be a vertex that breaks $P$ and let $P'$ be the resulting partition.
We use the potential function $\phi (P) = \sum_{i=1}^n \sum_{A \subseteq X_i,\\ |A| \leq t} w(A)$ for the proof.
Observe that $\phi(P') - \phi(P) = f_u(P') - f_u(P) > 0$.
As $\phi$ is upper-bounded, the dynamic of the coloring game always converges.
\end{proof}

Moreover, all results on uniform games hold.
Formally, for any group $X_j \in P$, the vertices $u$ in $X_j$ have a utility $\sum_{\{u\} \subset A \subseteq X_j,\\ |S| \leq t} w_{A} = \sum_{i=1}^{t-1} \binom{|X_j| - 1}{i}$, which is monotically increasing in the size of their group.

Given a general hypergraph $H=(V,w)$, we define the friendship hypergraph $H^+$ in the same way as the friendship graph defined in Section~\ref{subsec:-m01}.
A Berge cycle in $H^+$ is a sequence $v_0, e_0, v_1, \ldots, v_{p-1},e_{p-1}, v_0$, such that:
\begin{itemize}
\item $v_0, v_1, \ldots, v_{p-1}$ are pairwise different vertices;
\item $e_0, e_1, \ldots, e_{p-1}$ are pairwise different hyperedges;
\item for all $0 \leq i \leq p-1$, $v_i, v_{i+1} \in E_i$. 
\end{itemize}

If subsets of nodes have weights in $\mathcal{W}=\{-\infty,0,1\}$, there still exists a $k$-stable partition for a hypergraph $H$ if the girth of the friendship hypergraph $H^{+}$ is at least $k+1$. Note that it does not imply $2$-stability in general because the girth of a hypergraph may be $2$.
\begin{lemma}
\label{lem:hypergraph-1}
Let $H=(V,E)$ be an acyclic hypergraph, with $c$ connected components.
Then $|V| = \sum_{e \in E} (|e|-1) + c$.
\end{lemma}
\begin{proof}~[Lemma~\ref{lem:hypergraph-1}]
First we arbitrarily order the vertices in $V$, denoted $v_0, v_1, \ldots, v_{n-1}$.
We build an intersection graph $G=(V_H,E_H)$ as follows.
$V_H = \{ s_0, s_1, \ldots, s_{n-1}\}$, such that for all $0 \leq i \leq n-1$, $s_i = \{ e \in E : \{v_0, \ldots, v_i\} \cap e = \{v_i\}\}$.
Moreover, for all $0 \leq i < j \leq n-1$, the edge $s_is_j = s_js_i$ is in $E_H$ if and only if there is $e \in s_i$ such that $v_j \in e$.
Observe that such a hyperedge $e$ is necessarily unique because $H$ is acylic.
Let us call it $e_{i,j} = e_{j,i}$.
Thus, we get $|E_H| = \sum_{e \in E} (|e|-1)$.
Moreover, we claim that $G_H$ has the same number of connected components as $H$.
Indeed, if $s_i$ and $s_j$ are adjacent in $G_H$ then $v_i$ and $v_j$ both belong to $e_{i,j}$ in $H$ by definition of $G_H$.
Conversely, suppose that $v_i$ and $v_j$ are adjacent in $H$.
Let $e$ be a hyperedge such that $v_i,v_j \in e$.
By definition of $G_H$, there exists $0 \leq i_0 \leq \min (i,j)$ such that $e \in s_{i_0}$.
Furthermore, $s_{i_0}$ is adjacent to both $s_i$ and $s_j$.
As a consequence, we get there are $c$ connected components in $G_H$.
Finally, we claim that $G_H$ is acyclic, because any cycle in $G_H$ would give a Berge cycle in $H$.
Therefore, we have $|V_H| = |E_H| + c$.
\end{proof}

\begin{corollary}
Given an integer $k \geq 1$ and a hypergraph $H = (V,w)$ with $\mathcal{W} = \{-\infty,0,1\}$, there is a $k$-stable partition for $H$ if the girth of the friendship hypergraph $H^{+}$ is at least $k+1$.
\end{corollary}

We will omit the proof as it is quite the same as for Theorem~\ref{thm:stability:ktree}.
In particular, there always exists a $2$-stable partition if the friendship hypergraph is simple, that is the intersection of two different hyperedges has size at most $1$.

\subsection{Appendix 4: Efficiency Stability Trade-off}
\label{annexe:efficiency}

\def\MIS{{\sc Maximum Independent Set problem}\xspace}



\begin{proof}~[Lemma~\ref{lem:mcp-not-multiplicative}]
Let us constrain ourselves to the set of weights $\{-\infty,1\}$.
We consider a conflict graph $G^-=(V,E)$.
Given a partition $P=\{X_1,\ldots,X_n\}$ for $G^-$, the trick is to consider $\phi(P) = \sum_{i=1}^n |X_i|^2$ rather than $f(P)$; note that $\phi(P) = f(P) + n$ in this case.
Also, we will denote by $OPT(G^-)$ the optimal value taken by $\phi$ for any maximum partition of the nodes for $G^-$; we will denote by $MIS(G^-)$ the size of any maximum independent set of $G^-$.
Observe that such an independent set represents a clique of friends in the coloring game.

By contradiction, suppose there exists a polynomial-time algorithm $\mathcal{A}$, such that for any conflict graph $G^-$ we have $\phi(\mathcal{A}(G^-)) \geq n^{\epsilon - 1} OPT(G^-)$.
We claim we can use algorithm $\mathcal{A}$ to compute an independent set of $G^-$.
To do this, it suffices to take the largest group in $\mathcal{A}(G^-)$.
Moreover, we know by~\cite{Kar72}, the \MIS cannot be approximated within any $n^{\frac{\epsilon}{2}- 1}$-ratio in polynomial time, unless P=NP.
Then, if P $\neq$ NP, there are graphs $G^-$ such that the largest group in $\mathcal{A}(G^-)$ has size lesser than $n^{\frac{\epsilon}{2} - 1}MIS(G^-)$.
Especially, for such a graph $G^-$, we have $\phi(\mathcal{A}(G^-)) < n.(n^{\frac{\epsilon}{2} - 1}MIS(G^-))^2 = n^{\epsilon - 1}(MIS(G^-))^2$.
However, we can always define a partition of the nodes using a maximum independent set of $G^-$; the remaining vertices, if any, are supposed to be placed in singleton groups.
Moreover, for the partition $P_{MIS}$ that we get for $G^-$ we have $\phi(P_{MIS}) = (MIS(G^-))^2$.
So, we have $\phi(\mathcal{A}(G^-)) <  n^{\epsilon - 1} OPT(G^-)$, which is a contradiction unless P = NP.
\end{proof}

%

\textbf{Transformation for the proof of Theorem~\ref{the:equivalence-maximum}.}
We formally define the transformation as follows:
\begin{definition}
\label{def:transformation}
For any graph $G=(V,w)$ with $\setW \subseteq \{-\infty\} \cup \NatInt$, 
let $G'=(V',w')$ be as follows. \\
\textbf{a)} For any $u \in V$, we have $q$ nodes $u_{1},u_{2}, \ldots, u_{q} \in V'$ and we set $w'_{u_{i} u_{j}} = -\infty$ for any $i,j$, $1 \leq i < j \leq q$; \\
\textbf{b)} For any $u,v \in V$ such that $w_{uv} \in \{-\infty,0\}$, we set $w'_{u_{i} v_{j}} = w_{uv}$ for any $i,j$, $1 \leq i, j \leq q$; \\
\textbf{c)} For any $u,v \in V$ such that $w_{uv} > 0$, for any $g$, $1 \leq g \leq q$, we have two nodes $uv_{g,1}, uv_{g,2} \in V'$, as shown in Figure~\ref{fig:transformation}, such that for any $i, j$, $1 \leq i, j \leq q$: 
\\
$\bullet$ $w'_{u_{i} uv_{g,1}} = w'_{v_{i} uv_{g,1}} = (h(g,w_{uv}) - h(g-1,w_{uv}))/2$; \\
$\bullet$ $w'_{u_{i} v_{j}} = 0$; \\
$\bullet$ $w'_{uv_{g,1} uv_{g,2}} = 3 (h(g,w_{uv}) - h(g-1,w_{uv}))/4$; \\
$\bullet$ $w'_{uv_{g,2} y} = -\infty$, $\forall y \notin \{ uv_{g,1},uv_{g,2} \}$; \\
$\bullet$ $w'_{uv_{g,1} uv_{g',1}} = -\infty$, $\forall g' \neq g$, $1 \leq g' \leq q$; \\
$\bullet$ $w'_{uv_{g,1} y} = 0,\forall y\! \notin\!\{u_{i}, v_{i}, uv_{g',1}, uv_{g,2}, 1\!\!\leq i, g'\!\!\!\leq q\}$.
\end{definition}

\begin{figure}[t]
\centerline{\includegraphics[width=0.5\linewidth]{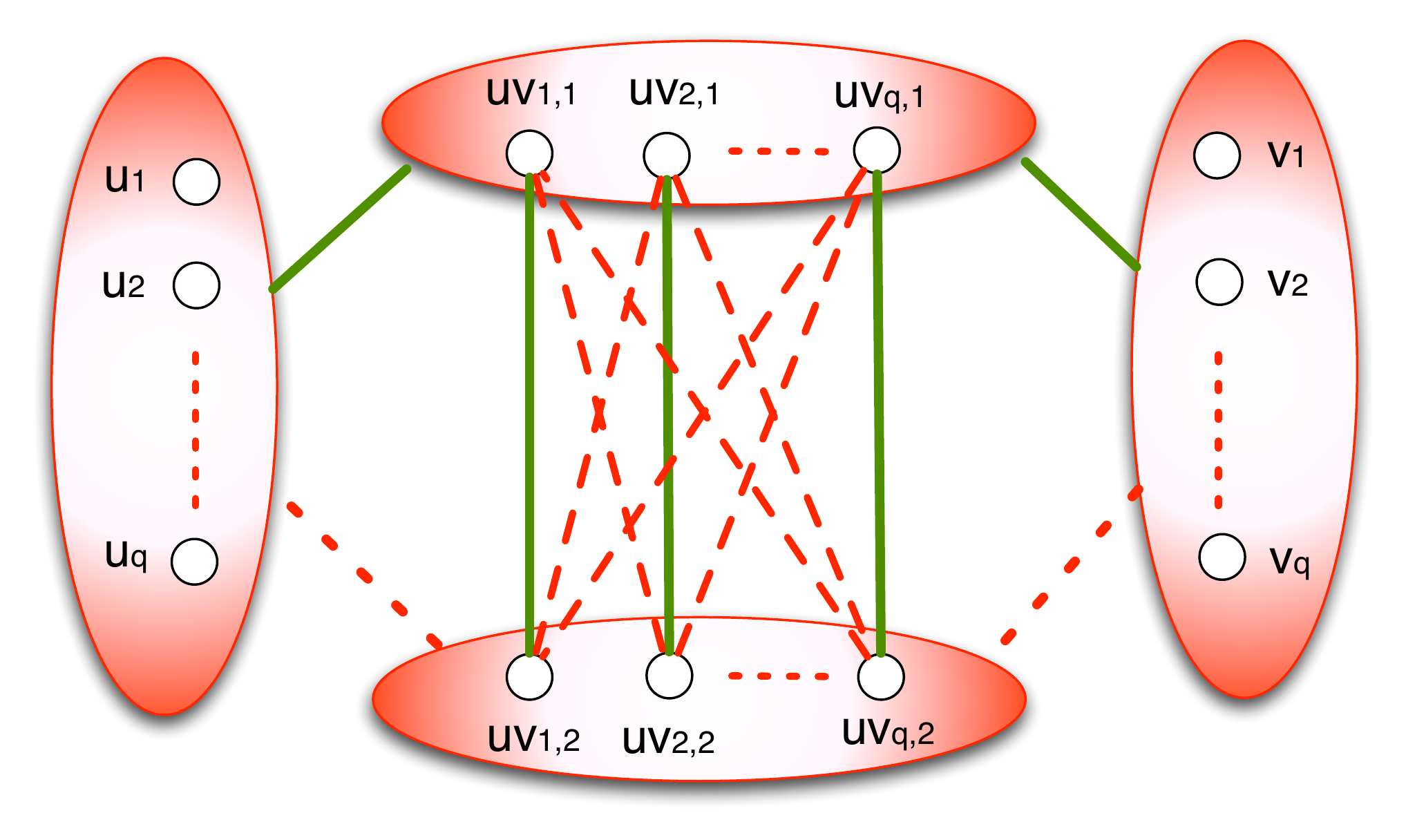}}
\caption{Representation in the graph G' of an edge uv with positive weight in the transformation.}
\label{fig:transformation}
\end{figure}

Let $P'$ be a maximum partition of the nodes for $G'$. 
Note that $X(u_{i}) \neq X(u_{j})$ for any $i\neq j$, and $X(u_{i}) = X(v_{j})$ if, and only if, a unique $uv_{g,1}$ exists such that $X(u_{i}) = X(uv_{g,1}) = X(v_{j})$. 
From $P'$ we can then build a configuration $C$ with $q$ channels for $G$ by replacing all nodes $u_i$ by $u$ and deleting all nodes $uv_{g,1}$ and $uv_{g,2}$. 
The proof concludes by proving that $P'$ is not maximum if $C$ isn't.
Formally:

\begin{proof}~[Theorem~\ref{the:equivalence-maximum}]
Let $P'$ be any maximum partition for $G'$.
We prove several properties for $P'$.

a) For any $u \in V$, we have $X(u_{i}) \neq X(u_{j})$ for any $i, j$, $1 \leq i < j \leq q$, because $w'_{u_{i} u_{j}} = -\infty$.

b) For any $u,v \in V$ such that $w_{uv} = -\infty$, we have $X(u_{i}) \neq X(v_{j})$ for any $i, j$, $1 \leq i, j \leq q$, because $w'_{u_{i} v_{j}} = -\infty$.

c) For any $u,v \in V$ such that $w_{uv} > 0$, if there exist $i, j$, $1 \leq i, j \leq q$, such that $X(u_{i}) = X(v_{j})$, then we claim there exists $g$, $1 \leq g \leq q$, such that $X(u_{i}) = X(v_{j}) = X(uv_{g,1})$.
By contradiction.
If $X(u_{i}) = X(v_{j})$ and $X(u_{i}) = X(v_{j}) \neq X(uv_{g,1})$ for any $g$, $1 \leq g \leq q$, then we show there exists a  partition $P''$ for $G'$ such that $f(P'') > f(P')$.
Indeed, we proved before that $X(u_{i'}) \neq X(u_{j'})$ and $X(v_{i'}) \neq X(v_{j'})$ for any $i', j'$, $1 \leq i' < j' \leq q$.
Thus, there exists $g$, $1 \leq g \leq q$, such that $X(uv_{g,1}) \neq X(u_{i'})$ and $X(uv_{g,1}) \neq X(v_{i'})$ for any $i'$, $1 \leq i' \leq q$.
We get two cases:
\begin{enumerate}
\item if $X(uv_{g,1}) \neq X(y)$ for any $y \in V'$, then we add $uv_{g,1}$ in the group of $u_{i}$ and $v_{j}$ and we get $f(P'') = f(P') + 2 (h(g,w_{uv}) - h(g-1,w_{uv}))$.
\item if $X(uv_{g,1}) = X(uv_{g,2})$, then $X(uv_{g,1}) \neq X(y)$ for any $y \in V'$ by construction of $G'$.
Thus, we add $uv_{g,1}$ in the group of $u_{i}$ and $v_{j}$ and we get $f(P'') = f(P') + 2 (h(g,w_{uv})-h(g-1,w_{uv})) - 6 ( h(g,w_{uv})-h(g-1,w_{uv}) )/4 > f(P')$.
\end{enumerate}

d) For any $u,v \in V$ such that $w_{uv} > 0$, if there exist $i, j, g$, $1 \leq i, j, g \leq q$, such that $X(u_{i}) = X(v_{j}) = X(uv_{g,1})$, then for any $g' < g$, we claim there exists $i', j'$, $1 \leq i', j' \leq q$, such that $X(u_{i'}) = X(v_{j'}) = X(uv_{g',1})$.
By contradiction.
If there exists $g' < g$ such that $X(u_{i'}) \neq  X(uv_{g',1})$ or $X(v_{j'}) \neq X(uv_{g',1})$ for any $i', j'$, $1 \leq i', j' \leq q$, then we again show there exists a partition $P''$ for $G'$ such that $f(P'') \geq f(P')$.
Indeed, as $h$ is concave, we have $h(g',w_{uv})-h(g'-1,w_{uv}) \geq h(g,w_{uv})-h(g-1,w_{uv})$ because $g' < g$.
There are three cases.
\begin{enumerate}
\item Without loss of generality, if for some $i'$, $1 \leq i' \leq q$, $X(u_{i'}) = X(uv_{g',1})$ and so, $X(v_{j'}) \neq X(uv_{g',1})$ for any $j'$, $1 \leq j' \leq q$, then we add $uv_{g',1}$ to the group of $u_{i}$ and $v_{j}$, and we add $uv_{g,1}$ to the group with $uv_{g,2}$.
We get the partition $P''$ such that $f(P'') = f(P') - (h(g',w_{uv})-h(g'-1,w_{uv})) + 2(h(g',w_{uv})-h(g'-1,w_{uv})) - 2(h(g,w_{uv})-h(g-1,w_{uv})) +  6 (h(g,w_{uv})-h(g-1,w_{uv}))/ 4 > f(P')$, because $X(uv_{g',1}) \neq X(uv_{g',2})$ and because $h$ is such that $h(g',w_{uv})-h(g'-1,w_{uv}) \geq h(g,w_{uv})-h(g-1,w_{uv})$.
\item If $X(u_{i'}) \neq X(uv_{g',1})$, $X(v_{j'}) \neq X(uv_{g',1})$ for any $i', j'$, $1 \leq i', j' \leq q$, and $X(uv_{g',1}) \neq X(uv_{g',2})$, then we add $uv_{g',1}$ to the group of $u_{i}$ and $v_{j}$, and we add $uv_{g,1}$ to the group with $uv_{g,2}$.
We obtain the partition $P''$ such that $f(P'') = f(P') - 2(h(g,w_{uv})-h(g-1,w_{uv})) + 2(h(g',w_{uv})-h(g'-1,w_{uv})) + 6 (h(g,w_{uv})-h(g-1,w_{uv}))/4 > f(P')$ by the choice of $h$.
\item If $X(u_{i'}) \neq X(uv_{g',1})$, $X(v_{j'}) \neq X(uv_{g',1})$ for any $i', j'$, $1 \leq i', j' \leq q$, and $X(uv_{g',1}) = X(uv_{g',2})$, then we add $uv_{g',1}$ to the group of $u_{i}$ and $v_{j}$, and we add $uv_{g,1}$ to the group with $uv_{g,2}$.
We obtain the partition $P''$ such that $f(P'') = f(P') - 2(h(g,w_{uv})-h(g-1,w_{uv})) + 2(h(g',w_{uv})-h(g'-1,w_{uv})) - 6 (h(g',w_{uv})-h(g'-1,w_{uv}))/4 + 6 (h(g,w_{uv})\\-h(g-1,w-{u,v}))/4 > f(P')$ by the choice of $h$.
\end{enumerate}

To summarize, any maximum partition $P'$ for $G'$ is such that:
for any $u \in V$, $X(u_{i}) \neq X(u_{j})$ for any $i \neq j$, $1 \leq i,j \leq q$;
for any $u, v \in V$ such that $w_{uv} = -\infty$, $X(u_{i}) \neq X(v_{j})$ for any $i, j$, $1 \leq i,j \leq q$;
for any $u, v \in V$ such that $w_{uv} > 0$, then there exists $g$, $0 \leq g \leq q$, such that for any $g' \leq g$, there exists $i, j$, $1 \leq i, j \leq q$, such that $X(u_{i}) = X(v_{j}) = X(uv_{g',1})$, and for all $g' > g$, $X(uv_{g',1}) = X(uv_{g',2})$ and so, $X(uv_{g',1}) = \{uv_{g',1},uv_{g',2}\}$.

We get that any maximum partition $P'$ for $G'$ is such that one can always write $f(P') = c + f'(P')$ where $c$ is the constant $\sum_{u,v \in V, w_{uv} > 0} \sum_{g=1}^{q}  6(h(g,w_{uv}) -h(g-1,w_{uv}))/4$, and $f'(P') = \sum_{u,v \in V, w_{uv} > 0} \sum_{g=1}^{q} \ind{\exists i,j, X(u_{i}) = X(u_{j})} (h(g,w_{uv})-h(g-1,w_{uv}))/2$.

We now construct the corresponding configuration $C$ with $q$ channels for $G$.
For any group $X_{j} \in P$, we construct the group $X_{j}' \in C$ as follows.
For any $u \in V$, $u \in X_{j}'$ if, and only if, there exists $i$, $1 \leq i \leq q$, such that $u_{i} \in X_{j}$.
Clearly, the construction of $G'$ and the construction of $C$ can be done in linear time in the size of $G$.

By construction, we get that $f(C) = 4 f'(P')$.

We now prove that if $P'$ is maximum for $G'$, then $C$ is maximum for $G$.
First suppose $P'$ is maximum for $G'$.
By contradiction.
Suppose there exists a configuration $C'$ with $q$ channels for $G$ such that $f(C') > f(C)$.
We prove that there exists a partition $P''$ for $G'$ such that $f(P'') > f(P')$.
We construct $P''$ as follows.
For any group $X_{i_{q}} \in C'$, we construct the group $X_{i_{q}}'$ for $P''$ as follows:
sequentially, for any $u,v \in X_{i_{q}}$, there exist $i, j$, $1 \leq i,j \leq q$ such that $u_{i}$ and $u_{j}$ are not yet in one group in $P''$.
If there exist, we choose $u_{i}$ and/or $v_{j}$ that are not yet in group $X_{i_{q}}'$.
Furthermore, if $w_{uv} > 0$, we choose the smallest $g$, $1 \leq g \leq q$, such that $uv_{g,1}$ is not yet in one group in $P''$.
We add $u_{i}$, $u_{j}$, and $uv_{g,1}$ in group $X_{i_{q}}'$.

We get that $f(P'') = c + f'(P'')$ where $f'(P'') = f(C')/4$ by construction. 

Thus, we get $f(P'') > f(P')$ because $f'(P'') > f'(P')$.
A contradiction because $P'$ is maximum for $G'$.
Finally, $C$ is maximum for $G$.
\end{proof}


%

\begin{proof}~[Theorem~\ref{the:min-q-pas-APX}]
We construct $G=(V,w)$ as follows.
Let $G_{1}=(V_{1},w^{1})$ be any graph such that $w^{1}_{uv} \in \{-\infty,0\}$ for any $u,v \in V_{1}$.
Let $V=V_{1} \cup \{u, v_{1}, v_{2}, \ldots, v_{|V_{1}|}\}$.
We set $w_{ux} = N$ for any $x \in V_{1}$.
We set $w_{uv_{i}} = 1$ for any $i$, $1 \leq i \leq |V_{1}|$.
We set $w_{v_{i}x} = -\infty$ for any $x \in V_{1}$ and for any $i$, $1 \leq i \leq |V_{1}|$.

Let $U = 2 |V_{i}|$.
If there exists a configuration $C$ with $q$ channels for $G$ such that $f(C) \geq U$, then we claim $|X(u) \cap X(x)| \geq 1$ for any $x \in V \setminus \{u\}$.
Indeed, there are exactly $2 |V_{i}|$ edges with weights $1$.
Thus, by construction of $G$, the problem remains to minimize the number of groups containing node $u$.
Clearly by construction, we have $|V_{i}|$ groups of size $2$ composed of $\{u,v_{i}\}$ for any $i$, $1 \leq i \leq |V_{i}|$.
Then, $u$ must be in one group with each $x \in V_{1}$.
As $w^{1}_{uv} \in \{-\infty,0\}$ for any $u,v \in V_{1}$, then we have to partition the nodes $V_{1}$ into groups without arcs with weight $-\infty$, and we have to minimize the number of groups.
Said differently, the problem remains to find the chromatic number of the conflict graph.
As for all $\varepsilon > 0$, approximating the chromatic number within $n^{1-\varepsilon}$ is NP-hard~\cite{Zuck}, then we get the result.
\end{proof}

\begin{proof}~[Lemma~\ref{poa:infinite-ratio}]
a) Case $\{-\infty,0,b\} \subseteq \mathcal{W}$.

Let $n' \geq \sqrt{R/2b}$ be an integer.
We now build $G = (V,w)$ with $V = V_1 \cup V_2 \cup \{ v_{1}, v_{2}\}$, such that $|V_1| = |V_2| = n'$. 
For any $u,v \notin\{ v_{1}, v_{2}\}$, we set $w_{uv} = b$ if $u$ and $v$ are not in the same subset $V_i$, and $w_{uv} = 0$ otherwise. 
Moreover, $v_{1}$ ($v_{2}$, respectively) is enemy with any vertex of $V_{2}$ ($V_{1}$, respectively), and for any $u \in V_{1}$, any $v \in V_2$, we have $w_{uv_{1}} = w_{vv_{2}} = w_{v_1v_2} = 0$. 
Partition $P = ( V_1 \cup \{v_{1}\}, V_2 \cup \{v_{2}\} )$ is $1$-stable for $G$, and $f(P) = 0$.\footnote{Note we can make $q$ copies of each groups and then, the $1$-stability also holds, for instance, if we choose $h(0,w) =0$, $h(g+1,w) = w$ for the proof.}
Furthermore, a maximum partition for $G$ is $P^+ = (V_1 \cup V_2, \{v_1,v_2\} )$  and $f(P^+) = 2bn'^2 \geq R$. \\

b) Case $\{-a,b\} \subseteq \mathcal{W}$.

Let $n' \geq \frac{R}{2b(b+a)}$ be an integer.
Let $d = 2n'b$, and $n = 2n'(b+a)+1$.
First observe that $n > d$.
We claim that there exists a $d$-regular graph $G$ with $n$ vertices.
Indeed, let us choose $G$ such that the vertices in $G$ are labeled by the integers of $\mathbb{Z}_n$; for all $i,j \in \mathbb{Z}_n$, there is an edge $ij$ in $G$ if, and only if, $j \equiv i - d/2 + \alpha$ (mod $n$), for every $0 \leq \alpha \leq d$ such that $\alpha \neq d/2$.

We now define the graph $G' = (V,w)$ as follows.
The vertices in $G'$ are the vertices in $G$.
For all $i,j \in \mathbb{Z}_n$, $i \neq j$, $w_{ij} = -a$ if $ij$ is an edge of $G$, and $w_{ij} = b$ otherwise.

We have that $P = (V)$ is a $1$-stable partition for $G'$.
Furthermore, for all $i \in \mathbb{Z}_n$, $f_i(P) = -da + (n-d-1)b = 0$.

On the other hand, $P' = (\{0,1\}, \{2,3\}, \ldots, \{n-3,n-2\}, \{n-1\})$ is a partition for $G'$ such that $f(P') = b(n-1) \geq R$.
\end{proof}

%

\begin{proof}~[Lemma~\ref{greedy:upper-bound-gal}]
Let $G=(V,w)$ be any graph, with $m_+$ denoting the number of edges with positive weights.
We choose $G$ as an input for the greedy algorithm.
It takes $s \geq 0$ steps for computing a $1$-stable configuration $C_s$ with $q$ channels for $G$.
Note that $f(C_s) \geq s$ because the global utility strictly increases at each step (Theorem~\ref{thm:stability:k=1}). 

Let us denote $V_s$ the set of nodes $u$ such that there is at least one of their $q$ groups in $C_s$ that is not a singleton group.
We prove that $n_s \leq \ 2s$, where $n_s = |V_s|$.
Indeed, at each step of the algorithm, any vertex $v$ has at most one of its singleton groups that has been modified into another, larger group.
Furthermore, any such step modifies the singleton groups of at most two distinct vertices. 

As a consequence, there are at most $nn_s \leq 2sn$ edges $uv$ such that $w_{uv} > 0$ and either $u \in V_s$ or $v \in V_s$.
Moreover, note that there does not exist any couple of vertices $u,v \in V \setminus V_s$, such that $w_{uv} > 0$, because $C_s$ is $1$-stable. 
Hence, $2sn \geq m_+$, and so $f(C^+)/f(C_s) \leq 2h(q,w_p)m_+/s = O(h(q,w_p)n)$.
\end{proof}

%
%

\begin{proof}~[Lemma~\ref{lemma:positive-impact}]
Notations are the same as for the previous proof of Lemma \ref{greedy:upper-bound-gal}.
We set $m_+ = m_{1,+} + m_{0,+}$, where $m_{1,+}$ is the number of edges $uv$ such that $w_{uv} > 0$ and $X(u) \cap X(v) \neq \emptyset$, and $m_{0,+}$ is the number of edges $u'v'$ such that $w_{u'v'} > 0$ whereas  $X(u') \cap X(v') = \emptyset$, respectively.
By induction on the number of steps of the greedy algorithm, we get that for every two vertices $u,v$ such that $w_{uv} = -\infty$, $X(u) \cap X(v) = \emptyset$, because $u$ would avoid at any cost to interact with $v$, and reciprocally.
As a consequence, we have $f(C_s) \geq \max (s,m_{1,+})$.
Furthermore, we denote $V_q$ the set of vertices such that none of their $q$ groups in $C_s$ is a singleton group, and we get $|V_q| \leq 2s/q$.
As $C_s$ is $1$-stable by the hypothesis, there does not exist any couple of vertices $u,v \in V \setminus V_q$, such that $w_{uv} > 0$ whereas $X(u) \cap X(v) = \emptyset$ and so, $m_{0,+} \leq 2sn/q \leq 2nf(C_s)/q$.

Finally, $\frac{f(C^+)}{f(C_s)} \leq 2h(q,w_p)\frac{m_+}{f(C_s)} = 2h(q,w_p)\frac{ m_{1,+} + m_{0,+}}{f(C_s)} \leq  2h(q,w_p)(1 + 2\frac{n}{q})$.
\end{proof}

%
%
%
%
%

\begin{proof}~[Theorem~\ref{poa:upper-bound:kgeq2}]
Again, let $m_+$ be the number of edges $uv \in E$ such that $w_{uv} > 0$.
We have that $f(P^+) \leq 2m_+w_p$.
Moreover, for every edge  $uv$ of $E$ such that $w_{uv} > 0$, for any $k$-stable partition $P_k$ for $G$, we claim either $f_{u}(P_k) > 0$ or $f_{v}(P_k) > 0$.
Indeed, suppose $f_{u}(P_k) = f_{v}(P_k) = 0$. 
Then $\{u,v\}$ breaks $P_k$ by founding a new group, which is a contradiction.
Thus, there are at least $m_+/\Delta_+$ vertices $v$ of $V$ such that  $f_{v}(P_k) > 0$, as the same node cannot be counted more than $\Delta_+$ times.
Consequently, $f(P^+)/f(P_k^-) = O(\Delta_+) = O(n)$.
\end{proof}

%
%
%
%

\begin{proof}~[Lemma~\ref{poa:lower-bound-1-2}]
a) Case $\mathcal{W} \supset \{-a,b\}$\\

\noindent
Let $n' = (k+1)^2 \lfloor \frac{n}{(k+1)^2} \rfloor \leq n$.
We build $G = (V,w)$ such that $|V| = n'$, as follows.
The vertices in $V$ are labeled $v_{i,j}$, with $1 \leq i \leq k+1$, and with $1 \leq j \leq n'/(k+1)$.
For every $1 \leq i \leq k+1$, the \emph{row} $L_i$ is the set $\{ v_{i,1}, v_{i,2}, \ldots, v_{i,\frac{n'}{k+1}}\}$.
In the same way, for every $1 \leq j \leq \frac{n'}{k+1}$, the \emph{column} $C_j$ is the set $\{v_{1,j}, \ldots, v_{k+1,j}\}$.
Furthermore, vertices $v_{i,j}$ and $v_{l,t}$ are friends, that is $w_{v_{i,j}v_{l,t}} = b > 0$, if, and only if, either $i = l$ or $j = t$. 
Otherwise, $w_{v_{i,j}v_{l,t}} = -a \leq 0$.\\

\noindent
Let us consider $P = ( L_1, L_2, \ldots, L_{k+1} )$ for $G$.
We have $f(P) =   n'b(\frac{n'}{k+1}-1)$.
On the other hand, we present $P' = ( C_1, C_2, \ldots, C_{\frac{n'}{k+1}} )$, and we claim that $P'$ is a $k$-stable partition for $G$.
Indeed, observe that for any $1 \leq j \leq \frac{n'}{k+1}$, any vertex $v_{i,j} \in C_j$ has one, and only one friend, in any other column. Hence, if $v_{i,j}$ were part of a subset $S$ that breaks $P'$, then after the $k$-deviation, $v_{i,j}$ would have, at most, $k$ friends in its group, which is already the case in $C_j$.
As a consequence, there is no vertex that can strictly increase its utility by a $k$-deviation, and so, the claim is proved.
However, $f(P') = n'bk$.
As a result, $\frac{f(P^+)}{f(P_k^-)} \geq  \frac{f(P)}{f(P')} = \frac{\frac{n'}{k+1}-1}{k} = \frac{ (k+1)\lfloor \frac{n}{(k+1)^2} \rfloor - 1}{k}$.\\

 b) Case $\mathcal{W} \supset \{b,b'\}$\\

\noindent
This time, let $n' = kb \lfloor \frac{n}{kb} \rfloor$.
We build $G = (V,w)$ such that $|V| = n'$, as follows.
The set of vertices is partitioned  into $\frac{n'}{kb}$ parts, denoted $V_1, V_2, \ldots, V_{\frac{n'}{kb}}$, of equal size $kb$. 
For every couple of vertices $u,v$, $w_{uv} = b$ if $u$ and $v$ are in the same part $V_i$, with $1 \leq i \leq \frac{n'}{kb}$, and $w_{uv} = b'$ otherwise.

Obviously, a maximum (stable) partition for $G$ is $P = ( V )$, and $f(P) = n'[b(kb-1)+b'(n' - kb)]$.

Furthermore, we claim that $P' = ( V_1, V_2, \ldots, V_{\frac{n'}{kb}} )$ is a $k$-stable partition too. 
Indeed, any vertex $v$ that might break $P'$, by joining a coalition $S$ whose size is greater or equal to $k$, would have its individual utility after the deviation that is at most $(k-1)b + kbb' = kb(1+b') - b \leq kb^2 - b = (kb-1)b$. 
Hence $P'$ is $k$-stable, and $f(P') = n'b(kb-1)$.

Finally, $\frac{f(P)}{f(P')} = 1 + \frac{b'}{b}\frac{n' - kb}{kb-1} = 1 + \frac{b'}{b}\frac{kb(\lfloor \frac{n}{kb} \rfloor - 1)}{kb-1} $.  
\end{proof}

\end{document}